\documentclass[superscriptaddress,nofootinbib,notitlepage,aps,pra,10pt]{revtex4-2}
\usepackage{graphicx}
\usepackage{amsthm}
\usepackage{thmtools}
\usepackage{mathtools}
\usepackage{amsmath}
\usepackage{amssymb}
\usepackage{bm}
\usepackage[table,x11names]{xcolor}
\usepackage[normalem]{ulem}
\usepackage[colorlinks]{hyperref}
\hypersetup{
	pdfstartview={FitH},
	pdfnewwindow=true,
	colorlinks=true,
	linkcolor=blue,
	citecolor=blue,
	filecolor=blue,
	urlcolor=blue}
\usepackage{physics}
\usepackage[capitalise]{cleveref}
\usepackage{float}

\newcommand{\thm}[1]{\hyperref[thm:#1]{Theorem~\ref*{thm:#1}}}
\newcommand{\defn}[1]{\hyperref[defn:#1]{Definition~\ref*{defn:#1}}}
\newcommand{\lem}[1]{\hyperref[lem:#1]{Lemma~\ref*{lem:#1}}}
\newcommand{\fig}[1]{\hyperref[fig:#1]{Figure~\ref*{fig:#1}}}
\newcommand{\tab}[1]{\hyperref[tab:#1]{Table~\ref*{tab:#1}}}
\renewcommand{\sec}[1]{\hyperref[sec:#1]{Section~\ref*{sec:#1}}}
\newcommand{\app}[1]{\hyperref[app:#1]{Appendix~\ref*{app:#1}}}
\newcommand{\nn}{\nonumber \\}
\newcommand{\append}[1]{\hyperref[append:#1]{Appendix~\ref*{append:#1}}}
\newcommand{\con}{\chi}
\newcommand{\kernel}{\Sigma}
\newcommand{\processor}{P}
\definecolor{green2}{rgb}{0.75, 1, 0.75}
\definecolor{green1}{rgb}{0.2, 1, 0.2}
\definecolor{blue2}{rgb}{0.8, 1, 1}
\definecolor{blue1}{rgb}{0.5, 0.7, 1}

\newtheorem{theorem}{Theorem}
\newtheorem{definition}[theorem]{Definition}
\newtheorem{lemma}[theorem]{Lemma}
\newtheorem{corollary}[theorem]{Corollary}

\allowdisplaybreaks

\newcommand{\MQ}{\affiliation{
School of Mathematical and Physical Sciences,
Macquarie University, Sydney, NSW 2109, AU} }

\newcommand{\UTS}{\affiliation{
Centre for Quantum Software and Information,
University of Technology Sydney, Sydney, NSW 2007, AU}}

\newcommand{\UTSC}{\affiliation{
Centre for Quantum Computation and Communication Technology,
University of Technology Sydney, Sydney, NSW 2007, AU}}

\newcommand{\tfNSW}{\affiliation{
Quantum for New South Wales, Sydney, NSW 2000, AU}}

\newcommand{\UM}{\affiliation{
Joint Center for Quantum Information and Computer Science (QuICS), University of Maryland, College Park, Maryland 20742, USA}}

\newcommand{\FAU}{\affiliation{
Department Physik, Friedrich-Alexander-Universit\"at Erlangen-N\"urnberg, Staudtstraße 7, 91058 Erlangen, Germany}}

\begin{document}

\title{Selection and improvement of product formulae for best performance of quantum simulation}
\date{\today}
\author{Mauro E.~S.~Morales}\UTSC\UTS\UM
\author{Pedro C.~S.~Costa}\MQ\tfNSW
\author{Giacomo Pantaleoni}\MQ
\author{Daniel K.~Burgarth}\MQ \FAU
\author{Yuval R.~Sanders}\UTS \MQ 
\author{Dominic W.~Berry}\MQ

\begin{abstract}
Quantum algorithms for simulation of Hamiltonian evolution are often based on product formulae.
The fractal methods give a systematic way to find arbitrarily high-order product formulae, but result in a large number of exponentials.
On the other hand, product formulae with fewer exponentials can be found by numerical solution of simultaneous nonlinear equations.
It is also possible to reduce the cost of long-time simulations by processing, where a kernel is repeated and a processor need only be applied at the beginning and end of the simulation.
In this work, we found thousands of new product formulae, and numerically tested these formulae, together with many formulae from prior literature.
We provide methods to fairly compare product formulae of different lengths and different orders.
For the case of 8th order, we have found new product formulae with exceptional performance, about two orders of magnitude better accuracy than prior work, both in the processed and non-processed cases.
The processed product formula provides the best performance due to being shorter than the non-processed product formula.
It outperforms all other tested product formulae over a range of many orders of magnitude in system parameters $T$ (time) and $\epsilon$ (allowable error).
That includes reasonable combinations of parameters to be used in quantum algorithms, where the size of the simulation is large enough to be classically intractable, but not so large it takes an impractically long time on a quantum computer.
\end{abstract}

\maketitle
\tableofcontents

\section{Introduction}
The Lie-Trotter product formula is commonly used in quantum algorithms for Hamiltonian simulation, where one can approximate the Hamiltonian evolution of $H=A+B$ in terms of exponentials of $A$ and $B$ when these operators do not commute.
For short time, a standard approximation is the second-order symmetric formula $S_2(t)=e^{-iAt/2}e^{-iBt}e^{-iAt/2}$, which satisfies $e^{-iHt}=S_2(t)+\mathcal{O}(t^3)$.
More generally, the error in an order $k$ formula is $\mathcal{O}(t^{k+1})$.
Longer times are simulated by breaking the evolution into many repetitions of a short time.
The number of repetitions needed is reduced with the order, motivating the search for higher-order product formulae.
A systematic method to produce arbitrarily high order formulae is the fractal method \cite{Creutz1989,Suzuki1990,Yoshida1990,Suzuki5}, which has found applications in Hamiltonian simulation \cite{Berry2007}.
The first explicit use of product formulae for quantum simulation was given in \cite{Lloyd1996universal}, applying it for quantum evolution under local Hamiltonians. Later work considered the broader class of sparse Hamiltonians \cite{Aharonov2003adiabatic} and higher orders \cite{Berry2007}, as well as methods beyond product formulae \cite{Berry2012blackbox,Berry2014exponential,Berry2015nearlyoptimal,Low2019spectralnorm}.

Recent work has shown that despite its simplicity, the Lie-Trotter product formula can compete with other Hamiltonian simulation algorithms due to the low error that it achieves in practice \cite{Childs2022trotter}.
Methods based on linear combinations of unitaries \cite{Childs2012LCU,Berry2015} or quantum signal processing \cite{Low2017} give complexity logarithmic in the inverse error, but the error is not required to be extremely small, meaning those methods do not provide as large an advantage as might be expected.
Product formula error bounds can be further improved by considering particular physical systems \cite{Babbush2015chemical,Childs2018speedup,Su2021nearlytight} or leveraging randomisation \cite{Zhang2012randomized,Campbell2019random,Childs2019fasterquantum}.
Moreover, Trotter formulae are expected to be relevant for both noisy intermediate-scale quantum (NISQ) computation and fault-tolerant computation.
It is then of great importance to seek efficient implementations of product formulae as it can have a great impact on the efficiency of Hamiltonian simulation algorithms in practice.

The downside of the fractal methods to generate higher-order formulae is that the number of exponential operators to implement it grows very rapidly.
Fractal product formulae are usually assumed in quantum computing, but they can be greatly improved upon.
An alternative method by Yoshida \cite{Yoshida1990} can be used to obtain product formulae with a smaller number of exponentials.
Similar to fractal formulae, they are given as a product of $S_2$ for different time intervals, but in contrast to the fractal approach there is not an explicit analytic form for the higher-order formulae.
Instead one needs to derive and solve a complicated set of simultaneous nonlinear polynomial equations.

In classical computing, product formulae are equivalent to symplectic integrators for numerical solution of differential equations.
Since the initial work of Yoshida, there have been many results on better integrators of many different orders \cite{Calvo93,Suzuki93,McLachlan95,Kahan1997integrators,Tsitouras1999,McLachlan2002,Hairer2006,Sofroniou2005integrators,Blanes2008summary,Blanes2006processing,Blanes2013highorder}.
See Ref.~\cite{Blanes_Casas_Murua_2024} for a comprehensive review of product formulae / integrators in the literature.
In particular, Refs.~\cite{McLachlan95} and \cite{Kahan1997integrators} gives what appear to be the optimal 8th order product formulae with 15 and 17 stages (the number of $S_2$ in the product formula).
Reference~\cite{Sofroniou2005integrators} gives improved 8th order product formulae with 19 and 21 stages.
For 10th order, Ref.~\cite{Sofroniou2005integrators} provides solutions with 31, 33, and 35 stages.
These solutions from 2005 in Ref.~\cite{Sofroniou2005integrators} appear to be the best in prior work for 8th and 10th order (see Ref.~\cite{Blanes_Casas_Murua_2024}).

Another approach is that of processed product formulae \cite{Lopez96, butcher1996number, mclachlan1996more, wisdom1996symplectic, blanes1999symplectic, blanes2001high,  Blanes2006processing, blanes2024families}.
Instead of the product formula being symmetric, it is of the form $P\Sigma P^{-1}$ for \emph{kernel} $\Sigma$ and \emph{processor} $P$.
The $P$ and $P^{-1}$ cancel when using the product of many of these for evolution over a long time, so the cost is dominated by the number of stages in the kernel.
Since the kernel has fewer restrictions on it, it can have fewer stages than a normal symmetric product formula.
According to Ref.~\cite{Blanes_Casas_Murua_2024}, the best 8th order product formula with processing is from 2006 \cite{Blanes2006processing}.

The objective of this work is to find improved product formulae and compare their performance in a consistent manner.
We show that the performance of product formulae is better quantified by the error in the eigenvalues, rather than the spectral-norm error as usually considered in prior work.
This is because it is the eigenvalue error that dominates the error for evolution at longer times.
We also derive a method to consistently compare the performance of product formulae with different orders and numbers of exponentials.

We use these methods to compare between the performance of our product formulae, as well as to compare to prior product formulae in the literature.
By our numerical search, we found hundreds of thousands of solutions at 8th order, including many that outperform those previously reported in the literature.
Our best 8th order product formula without processing improves over the best prior work by a factor of about 100, and our best 8th order \emph{processed} product formula improves over the best prior processed product formulae by a factor of nearly 300.
Our best product formula without processing has slightly lower error, but our product formula \emph{with} processing provides better performance due to being shorter.
For 10th order, one solution from Ref.~\cite{Sofroniou2005integrators} provides extremely high accuracy, a factor of nearly 300 times better than any others in prior work.
We provide a detailed numerical comparison of product formulae in \cref{sec:comparison-formulae}.

When comparing product formulae of different orders, it is better to use higher-order formulae for larger values of $T/\epsilon$, where $T$ is the total evolution time and $\epsilon$ is the required precision.
For smaller $T/\epsilon$ it is better to use lower-order formulae, but as it is increased there are threshold values beyond which it is optimal to use higher-order formulae.
We derive methods for determining these thresholds, and show that our 8th order product formula is best for $T/\epsilon$ from around $10^7$ to $10^{16}$, which includes the range of typical values for quantum chemistry applications.
This means that the best 8th order product formula we have found in this work will be best suited to real applications.

The upper threshold means $T/\epsilon\gtrsim 10^{16}$ would be needed for 10th order product formulae to outperform our 8th order product formula, which is unrealistically large.
Although the 10th order solution from Ref.~\cite{Sofroniou2005integrators} provides high accuracy, it is not competitive for quantum simulation due to its greater length.
It is more than twice as long as our processed 8th order formula, so for a similar number of exponentials the 8th order product formula can use a time step half as long, resulting in much better performance for realistic values of the parameters.
The lower threshold means $T/\epsilon\lesssim 10^{7}$ is needed for 6th order to outperform 8th order.
Accurate estimation of this threshold requires testing of the product formulae for larger time steps.
It is possible to reduce this threshold by adjusting the 8th order product formula to provide better performance for larger time steps.

The organisation of this work is as follows. In Section \ref{sec:background} we give a pedagogical explanation of the product formulae and the method of solution.
First, we define product formulae and introduce the fractal methods for generating higher-order product formulae in Section \ref{sec:background/product_formulae}, then in Section \ref{sec:yoshidas-method} we give a summary of Yoshida's method,
then we explain processors for product formulae in \cref{sec:other-high},
more general types of product formulae in \cref{sec:general},
and the Taylor series method in \cref{sec:Taylor}. In Sections \ref{sec:optimisation-8} and \ref{sec:optimisation-10} we present the results regarding new product formulae.
In Section \ref{sec:comparison-method} we explain the method for comparing product formulae, then in Section \ref{sec:comparison-formulae} we give the comparison of the product formulae based on numerical testing.
We finish in Section \ref{sec:discuss} by discussing what the results mean for choosing appropriate product formulae for quantum simulation, as well as searching for better-performing product formulae.

\section{Background}
\label{sec:background}

In this section, we review the background and methods for product formulae.
Readers who are familiar with the methods may wish to skip to Section \ref{sec:search} where we present new results.
See also Ref.~\cite{Blanes_Casas_Murua_2024} for a more complete recent review of this topic.
We begin by defining product formulae and the Baker-Campbell-Hausdorff formula, then we introduce fractal methods and Yoshida's method to obtain higher-order formulae, and describe the processed product formulae and the Taylor series method.

\subsection{Product formulae}
\label{sec:background/product_formulae}

It is well known that, for any non-commuting operators $X$ and $Y$,
\begin{equation}\label{eq:order1-prodformula}
    \exp((X + Y) t) = \exp(Xt) \exp(Yt) + \order{t^2}.
\end{equation}
where we have absorbed the $-i$ factor used in quantum simulation into $X$ and $Y$. The above equation demonstrates that the exponential of a sum of two operators is, to first order, equal to the product of the exponential of those operators.
The above equation is often referred as a `first-order product formula'.
Higher-order terms can be computed via the Baker-Campbell-Haussdorff (BCH) formula.

\begin{theorem}[Baker-Campbell-Haussdorff formula \cite{BLANES2004bch}]\label{thm:BCH}
Let $X$ and $Y$ be any operators such that $\norm{X}+\norm{Y}<\ln{2}$. We have for an operator $Z$ that $\exp(X)\exp(Y) = \exp(Z)$, with
\begin{equation}
\label{eq:BCH_formula}
    Z = \sum_{n=1}^{\infty} \frac{(-1)^{n-1}}{n} \sum_{\substack{r_1+s_1>0\\ \vdots\\ r_n +s_n>0}} \frac{[X^{r_1},Y^{s_1},\cdots X^{r_n},Y^{s_n}]}{\left(\sum_{j=1}^{n}r_i+s_i\right)\prod_{i=1}^{n}r_i! s_i!},
\end{equation}
where 
$$[X^{r_{1}},Y^{s_{1}},\dotsm X^{r_{n}},Y^{s_{n}}]=[\underbrace {X,[X,\dotsm [X} _{r_{1}},[\underbrace {Y,[Y,\dotsm [Y} _{s_{1}},\,\dotsm \,[\underbrace {X,[X,\dotsm [X} _{r_{n}},[\underbrace {Y,[Y,\dotsm Y} _{s_{n}}]]\dotsm ]].$$
\end{theorem}

The standard second-order symmetric product formula is as given in the definition below.

\begin{definition}[Symmetric product formula of order two]\label{def:symmetric-order2}
Let $X$ and $Y$ be non-commuting operators and let $t$ be a real variable. Define
\begin{equation}
\label{eq:sym_prod_formula_order_2}
S_2 (t) := \exp (\frac{1}{2}Xt) \exp (Yt) \exp (\frac{1}{2}Xt).
\end{equation}
\end{definition}

The operators $X$ and $Y$ used in the definition of $S_2$ should always be clear from context.
More generally, when there is a sum of $X_j$, the product formula is
\begin{equation}\label{eq:S2def}
S_2 (t) := \left[ \prod_{j=1}^J \exp (\frac{1}{2}X_jt)\right] \left[ \prod_{j=J}^1 \exp (\frac{1}{2}X_jt) \right] .
\end{equation}
Products are ordered with the starting index on the right and the final one on the left, so for $J=2$ terms the expression in the definition is obtained.
This form of the product formula can be used for Hamiltonians that are a product of any number of terms.
The higher-order product formulae considered in this work are constructed from products of $S_2$, and due to this method of construction these product formulae are suitable for Hamiltonians that are a sum of any number of terms.
The corollary below describes the form of the error terms in the symmetric product formula, and also implies that it is second order.

\begin{corollary}[Symmetric BCH formula \cite{Yoshida1990}]\label{cor:symmetric_BCH}
Let $X$ and $Y$ be any operators such that $\norm{X}+\norm{Y}<\ln{2}$ and let $t$ be a real variable. Define $Z$ such that $S_2 (t) = \exp (Z)$.
Then there is a sequence $\alpha_\ell$ consisting of linear combinations of $\ell$-term commutators in $X$ and $Y$ such that
\begin{equation}
  Z = \sum_{\ell=1}^\infty \alpha_\ell t^\ell.
\end{equation}
Moreover, $\alpha_\ell \equiv 0$ whenever $\ell$ is even.
\end{corollary}

Reference \cite{Yoshida1990} also shows that even terms are zero for more general symmetric product formulae.
  The first three non-zero $\alpha_\ell$ terms from above are
  \begin{align}
      \alpha_1 &= X + Y, \\
      \alpha_3 &= \frac{1}{12} [Y, [Y, X]] - \frac{1}{24} [X, [X, Y]], \\
      \alpha_5 &=  \frac{7}{5760}[X,X,X,X,Y]-\frac{1}{720}[Y,Y,Y,Y,X] + \frac{1}{360}[X,Y,Y,Y,X]+\frac{1}{360}[Y,X,X,X,Y]- \frac{1}{480}[X,X,Y,Y,X] \nonumber \\
      & \quad + \frac{1}{120}[Y,Y,X,X,Y].
  \end{align}
  Here the square brackets are used to indicate multicommutator expressions similar to the notation in Theorem \ref{thm:BCH}, for example
  \begin{equation}\label{eq:nested_commutator}
      [Y,Y,X,X,Y]\equiv [Y,[Y,[X,[X,Y]]]].
  \end{equation}
Expressions for each $\alpha_\ell$ can be derived from two applications of the usual BCH formula.

\begin{definition}[Product formula]
Let $X$ and $Y$ be any non-commuting operators.
Given a natural number $n$, a \emph{product formula of order $n$} is a pair
$(\vec c, \vec d)$ with $\vec c, \vec d \in \mathbb{R}^\ell$ and $\ell$ a natural number such that for all $t\in\mathbb{R}$
\begin{equation}\label{eq:general-prod}
    \exp ((X + Y)t) = \prod_{j=1}^\ell \exp(c_j X t) \exp(d_j Y t) + \order{t^{n+1}}.
\end{equation}
We refer to the number of non-zero coefficients in $\left(\vec c, \vec d\right)$ as the length of the product formula. 
\end{definition}

Hence $S_2$ is a length-$3$ product formula.
The number of exponentials used in the product formula is a crucial measure of its efficiency, and in quantum simulation it is proportional to the required number of gates.

\paragraph{Fractal methods for generating higher-order product formulae.}
Here we describe two fractal methods to obtain higher-order product formulae.
Starting from the symmetrised product formula in Eq.~\eqref{eq:sym_prod_formula_order_2}, fractal methods produce product formulae at all even orders.
The first fractal (or triple jump) technique to generate a product formula of order $k=2\kappa$ is \cite{Creutz1989,FOREST1990105,CAMPOSTRINI1990753,Suzuki1990,Yoshida1990}
\begin{equation}
\label{eq:Suzuki_symmetric_1}
    S_{2\kappa}(t)=S_{2\kappa-2}(s_\kappa t)S_{2\kappa-2}((1-2s_\kappa)t)S_{2\kappa-2}(s_\kappa t),
\end{equation}
where $s_\kappa=1/(2-2^{1/(2\kappa-1)})$.
This method can be used to generate even orders starting at $S_2$, or more generally to increase the order of symmetric product formulae.
A drawback to this method is that both $s_k$ and $1-2s_\kappa$ are greater than 1, so the coefficients in the higher-order methods are large, causing greater error.

The history of this approach is discussed in detail in a review article by Forest \cite{Forest_2006}.
This form for a fourth-order integrator was originally derived by Ruth about 1984, but not published until 1990 \cite{FOREST1990105}.
That work specifically considered a Hamiltonian that is a sum of two parts.
The iterative form was given by Creutz and Gocksch in 1989 \cite{Creutz1989}, but citing unpublished work by Campostrini and Rossi for the fourth-order case.
That work was published in 1990 \cite{CAMPOSTRINI1990753}, and cited Creutz and Gocksch for the iterative form.
Suzuki \cite{Suzuki1990} and Yoshida \cite{Yoshida1990} independently derived the iterative form in 1990.

The iterative form given by Creutz and Gocksch in 1989 \cite{Creutz1989} was the more general expression
\begin{equation}
    \Tilde{S}_{2\kappa-2}(u_\kappa t)^p\Tilde{S}_{2\kappa-2}((1-2p u_\kappa)t)\Tilde{S}_{2\kappa-2}(u_\kappa t)^p \, .
\end{equation}
A case of particular interest is the second fractal method, also derived by Suzuki in 1991 \cite{Suzuki5}
\begin{equation}
\label{eq:Suzuki_symmetric_2}
    \Tilde{S}_{2\kappa}(t)=\Tilde{S}_{2\kappa-2}(u_\kappa t)^2\Tilde{S}_{2\kappa-2}((1-4u_\kappa)t)\Tilde{S}_{2\kappa-2}(u_\kappa t)^2,
\end{equation}
where $u_\kappa=1/(4-4^{1/(2\kappa-1)})$.
This method has the advantage that both $u_\kappa$ and $1-4u_\kappa$ are less than 1, so the coefficients of higher-order formulae are small resulting in small error.
The drawback is that far more exponentials are required.
Each iteration uses 5 copies of the lower-order formula, whereas the previous one uses 3 copies.
The virtue of these fractal methods is that they allow one to generate arbitrarily high-order product formulae easily, albeit at the expense of a large number of exponentials.

\paragraph{Exponential length scaling of fractal methods.} For the first fractal method, the total number of exponentials for a given order $2\kappa={4,6,8,...}$ in the product formula $S_{2\kappa}$ is given by
\begin{equation}
    2(J-1) 3^{\kappa-1}+1 \, .
\end{equation}
For example $J=2$ for $X+Y$ and $\kappa=1$ just corresponds to $S_2$, and the expression gives 3 as expected.
For the second fractal method $\Tilde{S}_{2\kappa}$ the number of exponentials is
\begin{equation}
    2(J-1) 5^{\kappa-1}+1 .
\end{equation}
The number of exponential operators in both cases of the fractal method grows very rapidly.
Thus one may be interested in alternative method to obtain product formulae with a lower count, such as the method of Yoshida in the next section.

\subsection{Yoshida's method for deriving 6th order product formulae}
\label{sec:yoshidas-method}

Here, we give a pedagogical explanation of the method to derive product formulae.

\paragraph{Approach.}
Rather than using \cref{eq:Suzuki_symmetric_1,eq:Suzuki_symmetric_2},
Yoshida uses the general procedure \cite{Yoshida1990}
\begin{align}\label{eq:Yoshida_ansatz}
    S^{(m)}(t) = \bigg(\prod_{j=1}^m S_2(w_{m-j+1}t) \bigg) S_2(w_0 t) \bigg(\prod_{j=1}^m S_2(w_j t)\bigg), 
\end{align}
where $w_j\in \mathbb{R}$ for $j=0,1,\cdots,m$ are parameters to be determined. Note the number of exponentials in this product is given by
\begin{equation}\label{eq:noexp}
    (4m+2)(J-1)+1 .
\end{equation}
Given this ansatz, the task becomes to find $m$ and $w_i$ such that $S^{(m)}$ is an order $k$ product formula.
We will illustrate Yoshida's method by deriving the result for 6th order.

\paragraph{Expand Yoshida product using Baker-Campbell-Haussdorf formula.}
The method is to expand \cref{eq:Yoshida_ansatz} using the BCH formula from \cref{thm:BCH}
recursively as follows. First, note that from \cref{cor:symmetric_BCH}, $S_2(t)=e^{\frac{t}{2}X}e^{tY}e^{\frac{t}{2}X}=e^{t \alpha_1+t^3 \alpha_3\cdots}$, where $\alpha_\ell$ is a commutator of $\ell$ operators as described below \cref{cor:symmetric_BCH}. 
We are for the moment considering sums of two terms $X+Y$. 
Define $C=\sum_{i=1}^\infty t^{2i-1}w_1^{2i-1}\alpha_{2i-1}$ and $D=\sum_{i=1}^\infty t^{2i-1}w_0^{2i-1}\alpha_{2i-1}$. Then,
\begin{align}\label{eq:Yoshida_S2_6}
      &S_2 (w_1 t) S_2 (w_0 t) S_2 (w_1 t) \nn &=e^{C}e^{D}e^{C} \nn
      &= \exp \bigg\{ t w_1 \alpha_1 + t^3 w_1^3 \alpha_3 +t^5 w_1^5 \alpha_5 + \order{t^7} \bigg\} \exp\bigg\{t w_0 \alpha_1 + t^3 w_0^3 \alpha_3 + t w_0^5 \alpha_5 + \order{t^7} \bigg\} \nn
    &\quad\times \exp \bigg\{ t w_1 \alpha_1 + t^3 w_1^3 \alpha_3 +t^5 w_1^5 \alpha_5 + \order{t^7} \bigg\} \nonumber\\
        &= \exp\bigg\{t(2w_1 + w_0)\alpha_1 + t^3 (2w_1^3 + w_0^3)\alpha_3 + t^5 (2w_1^5 + w_0^5)\alpha_5 
    + \frac{1}{6}([D,D,C]-[C,C,D]) + \order{t^7}\bigg\} .
\end{align}
A simple computation shows 
\begin{align}
    [D,D,C]-[C,C,D] &= t^5 (w_0^2 w_1^3 -w_1^2 w_0^3 + w_1^4 w_0 - w_0^4 w_1)[\alpha_1,\alpha_1,\alpha_3] + \order{t^7}.
\end{align}
Define $\beta_5 = [\alpha_1,\alpha_1,\alpha_3]$ so
\begin{align}\label{eq:6-order-m}
    S_2(w_1 t)S_2(w_0 \tau)S_2(w_1 t) &= \exp\bigg\{t(2w_1 + w_0)\alpha_1 + t^3 (2w_1^3 + w_0^3)\alpha_3 + t^5 (2w_1^5 + w_0^5)\alpha_5 \nonumber\\
    &\quad + t^5 \frac{1}{6}(w_0^2 w_1^3 - w_1^2 w_0^3 +w_0w_1^4 - w_0^4w_1)\beta_5 + \order{t^7}\bigg\} .
\end{align}
By an induction argument Yoshida shows that
\begin{align}\label{eq:Yosh-order6}
    S^{(m)}(t) = \exp \bigg\{t A_{1,m} \alpha_1 + t^3 A_{3,m} \alpha_3 + t^5 (A_{5,m} \alpha_5 + B_{5,m}\beta_5) + \order{t^7}  \bigg\}, 
\end{align}
where $A_{j,m}$ and $B_{5,m}$ are polynomials on the variables $w_0,\dots, w_m$.

The case $m=0$ is just the symmetric BCH formula, so it is clear that Eq.~\eqref{eq:Yosh-order6} holds with
\begin{align}
    A_{1,0}&=w_{0},\nonumber\\
    A_{3,0}&=w_{0}^3,\nonumber\\
    A_{5,0}&=w_{0}^5 ,\nonumber\\
    B_{5,0}&= 0 . 
\end{align}
To prove Eq.~\eqref{eq:Yosh-order6} for $m>0$, one needs to show that the exponential is of the form with operator $\alpha_1$ for first order in $t$, operator $\alpha_3$ for third order in $t$, and operators $\alpha_5$ and $\beta_5$ for fifth order in $t$.
This result may be shown using
\begin{align}
    S^{(m+1)}(t)&=S_2(w_{m+1}t) S^{(m)}(t) S_2(w_{m+1}t)\nonumber\\
    &= \exp\bigg\{t w_{m+1}\alpha_1 +t^3 w_{m+1}^3 \alpha_3 + t^5 w_{m+1}^5 \alpha_5  + \order{t^7} \bigg\} \nonumber\\ 
    & \quad \times
    \exp\bigg\{ t A_{1,m} \alpha_1 + t^3 A_{3,m} \alpha_3 + t^5 (A_{5,m} \alpha_3+ B_{5,m} \beta_5)  + \order{t^7} \bigg\} \nonumber\\ 
    & \quad \times \exp\bigg\{t w_{m+1}\alpha_1 +t^3 w_{m+1}^3 \alpha_3 + t^5 w_{m+1}^5 \alpha_5 +\order{t^7} \bigg\}\nonumber\\
    & = \exp \bigg\{2 t w_{m+1} \alpha_1 + t A_{1,m} \alpha_1 + 2 t^3 w_{m+1}^3 \alpha_3 + t^3 A_{3,m} \alpha_3 + 2 t^5 w_{m+1}^5 \alpha_5 + t^5 A_{5,m} \alpha_5 + t^5 B_{5,m} \beta_5 \nonumber\\
    &\quad + \frac{1}{6} t^5 (A_{1,m}^2 w_{m+1}^3 - A_{1,m}A_{3,m}w_{m+1} - w_{m+1}^2 A_{3,m}+w_{m+1}^4 A_{1,m} )\beta_5 + \order{t^7} \bigg \} .
\end{align}
Hence, if the product formula can be expressed as in the form \eqref{eq:Yosh-order6} for $S^{(m)}(t)$, it can again be expressed in this form for $S^{(m+1)}(t)$, establishing it for all $m\ge 0$ by induction.
This expression also shows that the scalar coefficients can be determined from the formulae
\begin{align}
    A_{1,m+1}&=2w_{m+1} + A_{1,m},\nonumber\\
    A_{3,m+1}&=2w_{m+1}^3 + A_{3,m},\nonumber\\
    A_{5,m+1}&=2w_{m+1}^5 + A_{5,m},\nonumber\\
    B_{5,m+1}&=B_{5,m}+\frac{1}{6}(A_{1,m}^2 w_{m+1}^3 - A_{1,m}A_{3,m}w_{m+1} - w_{m+1}^2 A_{3,m}+w_{m+1}^4 A_{1,m}) .
\end{align}
See \cref{sec:yoshida-order10} for an explanation of how to extend this approach to 10th order.

\paragraph{Constraints to satisfy in order to derive 6th order formula.}
To derive a 6th order formula,
the lower-order terms in the exponential in \cref{eq:Yosh-order6} must be zero (see also Eq.~(5.16) of \cite{Yoshida1990}), which gives the four conditions
\begin{equation}
    A_{1,m} = 1, \quad A_{3,m} = 0, \quad A_{5,m} = 0, \quad B_{5,m} = 0.
\end{equation}
For $m=3$ there are four unknowns $w_0$ to $w_3$, and it can be expected there are solutions because there are the same number of equations as unknowns.
In practice $A_{1,m} = 1$ is satisfied by taking $w_0=1-2\sum_j w_j$, so there are then three equations for three unknowns $w_1,w_2,w_3$.
Whereas it is possible to solve the equations using the Newton-Raphson method,
the expression for the appropriate Jacobian matrix is complicated, so
Yoshida instead uses the Brent method.
Yoshida produced three $m=3$ solutions in this way, and states ``It seems that there is no other solution.''
We have performed an extensive search and also found no more solutions.

The product formulae obtained through the Yoshida method also work for exponentials of sums of more operators. 
The $S_2$ product formula can again be used as the building block for the product formula, and we can write 
\begin{equation}\label{eq:yoshida-multiterm}
     S_2(t)= \exp
     \left( \sum_{\ell=0}^\infty \tilde{\alpha}_\ell t^\ell
     \right),
\end{equation}
where $\tilde{\alpha}_\ell$ are now order-$\ell$ multicommutator expressions on the $J$ terms.
The reasoning for finding the product formula is entirely based on the construction with ${\alpha}_\ell$, but without taking advantage of its particular form, so exactly the same reasoning holds for $\tilde{\alpha}_\ell$.
This immediately implies that the higher-order product formulae work in general.
This is an advantage of constructing product formulae as products of $S_2$.

\subsection{Processed product formulae}\label{sec:other-high}
Another technique to obtain higher-order product formulae is that of processing \cite{Lopez96, butcher1996number, mclachlan1996more, wisdom1996symplectic, Blanes2006processing, blanes2001high, blanes1999symplectic}. In this technique a product formula $S_k$ of order $k$ is generated  by the composition of a \emph{kernel} $\Sigma$ conjugated by a \emph{processor} operator $P$ as
\begin{align}\label{eq:procecssor}
    S_k=P \Sigma P^{-1}.
\end{align}
The advantage of this method is that usually $\Sigma$ requires fewer stages than other methods, and due to the construction, we have that $S_k^n=P\Sigma^n P^{-1}$.
This means that the cost of using the product formula is effectively that of the kernel in the usual application where a long time evolution is partitioned into many repetitions of the product formula for short times.

Typically, one chooses a kernel $\Sigma$ as a product formula that has a smaller order than $k$, but conjugating by processors gives an order $k$ integrator. 
Some of the best-performing kernels are given in Ref.~\cite{Blanes2006processing}, and recently new kernels were reported in \cite{blanes2024families}.
The advantage of this method is that the number of nonlinear equations required to be solved to find new product formulae is reduced.
One could use a shorter minimum-length product formula, 
but best performance is provided by increasing the length of a product formula beyond the minimum.
The more the number of free parameters exceeds the number of equations, the more freedom there is to find product formulae with low error.
Reducing the number of equations by solving for the kernel enables reduced error for a given length.

To be more specific, Ref.~\cite{Blanes2006processing} gives a set of conditions for the kernels and processors in Table 4 of that work.
Reference~\cite{Blanes2006processing} uses the notation $Y_j$ rather than $\alpha_j$, so $S_2(t)=e^{t Y_1+t^3 Y_3\cdots}$.
It then uses the notation $E_{j,i}$ for multicommutator expressions, with $E_{j,1}=Y_j\equiv\alpha_j$ and for example
\begin{align}
    E_{5,2} &= [Y_1,Y_1,Y_3] \equiv \beta_5\, , \\
    E_{7,2} &= [Y_3,Y_1,Y_3] \equiv -\gamma_7 \, ,
\end{align}
where $\gamma_7$ is given in \cref{sec:yoshida-order10}.
In Ref.~\cite{Blanes2006processing} the quantities $f_{j,i}$ are the coefficients of $E_{j,i}$, so the equivalent of Eq.~\eqref{eq:Yosh-order6} in that notation is
\begin{align}
    S^{(m)}(t) = \exp \!\bigg\{ t f_{1,1} E_{1,1} + t^3 f_{3,1} E_{3,1} + t^5 (f_{5,1} E_{5,1} + f_{5,2}E_{5,2}) + \order{t^7}  \bigg\} \, .
\end{align}

To obtain the conditions for a 6th order \emph{kernel} from Ref.~\cite{Blanes2006processing}, one should use the conditions for $f_{j,i}$ up to $q=5$ in Table 4 of that work, which give $f_{1,1}=1$, $f_{3,1}=0$, $f_{5,1}=0$, which in the above notation are equivalent to
\begin{equation}
    A_{1,m} = 1, \quad A_{3,m} = 0, \quad A_{5,m} = 0.
\end{equation}
The conditions for the kernel are the same as for the complete product formula, but missing $B_{5,m} = 0$.
That is a general feature of the kernel order conditions; they are a subset of the conditions for the complete product formula, enabling shorter kernels, or more flexibility to choose kernels of the same length but lower error.
For 4th order there are the same number of conditions, so the kernel is a valid product formula of that order.

For the processors, Ref.~\cite{Blanes2006processing} uses the notation $p_{j,i}$ for the coefficients of $E_{j,i}$ in the expansion of $P$.
These need to be determined from the BCH formula rather than symmetric BCH, and instead of only odd-order terms there need to be only even-order terms.
Some symmetries mean that low-order odd terms automatically cancel, but higher-order ones need to be made zero by the choice of $P$.
Table 4 in Ref.~\cite{Blanes2006processing} gives the conditions on the even-order terms.
The first processor condition depending on $f_{j,i}$ is for $q=5$, which means that for 6th order and higher the processors are dependent on the kernel.
The 4th order the kernel is already 4th order, and the processors will yield another 4th order product formula, but can be chosen to reduce the error (which arises from higher-order terms).
See Appendix \ref{app:processors} for more detail on the method of determining processors.

\subsection{More general product formulae}\label{sec:general}
The form of product formulae constructed from products of $S_2$ can be generalised in a number of ways.
First, one may choose different weightings in the two parts of Eq.~\eqref{eq:S2def}, so one has
\begin{equation}
S (t_1,t_2) = \left[ \prod_{j=1}^J \exp (\frac{1}{2}X_jt_1)\right] \left[ \prod_{j=J}^1 \exp (\frac{1}{2}X_jt_2) \right] .
\end{equation}
That is, one still alternates between the order of applying the exponentials, but allows the times to be different.
This form is no longer symmetric, but there are twice as many parameters for a given length.
In the literature on symplectic integrators, the notation for product formulae of this type is \cite{Blanes_Casas_Murua_2024}
\begin{equation}
    \chi_{\alpha_{2s}h}\circ\chi^*_{\alpha_{2s-1}h}\circ \cdots \circ \chi_{\alpha_{2}h}\circ\chi^*_{\alpha_{1}h}\, ,
\end{equation}
where the asterisk is indicating that the exponentials are performed in reverse order (rather than conjugation).

In the case where the Hamiltonian is a sum of only two parts, then a product formula may be written in the form
\begin{equation}\label{eq:XY}
    e^{w_m X t} e^{w_{m-1} Y t} \cdots e^{w_1 X t} e^{w_{0} Y t} e^{w_1 X t} \cdots e^{w_{m-1} Y t} e^{w_m X t} \, ,
\end{equation}
if $m$ is odd, or ending in an exponential of $Y$ if $m$ is even.
Here we have written the product formula as symmetric, but one can also consider product formulae without that constraint.
Again there are twice as many parameters as for a product formula constructed as a product of $S_2$.
Product formulae of this type are often written as \cite{Blanes_Casas_Murua_2024}
\begin{equation}
    \varphi^{[1]}_{a_{s+1}h}\circ \varphi^{[2]}_{b_{s+1}h}\circ \varphi^{[1]}_{a_{s}h}\circ \cdots \circ \varphi^{[1]}_{a_{2}h}\circ \varphi^{[2]}_{b_{1}h}\circ \varphi^{[1]}_{a_{1}h}
    \, ,
\end{equation}
where $\varphi^{[1]}$ and $\varphi^{[2]}$ are indicating the two different exponentials.

In the case where the Hamiltonian is a sum of two terms, the product of $S(t_1,t_2)$ gives a form as in Eq.~\eqref{eq:XY}, and conversely the product formula as in Eq.~\eqref{eq:XY} can be expressed as a product of $S(t_1,t_2)$.
In Section~\ref{sec:comparison-formulae} we give results for testing some product formulae of this type, and find that they give the best results for 4th and 6th order.
This method is mostly considered only up to 6th order, because for higher orders the number of order conditions becomes significantly larger (see Table I of Ref.~\cite{Blanes2006processing}).

Another class of product formulae is those constructed from products of fourth-order product formulae $S_4$.
Some product formulae have been constructed in that way; for example see Ref.~\cite{Blanes2006processing}.
We have also tested those (using the fourth-order fractal product formulae), but do not include them in the comparison in Section~\ref{sec:comparison-formulae}.
We found that, although they reduced the error, the much larger number of exponentials in the product meant that they were not competitive.
Part of the reason is that the best gains are obtained when there are significantly more free parameters to adjust than the number of equations to satisfy.
This is the reason to use product formulae with more than the minimum length.
Constructing a product formula from $S_4$ increases the length without the added advantage of more free parameters to adjust.

Another option is to construct product formulae from a different second-order product formula.
For example, a more accurate second-order method is provided in Ref.~\cite{McLachlan95}.
In our testing, we find that yields slightly more accurate product formulae, but with twice the length, making them uncompetitive with product formulae constructed from the standard second-order symmetric product formula $S_2$.

\subsection{Solution using Taylor expansion}\label{sec:Taylor}

Another solution method is based on computing the Taylor expansion of both the exact exponential and its product formula approximation with given $w_j$.
This method is not normally used because it is much more computationally intensive than the symmetric BCH approach.
Here, we describe how to effectively code this method for numerical solution.
The Taylor expansion is performed on both sides up to the desired order of approximation for the integrator.
By imposing equality on terms of the same order one can obtain a series of equations for $w_j$ which can be solved.
For higher orders, a large number of simultaneous equations are obtained, so we need an automated way of generating them.

To make precise the problem we are solving, denote as $\mathcal{T}_k[\cdot]$ the map giving the Taylor expansion in $t$ around $0$, truncated at order $k$, so
\begin{align}\label{eq:def_Tot} 
     \mathcal{T}_k[e^{t (X+Y)}]\nonumber &= \sum_{p=0}^{k} \frac{t^p}{p!}(X+Y)^p \nonumber\\
       &=\sum_{p=0}^{k} \frac{t^p}{p!} \sum_{r_1,\cdots,r_p=0}^1 X^{r_1}Y^{1-r_1}\cdots X^{r_p}Y^{1-r_p} .
\end{align}
We consider a sum of two operators $X+Y$, but
this approach for solving for $w_j$ will also be sufficient to provide product formulae for sums of arbitrary numbers of terms.
That is because the solutions must also be solutions of the equations derived using Yoshida's method, and as explained above those equations will be the same when considering sums of arbitrary numbers of terms.
Now consider the ansatz operator of Yoshida from \cref{eq:Yoshida_ansatz} 
\begin{align}
    S^{(m)}(t,w_1,\cdots,w_m)&=e^{tw_m X/2 }e^{tw_m Y}e^{t(w_m+w_{m-1})X/2}e^{tw_{m-1}Y}e^{t(w_{m-1}+w_{m-2})X/2}\nonumber\\
    &\quad \cdots e^{tw_0 Y}e^{t(w_1+w_0)X/2}e^{tw_1 Y}\cdots e^{tw_m X/2}\nonumber\\
    &= e^{tc_1X}e^{tc_2Y}\cdots e^{tc_{4m+3}X},
\end{align}
where in the last line we have defined the constants $c_1=w_m/2$, $c_2=w_m$, $c_3=(w_m+w_{m-1})/2, \cdots$ $c_{4m+3}=w_m/2$.
Taylor expanding this ansatz up order $k$ gives
\begin{align}\label{eq:def_Prod}
     \mathcal{T}_k[S^{(m)}(t,w_1,\dots,w_m)]
    &= \sum_{\substack{r_1, \cdots, r_{4m+3} = 0 \\ r_1 + \dots + r_{4m+3} \leq k}}^k \frac{t^{r_1 +\dots + r_{4m+3}}}{r_1 ! \cdots r_{4m+3} !}c_1^{r_1}\dots c_{4m+3}^{r_{4m+3}}X^{r_1}Y^{r_2} \dots X^{r_{k-1}}Y^{r_{4m+3}}.
\end{align}
Note that the total number of exponentials in Yoshida's ansatz is $4m+3$.

The product formula obtained from the solution procedure needs to be independent of the choice of $X$ and $Y$, so we need to match the coefficients for each sequence of products of $X$ and $Y$.
Because $X$ and $Y$ are non-commuting, we need to record coefficients for each ordered sequence.
To do this in an automated way we construct a data structure to store the coefficients.

Given operators of the form $e^{cX}=I + cX + \frac{c^2}{2!}X^2 + \frac{c^3}{3!}X^3 + \cdots$ and $e^{d Y}=I + d Y + \frac{d^2}{2!}Y^2 + \frac{d^3}{3!}Y^3 + \cdots$  with $c,d$ scalar coefficients and $X, Y$ general operators, we describe this Taylor expansion up to order $k$ using an array encoding.
First, write monomials composed of $X$ and $Y$ operators in lexicographical order and note that these operators can be mapped to binary numbers:
\begin{equation}
\label{eq:lexi}
\begin{matrix}
I & X & Y & XX & XY &  YX &  YY & XXX & XXY & \cdots\\
1 & 10 & 11 & 100 & 101 &  110 &  111 & 1000 & 1001 & \cdots\\
1 & 2 & 3 & 4 & 5 &  6 &  7 & 8 & 9 & \cdots
\end{matrix}    
\end{equation}

To construct a bit string, we map each $X$ to 0 and each $Y$ to 1, then place a 1 on the left to flag the length of the string, as shown in the second line of \cref{eq:lexi}.
Then, to obtain the operator products, simply remove the leading 1 and then map $0$ to $X$ and $1$ to $Y$.
The empty string corresponds to the identity $I$.
Now, to store the coefficients in a sum of products of $X$ and $Y$, convert each product to a binary integer as above, then store the coefficient in the corresponding location in a vector.
This approach enables products of Taylor expansions of exponentials to be rapidly calculated.

By way of illustration, consider the polynomial $10 I + 3X + 2Y + 2X^2 + YX$.
This operator would be stored in an array as $[10,3,2,2,0,1,0,\cdots]$.
In this way, any polynomial of $X$ and $Y$ can be efficiently stored in a vector.
We denote this vector storing the coefficients of operators of order up to $k$ as $\mathrm{vec}_k[p(X,Y)]$, where $k$ denotes that the vector will only store the coefficients of the corresponding operators up to order $k$ (so a vector of length $2^{k+1}-1$) and $p(X,Y)$ is the polynomial in terms of $X$ and $Y$.

\section{Product formula search}\label{sec:search}

First we describe the numerical methods, then summarise solutions we found at 8th and 10th order.

\subsection{Numerical methods}

The size of the product formula we consider as in Eq.~\eqref{eq:Yoshida_ansatz} is governed by the value of $m$, so the number of $S_2$ in the product is $M=2m+1$.
The number of free parameters is $m$, and it needs to be a least a large as the number of independent equations in order for there to be a solution.
Yoshida \cite{Yoshida1990} finds that $m=7$ is minimal for 8th order, and $m=15$ is minimal for 10th order \cite{Sofroniou2005integrators}.
Choosing the minimal $m$ typically yields product formulae with large error, and for that reason we focus on larger values of $m$.

For the numerical solution of the simultaneous nonlinear equations, we found that Matlab's \texttt{fsolve} (using the Levenberg-Marquardt algorithm) provided rapid results using the vector of errors.
The method was to repeatedly solve from random starting vectors $\vec w$ generated according to the normal distribution.
The best solutions were those with smaller values for the coefficients; for much of the calculations we selected standard deviations of $0.6$ for 8th order and $0.9$ for 10th order.
We filtered the large number of solutions by numerically checking the error for example Hamiltonians, and selecting those that provided the best performance.
We then perform tests for larger numbers of samples, to more accurately select the best performing product formulae.
We also further refined the solutions to give the results to extended precision, which enables testing of the error with smaller values of $t$.
We were also able to further reduce the error by using an alternative solution method using the Taylor series.

The Taylor series approach could be used to solve for product formulae, but is more computationally intensive because it produces an error vector that is exponentially large in the order.
However, it is very effective at further improving the performance of solutions that have already been found.
If a solution is of 8th order (for example), then we can consider the error using the Taylor expansion up to and including the 9th order.
This will give nonzero error for the 8th order product formula, so minimising the error gives a new product formula.
That can be used as a starting point to again solve for an 8th order formula, and often it is found that the product formula adjusted in that way has reduced error.

\subsection{Improved 8th order}\label{sec:optimisation-8}

Following these methods, we have searched for new product formulae and found thousands of product formulae of 8th order.
Once we have obtained the potential solutions, we generate random Hamiltonians and
compute the product formula errors for a chosen short time.
The error for 8th order product formulae is given by $\delta(t)\propto t^9$, but the constant will depend on the measure of error used.
One could use the spectral-norm error for a single time step, but a more appropriate measure of the performance of the product formulae is the accuracy of the eigenvalues.
In the following we use $\con$ for the constant factor based on spectral-norm error, and $\zeta$ for the constant based on eigenvalue error, with $\con=\delta/t^9$ or $\zeta=\delta/t^9$.

For each product formula, we compute an average constant factor error; this average corresponds to the geometric mean of the constant factors computed for each random Hamiltonian.
See Section \ref{sec:comparison-formulae} for more details of the numerical methods, and the comprehensive comparison of product formulae.
This method of comparing the performance of product formulae through the estimation of the constant factor in the error has been used before (see for example \cite{Blanes2013highorder}) and is considered a good approximation to the performance of the product formula in practice.

The case $m=7$ yields product formulae of minimal length, because the number of equations is equal to the number of unknowns.
We have found over 600 solutions, and the search now finds almost only repeated solutions and very few new solutions.
This indicates that we have found nearly all solutions, but it is also possible that there are many more solutions with large values of $\vec w$.
Indeed, the most recent new solutions we found have significantly larger values of $\vec w$.
We find that large values of $\vec w$ correspond to worse product formulae with larger error.
Therefore, even if there are many more solutions with large $\vec w$, they likely will not give improved performance over those we have already found.

This search for $m=7$ (15 stages) verified that the best-performing product formula corresponds to that previously reported in Refs.~\cite{McLachlan95} and \cite{Kahan1997integrators}.
Similarly, a numerical search with $m=8$ (17 stages) yielded the best-performing product formula close to solutions s17odr8a and s17odr8b given in Ref.~\cite{Kahan1997integrators}.
It is not exactly the same, but the small difference is likely due to the system of equations being underdetermined.
The extensive nature of our search indicates that these product formulae reported in Refs.~\cite{McLachlan95,Kahan1997integrators} are optimal for 8th order product formulae with $m=7$ and $m=8$.

\begin{table}[tbh]
\centering
\begin{tabular}{ |c|c|c| } 
 \hline
 & Best 8th order for spectral-norm error & Best 8th order for eigenvalue error \\ 
 \hline
 $w_1$ & $0.59358060400850625863514059265224$ & $0.10467636532245895252340732579853$ \\ 
 $w_2$ & $-0.46916012347004197296293264921328$ & $-0.57896999331780988041471955125778$ \\
 $w_3$ & $0.2743566425898467907228242878146$ & $0.57503350160061785946141563279891$ \\
 $w_4$ & $0.17193879484656773059919074965377$ & $0.12231011868707029786561397542663$  \\
  $w_5$ & $0.23439874482541384415430578747541$ & $0.27793149999039524816733903301747$ \\
  $w_6$ & $-0.48616424480326193899617759997914$ & $-0.37349605088056728482635987352576$ \\
  $w_7$ & $0.49617367388114660354871757044906$ & $0.11575566589480463220616543972403$   \\
  $w_8$ & $-0.32660218948439130114501815323814$ & $0.1464645610975800618712569230326$   \\
  $w_9$ & $0.23271679349369857679445410270557$ & $-0.39443578322284085764474498594073$   \\
  $w_{10}$ & $0.098249557414708533273471906180643$ & $0.44370228726021218923197141183196$  \\
 \hline
\end{tabular}
\caption{Our best-performing 8th order solutions when setting $m=10$. }
\label{tab:order8-solm10}
\end{table}

To provide improved performance over prior work, we extended the search to $m=10$ (21 stages), and found over 100,000 solutions.
This size of 8th order product formula was previously studied in Ref.~\cite{Sofroniou2005integrators}, but we find solutions with significantly lower error.
Our best solutions for spectral-norm and eigenvalue error are given in \cref{tab:order8-solm10}.
The solution selected for minimum spectral-norm error has a constant of $\con=5.8\times 10^{-8}$.
Its constant factor for eigenvalue error $\zeta=7.0\times 10^{-9}$ is significantly smaller, but the best performance is obtained by selecting for minimum eigenvalue error.
The second solution given in \cref{tab:order8-solm10} has the significantly smaller constant factor in the error $\zeta=5.4\times 10^{-10}$.

These constant factors all improve over those provided by product formulae in prior work.
Our solution with smallest eigenvalue error is well over 100 times more accurate than the $m=10$ product formula in Ref.~\cite{Sofroniou2005integrators}, and about 100 times more accurate than the $m=9$ product formula in that work.
See Section \ref{sec:comparison-formulae} for the detailed numerical testing of the product formulae.

A further improvement in product formulae can be obtained by using processing.
We have performed a search for high-accuracy 8th order product formulae using this method with $m=8$, and found one particular example that provides excellent performance; see Table \ref{tab:process}.
This kernel provides a constant of $\zeta=8.1\times 10^{-10}$, only slightly larger than the best length $m=10$ product formula, while providing better performance due to the shorter length.
This error is nearly 300 times smaller than provided by 8th order processed product formulae in prior work \cite{Blanes2006processing}.
Therefore, we have provided a sequence of three solutions that all improve over prior work, with successively better and better performance.

\begin{table}[tbh]
\centering
\begin{tabular}{ |c|c|c|c| } 
 \hline
 & Best 8th order kernel & Processor for kernel & \\ 
 \hline
 $w_1$ & $0.21784176681731006074681969186513$ & $-0.44324901019570126590495430949294$ & $\gamma_1$ \\ 
 $w_2$ & $0.1947017706053903224022456342907$ & $0.25459857192003772850622377066944$ & $\gamma_2$ \\
 $w_3$ & $0.18372413281145589944261642180363$ & $-0.73862036266779261573694538099739$ & $\gamma_3$ \\
 $w_4$ & $-0.37307499512657736825709230652023$ & $-0.00024139614958652134370419495289618$ & $\gamma_4$  \\
  $w_5$ & $0.15757644257569146373033662060461$ & $0.73873460354125365739379753874964$ & $\gamma_5$ \\
  $w_6$ & $-0.33342207567391682979227850551172$ & $-0.20285971152536085519251666906017$ & $\gamma_6$ \\
  $w_7$ & $0.51788649682987924281787142226803$ & $0.44989521689676869571827637424046$ & $\gamma_7$   \\
  $w_8$ & $0.21456475499897766986381219621761$ & $0.29538398007876871184026747505657$ & $\gamma_8$   \\
  &  & $-0.3364996155865700091428329802017$ &  $\gamma_9$  \\
 \hline
\end{tabular}
\caption{Our best-performing 8th order kernel and processor when setting $m=8$. The value of $\gamma_{10}$ for the processor is chosen such that the sum is zero.}
\label{tab:process}
\end{table}

To provide an accurate result in terms of the spectral norm (rather than just the eigenvalues) it is necessary to solve for the processor.
For the kernel there are many solutions for the processor, but we list the one with the best performance of the solutions we have found.
This one was found by solving for many (hundreds of thousands) of processors, picking the one with the best performance, and further optimising.
Moreover, using the processor the constant for the spectral-norm error is $\chi=5.3\times 10^{-8}$.
That is even better than the best length $m=10$ (non-processed) formula above.

\subsection{Finding 10th order}\label{sec:optimisation-10}

We have also used our solution procedure to find new $10$th order product formulae.
For Yoshida's method applied to 10th order, there are $15$ independent equations to be solved (see \cref{sec:yoshida-order10}).
We performed searches for solutions with $m=15$ (the minimal number) to $m=18$.
Again the larger values of $m$ give the flexibility to adjust the solution to reduce the error.
As in \cref{sec:optimisation-8}, we compare the performance of product formulae of $10$th order by computing the constant factors $\con$ and $\zeta$ for random Hamiltonians.

Out of our solutions, the best eigenvalue performance is given by a $m=17$ solution (see \cref{tab:order10-solm17}), which yields $\zeta=6.1 \times 10^{-11}$.
This formula is slightly outperformed by one product formula in Ref.~\cite{Sofroniou2005integrators}, but for the purpose of quantum computing we will show that 8th order provides better performance than 10th order.
For this reason we did not perform further searches for processed 10th order product formulae or proceed to any higher orders.

In the search for 10th order product formulae, unlike in the case of 8th order, we find that almost all new solutions found are different from those found before.
That indicates there is an extremely large number of solutions, and we have only found a very small proportion of them.
There is still the potential for better-performing solutions to be found at this order, but we expect that the longer length will make them unsuitable for quantum computing.

\begin{table}[tbh]
\centering
\begin{tabular}{ |c|c|} 
 \hline
 & 10th order solution with $m=17$  \\ 
 \hline
 $w_1$ & $-0.28371232689144296279654621726493$    \\ 
 $w_2$  & $0.046779504778147381605331000278223$   \\
 $w_3$ & $0.36845892382797770619657504217539$    \\
 $w_4$ &$0.19186204094674514739760408197461$    \\
  $w_5$ & $-0.53123134392680669702873064192428$   \\
  $w_6$ & $-0.0081253242720827266680816105600661$   \\
  $w_7$ & $-0.16389450414378567860032917538393$   \\
  $w_8$ & $0.18514766119291405032528647881$   \\
  $w_9$& $0.5383584694754681989174668806505$   \\
  $w_{10}$& $-0.30583981835573485697292316732177$   \\
  $w_{11}$  & $0.43199935609523301289295473774488$    \\
  $w_{12}$& $0.1510502301631786853020124612813$   \\
  $w_{13}$& $-0.35051099204829676098801520498121$   \\
  $w_{14}$ & $0.1032971125844291674511513007661$   \\  
  $w_{15}$ & $0.15043936943817152697371946806229$  \\
    $w_{16}$ & $0.12118469498650736511410491586846$  \\
     $w_{17}$ & $0.10437742779547826358296681557444$  \\
 \hline
\end{tabular}
\caption{The 10th order solution with lowest eigenvalue error in our numerical search.}
\label{tab:order10-solm17}
\end{table}

\section{Methods for comparison of product formulae}\label{sec:comparison-method}
\subsection{Same order comparison}
A fair comparison between product formulae of different length may be made in the following way.
An order $k$ product formula for time $t$ will have an error $\delta=\zeta t^{k+1}$ (using $\zeta$ for eigenvalue error, or $\delta=\con t^{k+1}$ for the case of spectral-norm error).
Let $T$ be the total evolution time for a product formula of order $k$, and $\epsilon$ be the maximum allowable error. 
Subdivide the evolution time $T$ into $r$ subintervals, so $t={T}/{r}$ is the length of each time subinterval. 
We thus have $\zeta \left({T}/{r}\right)^{k+1}\approx{\epsilon}/{r}$, which gives
\begin{align}
    r\approx\left(\frac{\zeta T}{\epsilon}\right)^{{1}/{k}}T.
\end{align}

As explained above (see Eq.~\eqref{eq:noexp}), the number of exponentials in the product is $(4m+2)(J-1)+1$.
When applying products of these product formulae, two exponentials can be combined, so the effective number for each is $(4m+2)(J-1)$.
As a result, the total number of exponentials can be given as approximately
\begin{align}\label{eq:totexps}
    2(J-1) M\left(\frac{\zeta T}{\epsilon}\right)^{{1}/{k}}T \, ,
\end{align}
where $M=2m+1$.
The only quantities that vary between different product formulae of the same order are $M$ and $\zeta$,
so to compare product formulae, we need only compare the values of $M\zeta^{1/k}$; the one with the smaller value is the more efficient product formula.
Similarly, if we consider spectral-norm error then we should compare $M\con^{1/k}$ between the formulae.

As another motivation for this form, consider the case where the task is to estimate an eigenvalue of the Hamiltonian.
Then one would use an ancilla register to control the time of evolution, and use phase estimation.
If we are controlling between the forward and reverse time evolution as per Ref.~\cite{BabbushPRX18}, then the (root-mean square) error in estimating the phase of $e^{-iHT/r}$ is given by
\begin{equation}
    \epsilon_{\rm phase} = \frac{\pi}{2r}.
\end{equation}
For the error in the simulation of $e^{-iHT/r}$ to be no more than this phase estimation error, we should have
\begin{equation}
    \frac{\pi}{2r} \approx \zeta (T/r)^{k+1} ,
\end{equation}
or
\begin{equation}
    r \approx \left(\frac{2\zeta T}{\pi}\right)^{{1}/{k}}T \, .
\end{equation}
If $\epsilon_H$ is the allowed error in estimation of the eigenvalue of $H$, then it would correspond to $\pi/(2T)$.
Then, in terms of $\epsilon_H$ the choice of $r$ is
\begin{equation}
    r \approx \left(\frac{\zeta}{\epsilon_H}\right)^{{1}/{k}}\frac{\pi}{2\epsilon_H} \, .
\end{equation}
That would correspond to a total number of exponentials
\begin{equation}\label{eq:totexpspha}
    \pi(J-1)M \left(\frac{\zeta}{\epsilon_H}\right)^{{1}/{k}}\epsilon_H \, .
\end{equation}
Again the constant factor dependent on the product formula is $M\zeta^{1/k}$, so it is this expression that should be compared between formulae.
In terms of the overall complexity, this is similar to simulation of the evolution for time $T=\pi/(2\epsilon_H)$ with error $\epsilon=\pi/2$.

\subsection{Thresholds for different order}

If we wish to compare product formulae of \emph{different} order, then we need to take account of the values of $T$ and $\epsilon$. Assume we have two product formulae of order $k_1$ and $k_2$, with corresponding constants $\zeta_1$, $\zeta_2$.
Given $T$ and $\epsilon$, when the two product formulae use the same number of exponentials we have $M_1 r_1=M_2 r_2$, thus
\begin{align}
    M_1\left(\frac{\zeta_1T}{\epsilon}\right)^{{1}/{k_1}}T&=M_2\left(\frac{\zeta_2T}{\epsilon}\right)^{{1}/{k_2}}T\\
    \implies \quad \frac{T}{\epsilon} &= \left(\frac{M_2 \zeta_2^{1/k_2}}{M_1 \zeta_1^{1/k_1}}\right)^{\frac{1}{\frac{1}{k_1}-\frac{1}{k_2}}}.\label{eq:threshold}
\end{align}
This gives the threshold beyond which the higher-order product formula should be used for improved performance.
Again, when considering spectral-norm error, $\zeta$ would be replaced with $\con$.
In the case of eigenvalue estimation, the threshold is for $1/\epsilon_H$ rather than $T/\epsilon$.

A limitation of this analysis is that it is assumed the time step $(\epsilon/(\zeta T))^{1/k}$ is small enough for the scaling law for the error to hold.
To adjust the threshold for cases where the time step is large, we can consider a more general functional dependence of the error on the time interval as $f(t)$.
Then we require for time $T$ and error $\epsilon$ that
\begin{equation}
    f(t) \times T/t = \epsilon \, .
\end{equation}
The cost is then proportional to $M\times r = M T/t$.
For the threshold between two product formulae, we require
\begin{equation}
    M_1 T/t_1 = M_2 T/t_2 \, .
\end{equation}
That implies $t_2 = t_1 M_2/M_1$.
For the total error to be equal to $\epsilon$ in each case, we then require
\begin{equation}
    f_1(t_1) \times T/t_1 = f_2(t_2) \times T/t_2 \, .
\end{equation}
That can be rearranged to give
\begin{equation}
    f_1(t_1) = f_2(t_1 M_2/M_1) \times M_1/M_2 \, .
\end{equation}
This is an expression we can solve for $t_1$, which may then be used to determine the $T/\epsilon$ threshold from
\begin{equation}\label{eq:threshold_timestep}
    T/\epsilon = t_1/f_1(t_1)  \, .
\end{equation}
This means that the threshold is still for the ratio $T/\epsilon$, rather than having separate dependence on $T$ and $\epsilon$.

\subsection{Motivation for considering eigenvalue error}

Here, we give further explanation of why it is important to consider eigenvalue error.
If the goal is to estimate eigenvalues of the Hamiltonian (as is often the case in quantum chemistry), then the measure of the error should be that in the eigenvalues.
That is, how close are the eigenvalues of the product formula to those of the exact Hamiltonian evolution.
In the case that the goal is to accurately reproduce the final state after the evolution, then an appropriate measure of the error should be the spectral norm of the difference of the unitary operators.
This error accounts for both error in the eigenvalues and basis.

That error will upper bound the 2-norm error in the generated quantum state, but using the triangle inequality to bound the error for long evolution times overestimates the error.
This is because the error beyond that in the eigenvalues cancels when using a product of many short time intervals.
Let us denote the exact evolution operator for short time $U$, and the approximate evolution operator provided by the product formula as $\widetilde U$.
In the basis of the Hamiltonian's eigenstates, $U$ is diagonal, and we can diagonalise $\widetilde U$ in this basis as $\widetilde U = V D V^\dagger$.
Then the difference between $D$ and $U$ describes the eigenvalue error, since these are diagonal matrices with the eigenvalues on the diagonal.
The matrix $V$ describes the basis error.

It is convenient to write $V=e^{i\tilde H}$ for some Hermitian matrix $\tilde H$.
Then the error can be given as
\begin{align}
    \| \widetilde U - U \| &= \| e^{i\tilde H} D e^{-i\tilde H} - U \| \nn
    &= \| (I + i\tilde H + \mathcal{O}(\tilde H^2)) D (I - i\tilde H + \mathcal{O}(\tilde H^2)) - U \| \nn
    &= \| i [\tilde H, D] + D-U + \mathcal{O}(\tilde H^2)\| \nn
    &\le \| [\tilde H, D-I] \| + \|D-U\|  + \mathcal{O}(\tilde H^2) \, .
\end{align}
That is, the error can be split up between a part $[\tilde H, D-I]$ corresponding to basis error, and a part $D-U$ corresponding to eigenvalue error.
Note that $D-I$ is proportional to $t$, so the error in the operator due to the basis error is one order higher than the error in the basis.

Then, for the case of a large number of time steps $r$, we can bound the overall error as
\begin{align}
    \| \widetilde U^r - U^r \| &= \| V D^r V^\dagger - U^r \| \nn
    &= \| (V D^r V^\dagger - 
    D^rV^\dagger)
    +(D^r V^\dagger
    - U^r V^\dagger)
    +(U^r V^\dagger
    - U^r ) \| \nn
    &\le \| V - I \| + \|
    D^r-U^r\|+
    \| V^\dagger- I \| \nn
    &\le 2 \| V - I \| + r\| D-U\| \, .
\end{align}
That is, the error due to the eigenvalues is multiplied by a factor of $r$, but the error due to the basis is not.
The error due to the basis can be larger than for a single time step, but only by a factor of $\sim 1/t$.
That factor arises because for a single time step the basis error is commuted with $D-I$ which is proportional to $t$, so the expression here without $D-I$ can be a factor $\propto 1/t$ larger.

In practice we find that in many cases the error in the basis scales as $\mathcal{O}(t^k)$, resulting in the expected contribution to the spectral-norm error $\mathcal{O}(t^{k+1})$.
However, in other cases we find scaling of the error in the basis as $\mathcal{O}(t^{k+1})$, resulting in a higher-order contribution to the spectral-norm error. 
Let us give the error in the eigenvalues as $\zeta t^{k+1}$, and the error in the basis as $\mu t^{k+\nu}$ where $\nu\in[0,1]$ to account for the various scalings found numerically.

Then the choice of $r$ to make the error in the eigenvalues smaller than $\epsilon$ is
\begin{equation}
    r \approx \left( \frac{\zeta T}{\epsilon} \right)^{1/k} T .
\end{equation}
The contribution to the total error from the basis error is 
\begin{equation}
    2\mu (T/r)^{k+\nu} \approx 2\mu \left(\frac{\epsilon}{\zeta T}\right)^{1+\nu/k} \, .
\end{equation}
That gives the ratio of the basis error to the eigenvalue error as approximately
\begin{equation}
    \frac{2\mu}{\epsilon} \left(\frac{\epsilon}{\zeta T}\right)^{1+\nu/k} = \frac{2\mu}{\zeta T} \left(\frac{\epsilon}{\zeta T}\right)^{\nu/k} \, .
\end{equation}
Hence, for large $T$ the eigenvalue error should be dominant.
This means that the eigenvalue error for a single step is the relevant measure to estimate the spectral-norm error for long-time evolution.

\section{Numerical comparison of product formulae}\label{sec:comparison-formulae}
\subsection{Comparison for general matrices}

In what follows, we report on the comparison between product formulae of 4th, 6th, 8th and 10th order based on the threshold provided above.
We analyse those found in this work, and those from prior work including processed product formulae.
When giving the number of stages of a processed product formula, we refer to the number of stages in the kernel (not  $P$), because those are the ones that would be repeated for a simulation over an extended time.

We list 4th order product formulae in \tab{order4}, 6th order in \tab{order6-constant}, 8th order in \tab{order8-constant}, and 10th order in \tab{order10-constant}, and give the constant factors in the error scaling $\con,\zeta$ (spectral norm or eigenvalue error).
These product formulae include some constructed with even $M$ rather than $M=2m+1$ as for our product formulae.
We generated 1024 random Hamiltonians for each product formula, and then estimated the constants $\con$ and $\zeta$ by taking the geometric mean.
Each random Hamiltonian is generated as a sum of two random Hamiltonians $H=A+B$, where $A$ and $B$ are random Hermitian matrices of dimension $64\times 64$ and norm 1.

The size of the time step used for the estimation is $t=e^{-5/2}\approx 0.08$, which we found to be small enough for accurate estimation of the constants.
An exception is the 10th order product formula PP10s23 in the second-last line of \tab{order10-constant}.
That was evaluated for $t=e^{-3/2}\approx 0.2$, rather than $t\approx 0.08$ for the other product formulae.
Estimates for $t\approx 0.08$ are significantly larger, possibly because the coefficients for that product formula were not given to high precision.
We provide the estimates of $\con,\zeta$ for larger $t$ to evaluate it as fairly as possible.

\begin{table}[tbh]
\centering
\begin{tabular}{ |c|c|c|c|c|c|c|c|c| } 
 \hline
 label & $M$ & $S_2$ & processing & reference & $\con$ & $M\con^{1/k}$ & $\zeta$ & $M\zeta^{1/k}$ \\ 
 \hline
  S4m1 & $3$ & Y & N & first fractal product \cite{Suzuki1990} & $4.9\times 10^{-2}$ & 1.41 & $2.3\times 10^{-2}$  & 1.17 \\
  S4m2 & $5$ & Y & N & second fractal product \cite{Suzuki5} & $3.0\times 10^{-3}$ & 1.17 & $3.3\times 10^{-4}$  & 0.67 \\
O4M5 & $5$ & N & N & Ostmeyer Eq.~(40) of \cite{Ostmeyer_2023} & $3.0\times 10^{-4}$ & 0.66 & $7.7\times 10^{-5}$ &  0.47 \\
    \rowcolor{green1}
BM4M6 & $6$ & N & N & $S_6$ Table 2 of \cite{BLANES2002313} & $1.6\times 10^{-4}$ & 0.67 & $2.4\times 10^{-5}$ &  0.42 \\
 PPBCM4m6  & $6$ & N & Y & $P_6 4$ Table 5 of \cite{Blanes2006processing} & $5.7\times 10^{-5}$ & 0.52 & $2.4\times 10^{-5}$  & 0.42 \\
 BCE4m3  & $6$ & Y & Y & $\psi_3^{[4]}$ Table 6 of \cite{blanes2024families} & $1.6\times 10^{-2}$ & 2.15 & $1.5\times 10^{-3}$  & 1.17 \\
  BCE4m4 & $8$ & Y & Y & $\psi_4^{[4]}$ Table 6 of \cite{blanes2024families} & $3.4\times 10^{-4}$ & 1.09 & $9.8\times 10^{-5}$  & 0.80 \\
  BCE4m5 & $10$ & Y & Y & $\psi_5^{[4]}$ Table 6 of \cite{blanes2024families} & $6.7\times 10^{-5}$ & 0.90 & $2.1\times 10^{-5}$  & 0.68 \\
  BCE4m6 & $12$ & Y & Y & $\psi_6^{[4]}$ Table 6 of \cite{blanes2024families} & $2.6\times 10^{-5}$ & 0.85 & $7.0\times 10^{-6}$  & 0.62 \\  
  BCE4m7 & $14$ & Y & Y & $\psi_7^{[4]}$ Table 6 of \cite{blanes2024families} & $1.3\times 10^{-5}$ & 0.85 & $3.0\times 10^{-6}$  & 0.58 \\  
  BCE4m8 & $16$ & Y & Y & $\psi_8^{[4]}$ Table 6 of \cite{blanes2024families} & $7.9\times 10^{-5}$ & 0.85 & $1.5\times 10^{-6}$  & 0.56 \\  
  BCE4m9 & $18$ & Y & Y & $\psi_9^{[4]}$ Table 6 of \cite{blanes2024families} & $5.3\times 10^{-6}$ & 0.86 & $8.6\times 10^{-7}$  & 0.55 \\ 
   \hline
\end{tabular}
\caption{The list of 4th order product formulae.
The column labelled ``processing'' indicates whether the formula uses processor.
The label is the name we will refer to the formula by.
The constant factor in the error is denoted $\con$ for spectral-norm error and $\zeta$ for eigenvalue error, and the corresponding quantities $M\con^{1/k}$ and $M\zeta^{1/k}$ are given.
The column labelled $S_2$ indicates whether the product formula is a product of $S_2$.
The best result is highlighted in green, and is from 2002 \cite{BLANES2002313}.
The 2006 result PPBCM4m6 is slightly higher error, but it is not evident to 2 significant figures.}
\label{tab:order4}
\end{table}

\begin{table}[tbh]
\centering
\begin{tabular}{ |c|c|c|c|c|c|c|c|c| } 
 \hline
 label & $M$ & $S_2$ & processing & reference & $\con$ & $M\con^{1/k}$ & $\zeta$ & $M\zeta^{1/k}$ \\ 
 \hline
   S6m1 & $9$ & Y & N & first fractal product \cite{Suzuki1990} & $4.5\times 10^{-2}$ & 5.36 & $2.3\times 10^{-2}$  & 4.81 \\
  S6m2 & $25$ & Y & N & second fractal product \cite{Suzuki5} & $1.3\times 10^{-5}$ & 3.83 & $2.0\times 10^{-7}$  & 1.91 \\ 
 Y6m3a & $7$ & Y & N & Solution A in Table 1 of \cite{Yoshida1990} & $2.0\times 10^{-3}$ & 2.49 & $1.2\times 10^{-3}$ & 2.28 \\
  KL6s9a & $9$ & Y & N & s9odr6a in Appendix A  of \cite{Kahan1997integrators} & $2.9\times 10^{-4}$ & 2.31 & $1.7\times 10^{-4}$  & 2.12 \\
 KL6s9b & $9$ & Y & N & s9odr6b in Appendix A  of \cite{Kahan1997integrators} & $2.9\times 10^{-4}$ & 2.31 & $1.7\times 10^{-4}$ & 2.11  \\
 SS6s11 & $11$ & Y & N & Section 4.2 of \cite{Sofroniou2005integrators} & $4.2\times 10^{-5}$ & 2.05 & $1.4\times 10^{-5}$ & 1.71  \\
 SS6s13 & $13$ & Y & N & Section 4.2 of \cite{Sofroniou2005integrators} & $2.1\times 10^{-5}$ & 2.16 & $3.4\times 10^{-6}$ &  1.59  \\
  BM6M10 & $10$ & N & N & $S_{10}$ Table 2 of \cite{BLANES2002313} & $6.0\times 10^{-6}$ & 1.35 & $1.4\times 10^{-6}$  & 1.06 \\  
    \rowcolor{green1}
PPBCM6m9 & $9$ & N & Y & $P_9 6$ Table 5 of \cite{Blanes2006processing} & $5.0\times 10^{-6}$ & 1.18 & $4.3\times 10^{-7}$  & 0.78 \\
 PPBCM6m5 & $11$ & Y & Y & $P_{11}6$ in Table 6 of \cite{Blanes2006processing} & $1.9\times 10^{-6}$ & 1.22 & $9.0\times 10^{-7}$ & 1.08 \\
 PPBCM6m6 & $13$ & Y & Y & $P_{13}6$ in Table 6 of \cite{Blanes2006processing} & $4.8\times 10^{-7}$ & 1.15 & $2.5\times 10^{-7}$ & 1.03 \\
  BCE6m5 & $10$ & Y & Y & $\psi_5^{[6]}$ Table 8 of \cite{blanes2024families} & $5.8\times 10^{-3}$ & 4.24 & $3.2\times 10^{-3}$  & 3.84 \\
  BCE6m6 & $12$ & Y & Y & $\psi_6^{[6]}$ Table 8 of \cite{blanes2024families} & $7.7\times 10^{-5}$ & 2.47 & $2.0\times 10^{-5}$  & 1.98 \\  
  BCE6m7 & $14$ & Y & Y & $\psi_7^{[6]}$ Table 8 of \cite{blanes2024families} & $1.6\times 10^{-5}$ & 2.22 & $2.4\times 10^{-6}$  & 1.62 \\  
  BCE6m8 & $16$ & Y & Y & $\psi_8^{[6]}$ Table 8 of \cite{blanes2024families} & $4.6\times 10^{-6}$ & 2.06 & $3.2\times 10^{-7}$  & 1.32 \\  
  BCE6m9 & $18$ & Y & Y & $\psi_9^{[6]}$ Table 8 of \cite{blanes2024families} & $7.3\times 10^{-6}$ & 2.51 & $1.3\times 10^{-7}$  & 1.28 \\  
  BCE6m10 & $20$ & Y & Y & $\psi_{10}^{[6]}$ Table 8 of \cite{blanes2024families} & $3.2\times 10^{-6}$ & 2.42 & $9.9\times 10^{-9}$  & 0.93 \\  
  BCE6m11 & $22$ & Y & Y & $\psi_{11}^{[6]}$ Table 8 of \cite{blanes2024families} & $7.9\times 10^{-6}$ & 3.11 & $1.4\times 10^{-8}$  & 1.08 \\  
  \hline
\end{tabular}
\caption{The list of 6th order product formulae with their average errors.
The column labelled $S_2$ indicates whether it is a product of $S_2$.
The best result is highlighted in green, and is from 2006 \cite{Blanes2006processing}.}
\label{tab:order6-constant}
\end{table}

\begin{table}[tbh]
\centering
\begin{tabular}{ |c|c|c|c|c|c|c|c| } 
 \hline
 label & $M$ & processing & reference & $\con$ & $M\con^{1/k}$ & $\zeta$ & $M\zeta^{1/k}$ \\ 
 \hline
   S8m1 & $27$ & N & first fractal product \cite{Suzuki1990} & $5.6\times 10^{-2}$ & 18.8 & $1.6\times 10^{-2}$  & 16.2 \\
  S8m2 & $125$ & N & second fractal product \cite{Suzuki5} & $6.7\times 10^{-9}$ & 11.9 & $5.2\times 10^{-13}$  & 3.64 \\ 
 Y8m7d & $15$ & N & Solution D in Table 2 of \cite{Yoshida1990} & $1.1\times 10^{-3}$ & 6.41 & $1.3\times 10^{-4}$ & 4.89 \\
MC8s15 & $15$ & N & Table 2 of \cite{McLachlan95} & $6.5\times 10^{-6}$ & 3.37 & $2.0\times10^{-6}$ & 2.90 \\
MC8s17 & $17$ & N & Table 2 of \cite{McLachlan95} & $7.9\times 10^{-7}$ & 2.94 & $3.3\times10^{-7}$ & 2.63 \\
KL8s17a & $17$ & N & s17odr8a in \cite{Kahan1997integrators} & $6.1\times 10^{-7}$ & 2.84 & $2.7\times10^{-7}$ & 2.56 \\
KL8s17b & $17$ & N & s17odr8b in \cite{Kahan1997integrators}  & $5.9\times 10^{-7}$ & 2.83 & $2.5\times10^{-7}$ & 2.54 \\
\rowcolor{blue2}
SS8s19 & $19$ & N & Section 4.3 of \cite{Sofroniou2005integrators} & $1.8\times 10^{-7}$ & 2.72 & $5.3\times 10^{-8}$ & 2.34 \\
SS8s21 & $21$ & N & Section 4.3 of \cite{Sofroniou2005integrators} & $3.4\times 10^{-7}$ & 3.27 & $7.5\times 10^{-8}$ & 2.70 \\
PP8s13 & $13$ & Y & $P_{13}8$ in Table 6 of \cite{Blanes2006processing} & $1.2\times 10^{-6}$ & 2.37 & $6.5\times 10^{-7}$ & 2.19 \\
\rowcolor{green2}
PP8s19 & $19$ & Y & $P_{19}8$ in Table 6 of \cite{Blanes2006processing} & NA & NA & $2.4\times 10^{-7}$ & 2.83 \\
Y8m10 & $21$ & N & \cref{tab:order8-solm10} (our new result) & $5.8\times 10^{-8}$ & 2.61 & $7.0\times 10^{-9}$ & 2.01 \\
\rowcolor{blue1}
Y8m10b & $21$ & N & \cref{tab:order8-solm10} (our new result) & $6.3\times 10^{-7}$ & 3.53 & $5.4\times 10^{-10}$ & 1.46 \\
\rowcolor{green1}
YP8m8 & $17$ & Y & Table \ref{tab:process} (our new result) & $5.3\times 10^{-8}$ & 2.09 & $8.1\times 10^{-10}$ & 1.24 \\
 \hline
\end{tabular}
\caption{The list of 8th order product formulae with their average errors.
This table shows our new results for product formulae that improve over those in the prior literature.
The result that provides the most efficient simulations is highlighted in (dark) green, and is our processed product formula.
That provides the best performance due to the shorter length, but the lowest error is provided by our product formula highlighted in (dark) blue.
The lowest error results from prior work are highlighted in light blue (for the non-processed case), and light green (for the processed case).
The processed product formula PP8s13 ($P_{13}8$) has larger error than the highlighted PP8s19, but better performance due to its shorter length.}
\label{tab:order8-constant}
\end{table}

\begin{table}[htb]
\centering
\begin{tabular}{ |c|c|c|c|c|c|c|c| } 
 \hline
 label & $M$ & processing & reference & $\con$ & $M\con^{1/k}$ & $\zeta$ & $M\zeta^{1/k}$ \\ 
 \hline
S10m1 & $81$ & N & first fractal product \cite{Suzuki1990} & $9.0\times 10^{-2}$ & 63.7 & $2.7\times 10^{-3}$  & 44.8 \\
S10m2 & $625$ & N & second fractal product \cite{Suzuki5} & $4.1\times 10^{-13}$ & 36.0 & $6.1\times 10^{-19}$  & 9.44 \\ 
KL10s31a & $31$ & N & s31odr10a in Appendix A of \cite{Kahan1997integrators} & $8.7\times10^{-6}$ & 9.67 & $5.8\times10^{-6}$ & 9.28 \\
KL10s31b & $31$ & N & s31odr10b in Appendix A of \cite{Kahan1997integrators} & $8.4\times10^{-5}$ & 12.1 & $4.1\times10^{-5}$ & 11.3 \\
Tsi10s33 & $33$ & N & Table II of \cite{Tsitouras1999} & $3.3\times 10^{-6}$ & 9.33 & $6.4\times 10^{-7}$  & 7.93 \\ 
SS10s31 & $31$ & N & Section 4.4 of \cite{Sofroniou2005integrators} & $4.2\times10^{-8}$ & 5.66 & $2.4\times 10^{-8}$ & 5.36  \\
SS10s33 & $33$ & N & Section 4.4 of \cite{Sofroniou2005integrators} & $1.4\times10^{-8}$ & 5.42 & $8.8\times 10^{-9}$ & 5.17  \\
\rowcolor{green1}
SS10s35 & $35$ & N & Section 4.4 of \cite{Sofroniou2005integrators} & $1.0\times10^{-9}$ & 4.41 & $3.1\times 10^{-11}$ &  3.11 \\
Alberdi31 & $31$ & N & Appendix A of \cite{alberdi2019algorithm} &  $1.5\times 10^{-7}$ & 6.44 & $1.0\times 10^{-7}$ & 6.20 \\
Alberdi33 & $33$ & N & Appendix A of \cite{alberdi2019algorithm} &  $8.4\times 10^{-8}$ & 6.47 & $5.5\times 10^{-8}$ & 6.20 \\
Alberdi35 & $35$ & N & Appendix A of \cite{alberdi2019algorithm} &  $1.5\times 10^{-8}$ & 5.79 & $9.5\times 10^{-9}$ & 5.52 \\
PP10s19 & $19$ & Y & $P_{19}10$ in Table 6 of \cite{Blanes2006processing} & NA & NA & $6.2\times 10^{-6}$ & 5.73 \\
PP10s23 & $23$ & Y & $P_{23}10$ in Table 6 of \cite{Blanes2006processing} & $1.4\times10^{-6}$ & 5.96 & $1.6\times 10^{-8}$ & 3.82 \\
Y10m17 & $35$ & N & \cref{tab:order10-solm17} & $1.9\times 10^{-8}$ & 5.91 & $6.1\times 10^{-11}$ & 3.33 \\
 \hline
\end{tabular}
\caption{The list of 10th order product formulae with their average errors.
The best result is highlighted in green, and is the product formula from 2005 \cite{Sofroniou2005integrators}.
That product formula provides exceptional performance for 10th order, being about 280 times better than all others from prior work, and half the error of our solution (bottom line).}
\label{tab:order10-constant}
\end{table}

To compare among product formulae of the same order, we compare $M\zeta^{1/k}$ for eigenvalue error, and the best results in each table are highlighted.
We also provide values of $M\con^{1/k}$ for comparison, but these are not used to select the best product formulae.
For 8th order it can be seen that our new processed product formula provides the best performance out of all formulae tested.
Our two new product formulae of length $m=10$ ($M=21$) also give better performance than any prior product formulae tested, including processed product formulae.
Hence we have given a sequence of three new solutions here, all of which provide improved performance over prior work.
There are also other results in our database of over 100,000 solutions that outperform prior work, but we have just given a selection of the best here.

In terms of eigenvalue error, our solution Y8m10b is about 100 times more accurate than the lowest-error prior result, SS8s19.
Our processed product formula YP8m8 is about a factor of 300 times more accurate than the lowest-error prior result for processed product formulae, PP8S19, which is longer.
These surprisingly large improvements demonstrate that the solution space had not been fully explored in prior work.
Each increase in $m$ greatly expands the solution space to explore, and it is likely that there are further improvements still to be found.
The error for our processed product formula YP8m8 has slightly larger error than our non-processed formula Y8m10b, but its shorter length means that $M\zeta^{1/k}$ is smaller, thus providing the best performance out of all product formulae tested.

These improvements are consistent between the samples of the Hamiltonians.
In comparing our formula Y8m10b to the best from prior work SS8s19, our \emph{smallest} improvement (out of 1024 samples) was about a factor of 35.
A histogram of the improvement for our new product formula is shown in Fig.~\ref{fig:Hist}.
We also find that the improvements are consistent between sizes of the matrices.
We have tested sizes 6, 16, and 32 as well as 64 (which is shown in the tables), and find significant improvements at all sizes tested.

\begin{figure}
	\includegraphics[scale=0.6]{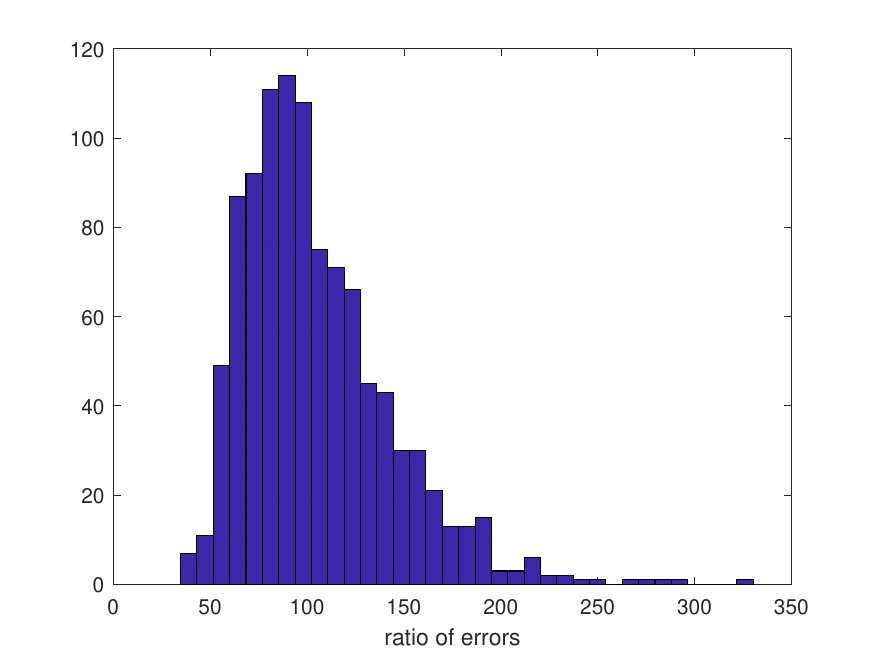}
	\caption{A histogram of the ratio of the error of the most accurate prior product formula SS8s19 to our new solution, in the case without processing.
 The histogram is over the 1024 samples of $64\times 64$ matrices.
 The average improvement is about a factor of 100, with the lowest being about 35.
}\label{fig:Hist}
\end{figure}

A recent review \cite{Blanes_Casas_Murua_2024} also provides a comparative analysis of many of the product formulae in the literature, and for 8th order their recommended methods are SS8s21 in the case without processing, and PP8s13 for the case with processing.
In our testing we find that SS8s19 is slightly better than SS8s21; this difference may be because we are testing Hamiltonians here, whereas Ref.~\cite{Blanes_Casas_Murua_2024} considers general matrices (which are non-Hermitian).
PP8s19 provides lower error than PP8s13, but PP8s13 provides better performance due to the shorter length (and correspondingly smaller value of $M\zeta^{1/k}$).
The three new results we have listed in Table \ref{tab:order8-constant} all outperform those methods recommended by Ref.~\cite{Blanes_Casas_Murua_2024}.

As a further test, we have calculated the error for simulating the transverse-field Ising model for an 8-qubit chain.
We used a time step of 1, with the two components of the Hamiltonian normalised to unit norm, so it is
\begin{equation}\label{eq:Ising}
    H = -\frac 17 \sum_{j=1}^7 Z_j Z_{j+1} + \frac 18 \sum_{j=1}^8 X_j \, .
\end{equation}
We then determined the spectral-norm error for up to 1000 steps, for our minimum eigenvalue error solution Y8m10b, as well as SS8s19 and SS8s21 which provided the minimum eigenvalue error out of results presented in prior work.
As seen in Figure~\ref{fig:Ising}, the error is far smaller for our new product formula than that provided by those from prior work.
The initial error is larger as indicated by the larger spectral-norm error for a single step, but then the error for a large number of time steps is dominated by the eigenvalue error which is much smaller for this product formula.
It is also surprising that the simulation is highly accurate for a time step of 1, which is not the short time step normally considered for scaling of the error.
The eigenvalue error for a single time step of 1 is only $2.7\times 10^{-12}$ for our product formula.

\begin{figure}
	\includegraphics[scale=0.4]{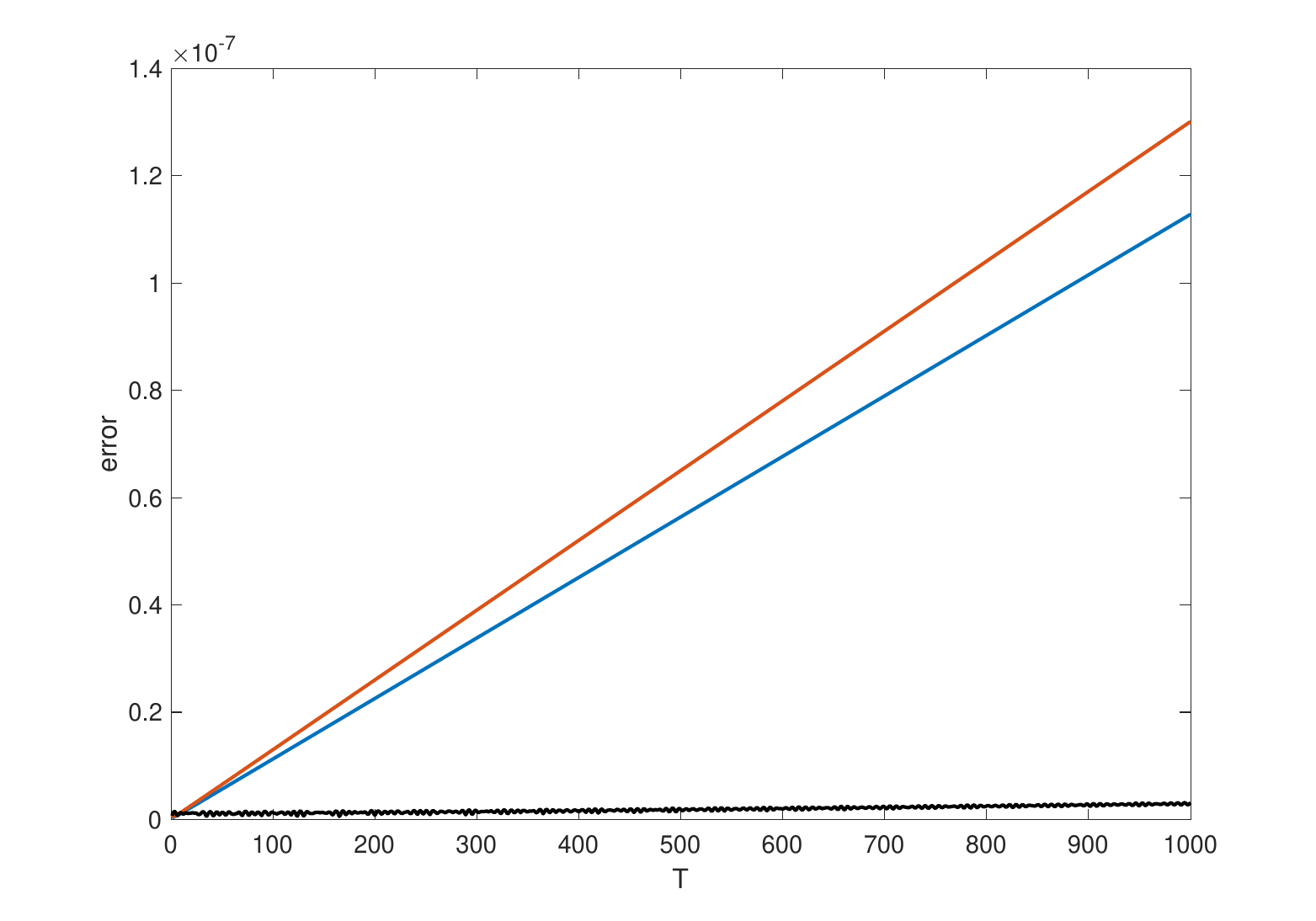}
	\caption{The spectral-norm error for simulating an 8-qubit transverse-field Ising Hamiltonian in Eq.~\eqref{eq:Ising} for our product formula Y8m10b (black), and SS8s19 (blue) and SS8s21 (red) from Ref.~\cite{Sofroniou2005integrators}.
 The size of the time step is 1, and the components of the Hamiltonian are normalised to unit norm.
}\label{fig:Ising}
\end{figure}

In the case of 10th order, we find that the best performance is provided by SS10s35, from Sofroniou and Spaletta \cite{Sofroniou2005integrators} (which is recommended by Ref.~\cite{Blanes_Casas_Murua_2024}).
The performance of that product formula is extraordinarily good, with error nearly a factor of 280 times less than any others we have tested from prior work.
Our solution Y10m17 provides intermediate performance, having error twice as large as SS10s35, but a factor of 140 times better than any others.
For example, it outperforms the processed product formula PP10s23, which was also recommended by Ref.~\cite{Blanes_Casas_Murua_2024}.
The fact that there is a very exceptional solution at one particular length for 10th order suggests that there are better solutions still to be found, motivating further searches in future work.

For 4th order the best result for eigenvalue error is given by the formula BM4M6 of Blanes and Moan~\cite{BLANES2002313}, which is not in the form of a product of $S_2$.
The processed product formula PPBCM4m6 also provides excellent performance, with a marginally larger error in our testing.
A recently proposed formula O4M5 of Ostmeyer~\cite{Ostmeyer_2023} also gives good performance.
In the 6th order case we find that the best performance is given by PPBCM6m9, which is a processed product formula not constructed as a product of $S_2$.
The recently proposed processed product formula BCE6m10 also gives excellent performance with low error, but is much longer, resulting in a larger value of $M\zeta^{1/k}$.

In all cases the fractal products provide poor performance.
The second fractal product yields very low error but very large $M$, equal to 125 for 8th order and 625 for 10th order.
This is far beyond that for any other product formulae considered here, making this product formula uncompetitive despite the small error.
The first fractal product yields large error as well as large $M$, making its performance by far the worst out of any product formula tested. 
We also tested various combinations of the two formulae (by using different formulae in the iteration to increase the order) or using fractal formulae to increase the order of lower-order formulae.
None of these provided competitive performance, so are not included in the tables.

To consider the threshold to use a higher-order product formula over a lower-order one, we first consider the asymptotic expression using $M\zeta^{1/k}$ for eigenvalue error. 
For 4th order the value of this quantity is 0.42 (BM4M6) and for 6th order is 0.78 (PPBCM6m9), which gives a threshold of 1700.
The value for the best 8th order is 1.24 (our result YP8m8), which gives a threshold of around 6000 when compared with 4th order and 68,000 when compared with 6th order.

At 10th order the best value is 3.11 (SS10s35), which when compared with 8th order gives a threshold of $9.3\times 10^{15}$.
This is far greater than can be expected for realistic simulations.
To see the reason for this, if we use Eq.~\eqref{eq:totexpspha} for the approximate number of exponentials needed for phase estimation (where the threshold is for $1/\epsilon_H$), it would be more than $3\times 10^{18}$.
In realistic implementations of quantum chemistry it takes about $10^4$ Toffoli gates per exponential, so this would correspond to a Toffoli complexity over $10^{22}$.
For comparison, in the analysis of Shor's algorithm in Ref.~\cite{Gidney2021}, $3\times 10^9$ Toffolis were estimated to take 8 hours.
Even with improvements in gate speed, a simulation sufficiently large for 10th order product formulae to provide an advantage would require a computation time of many millions of years.

This threshold between 8th and 10th order can be expected to be accurate (at least for the estimate of the order of magnitude), because it corresponds to small time steps below $0.3$.
On the other hand, the other thresholds would correspond to large time steps, so would be inaccurate because the asymptotic expressions break down.
The calculated threshold between 6th and 8th order would correspond to a time step larger than 3, which is far too large for the scaling of the error to be accurate.

In this case we find numerically that the threshold between these 6th and 8th order formulae (PPBCM6m9 and YP8m8) is increased dramatically, to about $3.2\times 10^{7}$, corresponding to a time step of about $1.22$.
If the target $\epsilon$ is $0.01$ for example, that would correspond to about 260,000 time steps.
It is possible to optimise our kernel for the larger time step size in order to provide better performance.
This optimisation provides a kernel (see Table \ref{tab:processl}) such that the threshold is reduced about a factor of twelve to $2.8\times 10^{6}$.
This now corresponds to 15,000 time steps with a time step size of about $1.82$.
This is fewer time steps than expected for simulations of classically intractable systems.
Thus the optimisation provides another order of magnitude where 8th order product formulae are optimal, albeit with a product formula tailored to the time step size.

\begin{table}[tbh]
\centering
\begin{tabular}{ |c|c| } 
 \hline
 & kernel for large time step \\ 
 \hline
 $w_1$ & $0.1777372900430394$  \\ 
 $w_2$ & $0.2862580532195395$  \\
 $w_3$ & $0.1701306063199336$  \\
 $w_4$ & $-0.3746748008394162$  \\
  $w_5$ & $0.1485267804844835$  \\
  $w_6$ & $-0.3773225725485588$  \\
  $w_7$ & $0.5395886879620081$   \\
  $w_8$ & $0.2210419534887659$  \\
 \hline
\end{tabular}
\caption{The 8th order kernel tailored for large time step size $t=1.82$ to improve the threshold.}
\label{tab:processl}
\end{table}

One must also be careful in considering thresholds based on eigenvalue error, because there needs to be a large number of time steps such that the spectral-norm error in the complete evolution time is dominated by the eigenvalue error.
There is a factor of about 65 between the spectral-norm error and the eigenvalue error for the 8th order product formula YP8m8.
For these thresholds the number of time steps is orders of magnitude larger than this, which indicates that the eigenvalue error is an appropriate measure.

These threshold calculations are for Hamiltonians that are sums of two terms, each of which is normalised.
In the case where the Hamiltonian is not normalised, the threshold should be scaled by the norm.
In the simple case where the Hamiltonian is a sum of two terms with equal norms $\|A\|=\|B\|$, then the threshold will be for $\|A\|T/\epsilon$ rather than $T/\epsilon$.
More generally the Hamiltonian may be a sum of any number of terms with different norms, which would change the threshold and make it unclear what quantity the threshold should be considered in terms of.
There is also the possibility that the threshold can be changed for Hamiltonians that have structure to them instead of being random, for example those for quantum chemistry.

\subsection{Comparison for fermionic Hamiltonians}

Fermionic Hamiltonians encountered in quantum chemistry often have the form 
\begin{align}\label{eq:fermion-ham}
    \sum_{p,q=1}^{d} \tau_{pq} a^{\dag}_{p} a_q + \sum_{p,q=1}^d \nu_{pq} a^{\dag}_{p} a_p a^{\dag}_{q} a_q
\end{align}
where $a^\dag_p$ and $a_p$ are the fermionic creation and annihilation operators acting on orbital $p$, and there are a total of $d$ orbitals.
Each entry $\tau_{pq}$, $\nu_{pq}$ is real and there is symmetry in exchanging indices.
The behaviour of the error as the size of system is changed can be predicted based on the result in Theorem 4 of \cite{Low2022trotter}.

\begin{theorem}[Theorem 4 in \cite{Low2022trotter}]
Let $H=T+V=\sum_{p,q}^{d} \tau_{pq} a^{\dag}_{p} a_q + \sum_{p,q}^d \nu_{pq} a^{\dag}_{p} a_p a^{\dag}_{q} a_q$ be an interacting-electronic Hamiltonian, and $S_k(t)$ be a $k$th order product formula splitting the evolutions under $T$ and $V$. Then
\begin{equation}\label{eq:Low_bound}
    \norm{S_k(t)-e^{-itH}}_{W_\eta}=\order{(\norm{\tau}_{1}+\norm{\nu}_{1,[\eta]})^{k-1}\norm{\tau}_1\norm{\nu}_{1,[\eta]}\eta t^{k+1}} \, ,
\end{equation}
where $\norm{\cdot}_{W_\eta}$ corresponds to the operator norm on the operator acting in the $\eta$-electron subspace, $\norm{\tau}_1 = \max_p \sum_{q} \abs{\tau_{pq}}$ and $\norm{\nu}_{1,[\eta]}=\max_p \max_{q_1<q_1<...<q_\eta}(\abs{\nu_{p q_1}}+\cdots+\abs{\nu_{pq_\eta}})$.
\end{theorem}

To fairly compare product formulae as applied to quantum chemistry, we can define $\omega$ so that
\begin{align}
    \zeta &= \omega \, (\norm{\tau}_{1}+\norm{\nu}_{1,[\eta]})^{k-1}\norm{\tau}_1\norm{\nu}_{1,[\eta]} \eta \, .
\end{align}
Note that the expression in Ref.~\cite{Low2022trotter} was derived for error in the spectral norm, but that also provides an upper bound on the error in the eigenvalues, so it is reasonable to consider the constant defined for eigenvalue error.
Then the formula for the threshold $T/\epsilon$ becomes
\begin{align}
\frac{T}{\epsilon} &= \left(\frac{M_2 [\omega_2 (\norm{\tau}_{1}+\norm{\nu}_{1,[\eta]})^{k_2-1}\norm{\tau}_1\norm{\nu}_{1,[\eta]}\eta]^{1/k_2}}{M_1 [\omega_1 (\norm{\tau}_{1}+\norm{\nu}_{1,[\eta]})^{k_1-1}\norm{\tau}_1\norm{\nu}_{1,[\eta]}\eta]^{1/k_1}}\right)^{\frac{1}{\frac{1}{k_1}-\frac{1}{k_2}}} \nn
&= \left(\frac{M_2 [\omega_2 \norm{\tau}_1\norm{\nu}_{1,[\eta]}\eta/(\norm{\tau}_{1}+\norm{\nu}_{1,[\eta]})]^{1/k_2}}{M_1 [\omega_1 \norm{\tau}_1\norm{\nu}_{1,[\eta]}\eta/(\norm{\tau}_{1}+\norm{\nu}_{1,[\eta]})]^{1/k_1}}\right)^{\frac{1}{\frac{1}{k_1}-\frac{1}{k_2}}} \nn
\frac{\norm{\tau}_1\norm{\nu}_{1,[\eta]}\eta}{\norm{\tau}_{1}+\norm{\nu}_{1,[\eta]}} \frac{T}{\epsilon} &= \left(\frac{M_2 \omega_2^{1/k_2}}{M_1 \omega_1^{1/k_1}}\right)^{\frac{1}{\frac{1}{k_1}-\frac{1}{k_2}}} .
\label{chemthresh}
\end{align}
For estimation of eigenvalues (rather than simulation of time evolution) the same derivation holds with $T/\epsilon$ replaced by $1/\epsilon_H$.
Thus we see that the ratio ${\norm{\tau}_1\norm{\nu}_{1,[\eta]}\eta}/(\norm{\tau}_{1}+\norm{\nu}_{1,[\eta]})$ governs the threshold where a higher-order product formula will improve over a lower-order product formula.
This is somewhat different from the factor of the norm of the Hamiltonian that might otherwise be expected for unstructured Hamiltonians.

To analyse the non-asymptotic regime, we rewrite the error to be proportional to
\begin{equation}
    \frac{\norm{\tau}_1\norm{\nu}_{1,[\eta]}\eta}{\norm{\tau}_{1}+\norm{\nu}_{1,[\eta]}} t \times \left[ {t}\left(\norm{\tau}_{1}+\norm{\nu}_{1,[\eta]}\right)\right]^{k}.
\end{equation}
Following the derivation of Ref.~\cite{Low2022trotter}, it is easily seen that the contributions to the error from each of the orders higher than $k+1$ will be of a similar form.
These higher-order contributions to the error can be used to determine the error in the non-asymptotic regime where the order $k+1$ term is not dominant.

To see why that is the case, note that Lemma 2 of Ref.~\cite{Low2022trotter} describes the error at order $k+1$ ($p+1$ in the notation of that work) in terms of a sum over norms of multicommutators of $k+1$ terms in the Hamiltonian.
It is easily seen that the error at each of the higher orders will also be given by similar multicommutator forms.
The method of proof of Theorem 4 in Ref.~\cite{Low2022trotter} is based on simplifying the multicommutator expressions from Lemma 2, and proceeds in exactly the same way for the higher orders.

Therefore, the sum over errors of all orders would give a total error of the form
\begin{equation}
    \frac{\norm{\tau}_1\norm{\nu}_{1,[\eta]}\eta}{\norm{\tau}_{1}+\norm{\nu}_{1,[\eta]}} t \times \sum_{j=k+1}^\infty \xi_j \left[ t \left(\norm{\tau}_{1}+\norm{\nu}_{1,[\eta]}\right)\right]^{j}.
\end{equation}
This sum would then give the non-asymptotic form for the error
\begin{equation}
    \frac{\norm{\tau}_1\norm{\nu}_{1,[\eta]}\eta}{\norm{\tau}_{1}+\norm{\nu}_{1,[\eta]}} t \times g\!\left( {t}\left(\norm{\tau}_{1}+\norm{\nu}_{1,[\eta]}\right)\right) \, ,
\end{equation}
for some function $g$.

Then following the same procedure to determine the threshold as before, we will have a total error given by
\begin{equation}
    \frac{\norm{\tau}_1\norm{\nu}_{1,[\eta]}\eta}{\norm{\tau}_{1}+\norm{\nu}_{1,[\eta]}} T \times g\!\left( {t}\left(\norm{\tau}_{1}+\norm{\nu}_{1,[\eta]}\right)\right) \, .
\end{equation}
That would imply that comparing two product formulae, the threshold will be obtained for
\begin{equation}\label{eq:ferm_equality_error}
    g_1\!\left( {t_1}\left(\norm{\tau}_{1}+\norm{\nu}_{1,[\eta]}\right)\right) = g_2\!\left( {t_2}\left(\norm{\tau}_{1}+\norm{\nu}_{1,[\eta]}\right)\right) \, .
\end{equation}
Using $t_2=t_1M_2/M_1$ and solving for $t_1$ as before, the $T/\epsilon$ threshold would then be given by substituting the solution for $t_1$ into
\begin{equation}\label{eq:fermTthresh}
    \frac{\norm{\tau}_1\norm{\nu}_{1,[\eta]}\eta}{\norm{\tau}_{1}+\norm{\nu}_{1,[\eta]}} \frac{T}{\epsilon} = 1 / g_1\!\left( {t_1}\left(\norm{\tau}_{1}+\norm{\nu}_{1,[\eta]}\right)\right) \, .
\end{equation}

The threshold can be calculated for Hamiltonians with a given value of $\norm{\tau}_{1}+\norm{\nu}_{1,[\eta]}$, then for some new Hamiltonian with different values of the norms, the threshold should still be for the same value of ${t_1}\big(\norm{\tau}_{1}+\norm{\nu}_{1,[\eta]}\big)$.
This implies that the right-hand-side of Eq.~\eqref{eq:fermTthresh} should be the same for the threshold with the new Hamiltonian.
As a result, we expect that the threshold should be in terms of the left-hand-side of Eq.~\eqref{eq:fermTthresh}, which is the same expression as found in the asymptotic regime.

In our numerical testing, the coefficients $\tau_{pq}$, $\nu_{pq}$ were chosen uniformly from the interval $[-1,1]$.
We computed $\omega$ for a selection of the product formulae with the best performance for random Hamiltonians.
In \cref{tab:factor6-chemistry-eig6}, \cref{tab:factor8-chemistry-eig6} and \cref{tab:factor10-chemistry-eig6} we give the computed result for $d=6$ orbitals, assuming half-filling of the orbitals.
Our numerics indicate that the error is roughly proportional to the bound in \cref{eq:Low_bound}, independent of the number of orbitals, though the constant factors are small. In \cref{sec:fermHam_d4} we give the constant factors for the case $d=4$, showing that the computed constant does not change much with a different $d$.

\begin{table}[tbh]
\centering
\begin{tabular}{|c|c|c|c|c|c|} 
\hline
label & $M$ & processing & reference                          & $\omega$   & $M\omega^{1/k}$  \\ 
\hline
PPBCM6m9 & $9$ & Y & $P_9 6$ Table 5 of \cite{Blanes2006processing} & $2.8\times 10^{-9}$ & $0.33$ \\

PPBCM6m6 & $13$ & Y & $P_{13}6$ in Table 6 of \cite{Blanes2006processing} & $1.2\times 10^{-9}$ & $0.42$ \\
BCE6m10 & $20$ & Y & $\psi_{10}^{[6]}$ Table 8 of \cite{blanes2024families} & $3.4\times 10^{-11}$ & $0.36$ \\
\hline
\end{tabular}
\caption{Comparison of constant factors $\omega$ for a selection of the lowest-error best product formulae for 6th order. We generate $1000$ random Hamiltonians with $d=6$ orbitals as in \cref{eq:fermion-ham} and compute the average $\omega$.}
\label{tab:factor6-chemistry-eig6}
\end{table}

\begin{table}[tbh]
\centering
\begin{tabular}{|c|c|c|c|c|c|} 
\hline
label & $M$ & processing & reference                     & $\omega$   & $M\omega^{1/k}$  \\ 
\hline
SS8s19 & 19 & N & Section 4.3 of \cite{Sofroniou2005integrators} &   $3.5\times 10^{-11}$                          &      $0.94$                                  \\
PP8s13 & 13 & Y & $P_{13}8$ in Table 6 of \cite{Blanes2006processing} &   $4.4\times 10^{-10}$                           &      $0.88$                                  \\
Y8m10 & 21 & N & \cref{tab:order8-solm10} (our new result) & $8.5\times 10^{-12}$  & $0.87$ \\
Y8m10b & 21 & N & \cref{tab:order8-solm10} (our new result) & $1.3\times 10^{-12}$ & $0.68$ \\
YP8m8 & 17 & Y & \cref{tab:process} (our new result) & $1.7\times 10^{-12}$ & $0.57$ \\
\hline
\end{tabular}
\caption{Comparison of constant factors $\omega$ for a selection of the lowest-error product formulae for 8th order. We generate $1000$ random Hamiltonians with $d=6$ orbitals as in \cref{eq:fermion-ham} and compute the average $\omega$.}
\label{tab:factor8-chemistry-eig6}
\end{table}

\begin{table}[tbh]
\centering
\begin{tabular}{|c|c|c|c|c|c|} 
\hline
label & $M$ & processing & reference  & $\omega$   & $M\omega^{1/k}$  \\ 
\hline
SS10s35 & $35$ & N & Section 4.4 of \cite{Sofroniou2005integrators} & $3.0\times 10^{-15}$ & $1.24$ \\
PP10s23 & $23$ & Y & $P_{23}10$ in Table 6 of \cite{Blanes2006processing} & $2.3\times 10^{-12}$ & $1.58$ \\
Y10m17 & $35$ & N & \cref{tab:order10-solm17} & $2.6\times 10^{-14}$ & $1.53$ \\
\hline
\end{tabular}
\caption{Comparison of constant factors $\omega$ for a selection of the lowest-error product formulae for 10th order. We generate $1000$ random Hamiltonians with $d=6$ orbitals as in \cref{eq:fermion-ham} and compute the average $\omega$.}
\label{tab:factor10-chemistry-eig6}
\end{table}

From the results in these tables, we see that the relative performance for the product formulae in unchanged from the results with the random matrices.
This is a useful independent verification that the assessment of the performance of the product formulae is relatively independent of the class of matrices.
In particular, we find that for 8th order our three new results all outperform the best product formulae in prior work (where we have compared to the two prior formulae with best performance).
In the case of 10th order the product formula SS10s35 of Ref.~\cite{Sofroniou2005integrators} still provides exceptional performance.

For calculating the thresholds based on the asymptotic formula in \cref{chemthresh}, we should compare our best 8th order formula YP8m8 to SS10s35 in the 10th order case, and PPBCM6m9 in the 6th order case.
Then we obtain a threshold of $3\times 10^{13}$ for 8th to 10th order, and $5 \times 10^5$ for 6th to 8th order, indicating about 9 orders of magnitude for $T/\epsilon$ where our 8th order product formula is optimal.

Similar to the case for general matrices, we expect that the threshold for 8th to 10th order is sufficiently large that the time steps are small and the asymptotic formulae are accurate, but the threshold between 6th and 8th order may be unreliable.
In that case we solve \cref{eq:ferm_equality_error} for $t_1$ with $t_2=t_1 M_2/M_2$, and find the threshold using \cref{eq:fermTthresh}.
We expect that the 6th order product formula PPBCM6m9 can give better performance for larger time steps, and indeed we find that the threshold is further increased to $1.7\times 10^7$. By using our 8th order kernel optimized for large time step in \cref{tab:processl} we can improve the threshold between PPBCM6m6 and our 8th order formula in the non-asymptotic case.
We then obtain a threshold of around $3\times 10^6$, which is an improvement by a factor of 5 times.

Next we compare these thresholds to the parameters for specific systems.
From Ref.~\cite{Low2022trotter} the norms can be expected to scale as
\begin{equation}
    \norm{\nu}_{1,[\eta]} = \mathcal{O}\left( \frac{\eta^{2/3}N^{1/3}}{\Omega^{1/3}} \right), \qquad \norm{\tau}_1 = \mathcal{O}\left(\frac{N^{2/3}}{\Omega^{2/3}} \right),
\end{equation}
where $N$ is the number of orbitals and $\Omega$ is the volume (denoted $n$ and $\omega$ in \cite{Low2022trotter}).
The constants for the scaling of these norms are derived in Ref.~\cite{rubin2023quantum} as
\begin{equation}
    \norm{\nu}_{1,[\eta]} \approx \pi^{1/3} (3/4)^{2/3}\frac{\eta^{2/3} N^{1/3}}{\Omega^{1/3}}, \qquad \norm{\tau}_1 \approx\frac{3\pi^2N^{2/3}}{2\Omega^{2/3}} .
\end{equation}
For an initial order of magnitude estimate, the chemical accuracy required for the phase estimation is about $0.001$ Hartree.
As discussed above, for phase estimation we replace $T/\epsilon$ in the threshold with $1/\epsilon_H$.
With estimates $\eta\approx 100$ and $N/\Omega$ in the range $10^3$ to $10^9$ (corresponding to $\Delta$ between $10^{-1}a_0$ and $10^{-3}a_0$ in Ref.~\cite{SuPRXQuantum21}), the left-hand-side (LHS) of Eq.~\eqref{chemthresh} would be in the range $2.2\times 10^{7}$ to $2.6\times 10^{9}$.
That is far below the threshold between 8th and 10th order, demonstrating that our 8th order should be preferred to 10th order.
It is also above the threshold between 6th and 8th order, indicating that 8th order provides the best performance for simulations of this type.
For smaller simulations 6th order will be optimal, but it may be expected that the small values of $N/\Omega$ would not provide a sufficiently accurate representation of the quantum system.

For comparison, we also consider the Alpha + Hydrogen system from Ref.~\cite{rubin2023quantum}. The parameters $\eta$, $N$ and $\Omega$ are given in Table 1 of that work as 218, $53^3$, and 2420, respectively.
Some of the most challenging values of $T,\epsilon$ are given in Figure 4 of that work, with for example $T=40$ and $\epsilon=10^{-4}$.
With these parameters the LHS of Eq.~\eqref{chemthresh} is approximately $8.6\times 10^9$, which is below the threshold (between 8th and 10th order) by a few orders of magnitude.
The smallest values of $T/\epsilon$ from Ref.~\cite{rubin2023quantum} are obtained with $T=10$ and $\epsilon=0.01$ (see Table IV).
That results in a value for the LHS of Eq.~\eqref{chemthresh} of $2.2\times 10^7$, which is above the threshold between 6th and 8th order, so our 8th order product formula would still be optimal.
This indicates that the new 8th order product formula found here should be optimal for the full range of parameter values used in Ref.~\cite{rubin2023quantum}.
Nonetheless, for smaller values of $T/\epsilon$ it may be advantageous to further customise the product formula to provide better performance at larger time steps, as described above to improve the threshold between 6th order and 8th order.

Note also that the fact that the threshold here is different than for random Hamiltonians means that the threshold will depend on the class of Hamiltonians.
Selecting a different type of Hamiltonians can change the threshold.
Similarly, it will be expected that choosing a particular class of fermionic Hamiltonians (such as those arising from plane waves) rather than random fermionic Hamiltonians can give a different threshold.
Nevertheless, changing the class of Hamiltonians still resulted in an extremely large threshold for 10th order product formulae to be optimal.
This indicates that the 8th order formula will be optimal for realistic simulations regardless of the class of Hamiltonians.

\section{Discussion}\label{sec:discuss}

We have found 8th order product formulae with significantly lower error than in prior work.
For the case of processed product formulae, our best solution is about 300 times more accurate than that in prior work, and in the case of product formulae without processing our solution is about 100 times more accurate.
The size of the improvement is estimated based on testing with random $64\times 64$ matrices.
We have also tested on a range of different matrix sizes, fermionic Hamiltonians of different sizes, and a range of sizes of the transverse-field Ising model, and find consistent improvements with well over an order of magnitude in all cases.

The size of the improvement is surprising, especially considering the best prior results are from 2005 (without processing) and 2006 (with processing).
To achieve this improvement we have used distinct techniques from prior work.
The processed product formulae from Ref.~\cite{Blanes2006processing} set multiple coefficients equal to one another. Specifically, in the \( P_{19}8 \) formula, six of the coefficients are identical, following a heuristic suggested by McLachlan \cite{McLachlan2002}. However, our results suggest that using distinct values for each coefficient yields better outcomes.
The advantage of taking longer product formulae with distinct coefficients is that there are significantly more parameters than constraints, giving more freedom for optimisation of the product formula.
The solutions which we find give the best performance have no clear pattern to the coefficients, indicating that assuming these patterns is too restrictive.

The way that product formulae were evaluated in Ref.~\cite{Blanes2006processing} was based on the $\ell^p$ norms of the coefficients for $p=1$ and $p=k+1$.
Those norms are both smaller for $P_{19}8$ given in Ref.~\cite{Blanes2006processing} than our lower-error solution.
Similarly, these norms are both smaller for the solutions from Ref.~\cite{Sofroniou2005integrators} than our non-processed product formula.
Also, Ref.~\cite{alberdi2019algorithm} found 10th order integrators with 1-norm lower than those in Ref.~\cite{Sofroniou2005integrators}, but in our testing their error is larger.
In all these cases the solutions with lower error have larger norms.
The $\ell^{k+1}$ norm corresponds to just one of the $\mathcal{O}(t^{k+1})$ error terms, and a product formula that minimises that term while giving large values for the others will not perform well.
We find that the strategy of selecting product formulae with low error for examples can yield much better results than attempting to minimise norms.

In this work we compare the eigenvalue error, rather than the spectral-norm error as in prior work, because it is the eigenvalue error in a single step which dominates the error over longer time evolutions, with basis error cancelling.
In many applications of quantum computing the goal is to estimate eigenvalues, so it is solely the eigenvalue error which should be considered.
For most product formulae tested, the spectral norm of the error (for a single step) is substantially larger than the eigenvalue error, meaning that it can be an unrealistic measure of accuracy.
This also suggests that techniques for reducing error based on randomisation of the term order \cite{Childs2019fasterquantum} or randomised corrections \cite{Cho2024} may provide less improvement than anticipated, because if successive steps are not identical then the basis error will not cancel.

We have shown how to fairly compare the performance of product formulae with different numbers of factors and different orders.
The optimal order of integrator to use depends primarily on the ratio $T/\epsilon$, the ratio of the total evolution time to the allowed error.
This is because, as $T/\epsilon$ becomes larger, the error allowed for each individual time step becomes smaller, and the time steps must become shorter to make the error sufficiently small.
That results in a larger overall number of time steps, and so a larger complexity.
The higher-order integrators better reduce the error with shorter time steps, and so for large $T/\epsilon$ that improvement more than compensates for the larger number of exponentials for the higher-order product formulae.
For eigenvalue estimation there is an equivalent threshold in $1/\epsilon_H$, in terms of the allowed error in the eigenvalue.

The threshold for 10th order product formulae to outperform 8th order product formulae is very large, so 10th order would only be optimal for simulations well beyond the scale considered for quantum algorithms.
Similarly, it would not be useful to consider 12th or higher order product formulae for the purpose of Hamiltonian simulation, as the threshold would be even larger.
For large quantum simulations that are still reasonable to perform on a quantum computer, such as those considered in Refs.~\cite{rubin2023quantum} or \cite{SuPRXQuantum21}, the best performance is provided by our 8th order product formulae (out of the extensive set of product formulae tested).
For more moderately sized simulations 6th order will be optimal, and in our testing we find that the product formula $S_9 6$ of Ref.~\cite{Blanes2006processing} provides excellent performance.
For smaller simulations, as in NISQ experiments, 4th order product formulae would be ideal, with $S_6$ of Ref.~\cite{BLANES2002313} and $P_6 4$ of \cite{Blanes2006processing} providing the best performance.

Care is required in determining the thresholds between the orders.
Using the formula based on the asymptotic scaling can yield estimates of the $T/\epsilon$ threshold corresponding to large time steps where the asymptotic scaling no longer holds.
When properly accounting for the error for a given time step, the $T/\epsilon$ threshold can change by many orders of magnitude.
This can also be a problem more generally in estimating the complexity of an algorithm.
Estimating the number of time steps based on the asymptotic expression may correspond to unrealistically large time steps, so the asymptotic formula is inaccurate.
We find that significant performance gains can be obtained by fine-tuning a product formula for larger time steps.

A potential topic for future work is further customising product formulae for large time steps.
We have customised the 8th order product formula to improve the threshold, but it would also be possible to customise the 6th and 4th order product formulae.
It would also be possible to customise the product formulae for larger time steps with the fermionic Hamiltonians.
This approach could be used to provide better performance for small-scale simulations suitable for NISQ devices.
It could also potentially extend the parameter range where 6th order product formulae provide best performance to larger simulations on fault tolerant quantum computers.

\section*{Acknowledgements}
MESM was supported by the ARC Centre of Excellence for Quantum Computation and Communication Technology (CQC2T), project number CE170100012 and a scholarship top-up and extension from the Sydney Quantum Academy. MESM acknowledges support from the U.S. Department of Defense through a QuICS Hartree Fellowship. MESM and YRS were supported by the Defense Advanced Research Projects Agency under Contract No.~HR001122C0074. Any opinions, findings and conclusions or recommendations expressed in this material are those of the author(s) and do not necessarily reflect the views of the Defense Advanced Research Projects Agency. DKB acknowledges funding by the Australian Research Council (project numbers FT190100106, DP210101367, CE170100009).
DWB worked on this project under a sponsored research agreement with Google Quantum AI. DWB is also supported by Australian Research Council Discovery Projects DP210101367 and DP220101602.

\bibliographystyle{apsrev4-2}

\begin{thebibliography}{55}%
\makeatletter
\providecommand \@ifxundefined [1]{%
 \@ifx{#1\undefined}
}%
\providecommand \@ifnum [1]{%
 \ifnum #1\expandafter \@firstoftwo
 \else \expandafter \@secondoftwo
 \fi
}%
\providecommand \@ifx [1]{%
 \ifx #1\expandafter \@firstoftwo
 \else \expandafter \@secondoftwo
 \fi
}%
\providecommand \natexlab [1]{#1}%
\providecommand \enquote  [1]{``#1''}%
\providecommand \bibnamefont  [1]{#1}%
\providecommand \bibfnamefont [1]{#1}%
\providecommand \citenamefont [1]{#1}%
\providecommand \href@noop [0]{\@secondoftwo}%
\providecommand \href [0]{\begingroup \@sanitize@url \@href}%
\providecommand \@href[1]{\@@startlink{#1}\@@href}%
\providecommand \@@href[1]{\endgroup#1\@@endlink}%
\providecommand \@sanitize@url [0]{\catcode `\\12\catcode `\$12\catcode `\&12\catcode `\#12\catcode `\^12\catcode `\_12\catcode `\%12\relax}%
\providecommand \@@startlink[1]{}%
\providecommand \@@endlink[0]{}%
\providecommand \url  [0]{\begingroup\@sanitize@url \@url }%
\providecommand \@url [1]{\endgroup\@href {#1}{\urlprefix }}%
\providecommand \urlprefix  [0]{URL }%
\providecommand \Eprint [0]{\href }%
\providecommand \doibase [0]{https://doi.org/}%
\providecommand \selectlanguage [0]{\@gobble}%
\providecommand \bibinfo  [0]{\@secondoftwo}%
\providecommand \bibfield  [0]{\@secondoftwo}%
\providecommand \translation [1]{[#1]}%
\providecommand \BibitemOpen [0]{}%
\providecommand \bibitemStop [0]{}%
\providecommand \bibitemNoStop [0]{.\EOS\space}%
\providecommand \EOS [0]{\spacefactor3000\relax}%
\providecommand \BibitemShut  [1]{\csname bibitem#1\endcsname}%
\let\auto@bib@innerbib\@empty
\bibitem [{\citenamefont {Creutz}\ and\ \citenamefont {Gocksch}(1989)}]{Creutz1989}%
  \BibitemOpen
  \bibfield  {author} {\bibinfo {author} {\bibfnamefont {M.}~\bibnamefont {Creutz}}\ and\ \bibinfo {author} {\bibfnamefont {A.}~\bibnamefont {Gocksch}},\ }\href {https://doi.org/10.1103/PhysRevLett.63.9} {\bibfield  {journal} {\bibinfo  {journal} {Physical Review Letters}\ }\textbf {\bibinfo {volume} {63}},\ \bibinfo {pages} {9} (\bibinfo {year} {1989})}\BibitemShut {NoStop}%
\bibitem [{\citenamefont {Suzuki}(1990)}]{Suzuki1990}%
  \BibitemOpen
  \bibfield  {author} {\bibinfo {author} {\bibfnamefont {M.}~\bibnamefont {Suzuki}},\ }\href {https://doi.org/10.1016/0375-9601(90)90962-N} {\bibfield  {journal} {\bibinfo  {journal} {Physics Letters A}\ }\textbf {\bibinfo {volume} {146}},\ \bibinfo {pages} {319} (\bibinfo {year} {1990})}\BibitemShut {NoStop}%
\bibitem [{\citenamefont {Yoshida}(1990)}]{Yoshida1990}%
  \BibitemOpen
  \bibfield  {author} {\bibinfo {author} {\bibfnamefont {H.}~\bibnamefont {Yoshida}},\ }\href {https://doi.org/10.1016/0375-9601(90)90092-3} {\bibfield  {journal} {\bibinfo  {journal} {Physics Letters A}\ }\textbf {\bibinfo {volume} {150}},\ \bibinfo {pages} {262} (\bibinfo {year} {1990})}\BibitemShut {NoStop}%
\bibitem [{\citenamefont {Suzuki}(1991)}]{Suzuki5}%
  \BibitemOpen
  \bibfield  {author} {\bibinfo {author} {\bibfnamefont {M.}~\bibnamefont {Suzuki}},\ }\href {https://doi.org/10.1063/1.529425} {\bibfield  {journal} {\bibinfo  {journal} {Journal of Mathematical Physics}\ }\textbf {\bibinfo {volume} {32}},\ \bibinfo {pages} {400} (\bibinfo {year} {1991})}\BibitemShut {NoStop}%
\bibitem [{\citenamefont {Berry}\ \emph {et~al.}(2007)\citenamefont {Berry}, \citenamefont {Ahokas}, \citenamefont {Cleve},\ and\ \citenamefont {Sanders}}]{Berry2007}%
  \BibitemOpen
  \bibfield  {author} {\bibinfo {author} {\bibfnamefont {D.~W.}\ \bibnamefont {Berry}}, \bibinfo {author} {\bibfnamefont {G.}~\bibnamefont {Ahokas}}, \bibinfo {author} {\bibfnamefont {R.}~\bibnamefont {Cleve}},\ and\ \bibinfo {author} {\bibfnamefont {B.~C.}\ \bibnamefont {Sanders}},\ }\href {https://doi.org/10.1007/s00220-006-0150-x} {\bibfield  {journal} {\bibinfo  {journal} {Communications in Mathematical Physics}\ }\textbf {\bibinfo {volume} {270}},\ \bibinfo {pages} {359} (\bibinfo {year} {2007})}\BibitemShut {NoStop}%
\bibitem [{\citenamefont {Lloyd}(1996)}]{Lloyd1996universal}%
  \BibitemOpen
  \bibfield  {author} {\bibinfo {author} {\bibfnamefont {S.}~\bibnamefont {Lloyd}},\ }\href {https://doi.org/10.1126/science.273.5278.1073} {\bibfield  {journal} {\bibinfo  {journal} {Science}\ }\textbf {\bibinfo {volume} {273}},\ \bibinfo {pages} {1073} (\bibinfo {year} {1996})}\BibitemShut {NoStop}%
\bibitem [{\citenamefont {Aharonov}\ and\ \citenamefont {Ta-Shma}(2003)}]{Aharonov2003adiabatic}%
  \BibitemOpen
  \bibfield  {author} {\bibinfo {author} {\bibfnamefont {D.}~\bibnamefont {Aharonov}}\ and\ \bibinfo {author} {\bibfnamefont {A.}~\bibnamefont {Ta-Shma}},\ }in\ \href {https://doi.org/10.1145/780542.780546} {\emph {\bibinfo {booktitle} {Proceedings of the Thirty-Fifth Annual ACM Symposium on Theory of Computing}}},\ \bibinfo {series and number} {STOC '03}\ (\bibinfo  {publisher} {Association for Computing Machinery},\ \bibinfo {address} {New York, NY, USA},\ \bibinfo {year} {2003})\ p.\ \bibinfo {pages} {20–29}\BibitemShut {NoStop}%
\bibitem [{\citenamefont {Berry}\ and\ \citenamefont {Childs}(2012)}]{Berry2012blackbox}%
  \BibitemOpen
  \bibfield  {author} {\bibinfo {author} {\bibfnamefont {D.~W.}\ \bibnamefont {Berry}}\ and\ \bibinfo {author} {\bibfnamefont {A.~M.}\ \bibnamefont {Childs}},\ }\href {https://doi.org/10.26421/QIC12.1-2-4} {\bibfield  {journal} {\bibinfo  {journal} {Quantum Information and Computation}\ }\textbf {\bibinfo {volume} {12}},\ \bibinfo {pages} {29–62} (\bibinfo {year} {2012})}\BibitemShut {NoStop}%
\bibitem [{\citenamefont {Berry}\ \emph {et~al.}(2014)\citenamefont {Berry}, \citenamefont {Childs}, \citenamefont {Cleve}, \citenamefont {Kothari},\ and\ \citenamefont {Somma}}]{Berry2014exponential}%
  \BibitemOpen
  \bibfield  {author} {\bibinfo {author} {\bibfnamefont {D.~W.}\ \bibnamefont {Berry}}, \bibinfo {author} {\bibfnamefont {A.~M.}\ \bibnamefont {Childs}}, \bibinfo {author} {\bibfnamefont {R.}~\bibnamefont {Cleve}}, \bibinfo {author} {\bibfnamefont {R.}~\bibnamefont {Kothari}},\ and\ \bibinfo {author} {\bibfnamefont {R.~D.}\ \bibnamefont {Somma}},\ }in\ \href {https://doi.org/10.1145/2591796.2591854} {\emph {\bibinfo {booktitle} {Proceedings of the Forty-Sixth Annual ACM Symposium on Theory of Computing}}},\ \bibinfo {series and number} {STOC '14}\ (\bibinfo  {publisher} {Association for Computing Machinery},\ \bibinfo {address} {New York, NY, USA},\ \bibinfo {year} {2014})\ p.\ \bibinfo {pages} {283–292}\BibitemShut {NoStop}%
\bibitem [{\citenamefont {Berry}\ \emph {et~al.}(2015{\natexlab{a}})\citenamefont {Berry}, \citenamefont {Childs},\ and\ \citenamefont {Kothari}}]{Berry2015nearlyoptimal}%
  \BibitemOpen
  \bibfield  {author} {\bibinfo {author} {\bibfnamefont {D.~W.}\ \bibnamefont {Berry}}, \bibinfo {author} {\bibfnamefont {A.~M.}\ \bibnamefont {Childs}},\ and\ \bibinfo {author} {\bibfnamefont {R.}~\bibnamefont {Kothari}},\ }in\ \href {https://doi.org/10.1109/FOCS.2015.54} {\emph {\bibinfo {booktitle} {2015 IEEE 56th Annual Symposium on Foundations of Computer Science}}}\ (\bibinfo {year} {2015})\ pp.\ \bibinfo {pages} {792--809}\BibitemShut {NoStop}%
\bibitem [{\citenamefont {Low}(2019)}]{Low2019spectralnorm}%
  \BibitemOpen
  \bibfield  {author} {\bibinfo {author} {\bibfnamefont {G.~H.}\ \bibnamefont {Low}},\ }in\ \href {https://doi.org/10.1145/3313276.3316386} {\emph {\bibinfo {booktitle} {Proceedings of the 51st Annual ACM SIGACT Symposium on Theory of Computing}}},\ \bibinfo {series and number} {STOC 2019}\ (\bibinfo  {publisher} {Association for Computing Machinery},\ \bibinfo {address} {New York, NY, USA},\ \bibinfo {year} {2019})\ p.\ \bibinfo {pages} {491–502}\BibitemShut {NoStop}%
\bibitem [{\citenamefont {Childs}\ \emph {et~al.}(2021)\citenamefont {Childs}, \citenamefont {Su}, \citenamefont {Tran}, \citenamefont {Wiebe},\ and\ \citenamefont {Zhu}}]{Childs2022trotter}%
  \BibitemOpen
  \bibfield  {author} {\bibinfo {author} {\bibfnamefont {A.~M.}\ \bibnamefont {Childs}}, \bibinfo {author} {\bibfnamefont {Y.}~\bibnamefont {Su}}, \bibinfo {author} {\bibfnamefont {M.~C.}\ \bibnamefont {Tran}}, \bibinfo {author} {\bibfnamefont {N.}~\bibnamefont {Wiebe}},\ and\ \bibinfo {author} {\bibfnamefont {S.}~\bibnamefont {Zhu}},\ }\href {https://doi.org/10.1103/PhysRevX.11.011020} {\bibfield  {journal} {\bibinfo  {journal} {Physical Review X}\ }\textbf {\bibinfo {volume} {11}},\ \bibinfo {pages} {011020} (\bibinfo {year} {2021})}\BibitemShut {NoStop}%
\bibitem [{\citenamefont {Childs}\ and\ \citenamefont {Wiebe}(2012)}]{Childs2012LCU}%
  \BibitemOpen
  \bibfield  {author} {\bibinfo {author} {\bibfnamefont {A.~M.}\ \bibnamefont {Childs}}\ and\ \bibinfo {author} {\bibfnamefont {N.}~\bibnamefont {Wiebe}},\ }\href {https://doi.org/10.26421/QIC12.11-12-1} {\bibfield  {journal} {\bibinfo  {journal} {Quantum Information and Computation}\ }\textbf {\bibinfo {volume} {12}},\ \bibinfo {pages} {901–924} (\bibinfo {year} {2012})}\BibitemShut {NoStop}%
\bibitem [{\citenamefont {Berry}\ \emph {et~al.}(2015{\natexlab{b}})\citenamefont {Berry}, \citenamefont {Childs}, \citenamefont {Cleve}, \citenamefont {Kothari},\ and\ \citenamefont {Somma}}]{Berry2015}%
  \BibitemOpen
  \bibfield  {author} {\bibinfo {author} {\bibfnamefont {D.~W.}\ \bibnamefont {Berry}}, \bibinfo {author} {\bibfnamefont {A.~M.}\ \bibnamefont {Childs}}, \bibinfo {author} {\bibfnamefont {R.}~\bibnamefont {Cleve}}, \bibinfo {author} {\bibfnamefont {R.}~\bibnamefont {Kothari}},\ and\ \bibinfo {author} {\bibfnamefont {R.~D.}\ \bibnamefont {Somma}},\ }\href {https://doi.org/10.1103/PhysRevLett.114.090502} {\bibfield  {journal} {\bibinfo  {journal} {Physical Review Letters}\ }\textbf {\bibinfo {volume} {114}},\ \bibinfo {pages} {90502} (\bibinfo {year} {2015}{\natexlab{b}})}\BibitemShut {NoStop}%
\bibitem [{\citenamefont {Low}\ and\ \citenamefont {Chuang}(2017)}]{Low2017}%
  \BibitemOpen
  \bibfield  {author} {\bibinfo {author} {\bibfnamefont {G.~H.}\ \bibnamefont {Low}}\ and\ \bibinfo {author} {\bibfnamefont {I.~L.}\ \bibnamefont {Chuang}},\ }\href {https://doi.org/10.1103/PhysRevLett.118.010501} {\bibfield  {journal} {\bibinfo  {journal} {Physical Review Letters}\ }\textbf {\bibinfo {volume} {118}},\ \bibinfo {pages} {010501} (\bibinfo {year} {2017})}\BibitemShut {NoStop}%
\bibitem [{\citenamefont {Babbush}\ \emph {et~al.}(2015)\citenamefont {Babbush}, \citenamefont {McClean}, \citenamefont {Wecker}, \citenamefont {Aspuru-Guzik},\ and\ \citenamefont {Wiebe}}]{Babbush2015chemical}%
  \BibitemOpen
  \bibfield  {author} {\bibinfo {author} {\bibfnamefont {R.}~\bibnamefont {Babbush}}, \bibinfo {author} {\bibfnamefont {J.}~\bibnamefont {McClean}}, \bibinfo {author} {\bibfnamefont {D.}~\bibnamefont {Wecker}}, \bibinfo {author} {\bibfnamefont {A.}~\bibnamefont {Aspuru-Guzik}},\ and\ \bibinfo {author} {\bibfnamefont {N.}~\bibnamefont {Wiebe}},\ }\href {https://doi.org/10.1103/PhysRevA.91.022311} {\bibfield  {journal} {\bibinfo  {journal} {Physical Review A}\ }\textbf {\bibinfo {volume} {91}},\ \bibinfo {pages} {022311} (\bibinfo {year} {2015})}\BibitemShut {NoStop}%
\bibitem [{\citenamefont {Childs}\ \emph {et~al.}(2018)\citenamefont {Childs}, \citenamefont {Maslov}, \citenamefont {Nam}, \citenamefont {Ross},\ and\ \citenamefont {Su}}]{Childs2018speedup}%
  \BibitemOpen
  \bibfield  {author} {\bibinfo {author} {\bibfnamefont {A.~M.}\ \bibnamefont {Childs}}, \bibinfo {author} {\bibfnamefont {D.}~\bibnamefont {Maslov}}, \bibinfo {author} {\bibfnamefont {Y.}~\bibnamefont {Nam}}, \bibinfo {author} {\bibfnamefont {N.~J.}\ \bibnamefont {Ross}},\ and\ \bibinfo {author} {\bibfnamefont {Y.}~\bibnamefont {Su}},\ }\href {https://doi.org/10.1073/pnas.1801723115} {\bibfield  {journal} {\bibinfo  {journal} {Proceedings of the National Academy of Sciences}\ }\textbf {\bibinfo {volume} {115}},\ \bibinfo {pages} {9456} (\bibinfo {year} {2018})}\BibitemShut {NoStop}%
\bibitem [{\citenamefont {Su}\ \emph {et~al.}(2021{\natexlab{a}})\citenamefont {Su}, \citenamefont {Huang},\ and\ \citenamefont {Campbell}}]{Su2021nearlytight}%
  \BibitemOpen
  \bibfield  {author} {\bibinfo {author} {\bibfnamefont {Y.}~\bibnamefont {Su}}, \bibinfo {author} {\bibfnamefont {H.-Y.}\ \bibnamefont {Huang}},\ and\ \bibinfo {author} {\bibfnamefont {E.~T.}\ \bibnamefont {Campbell}},\ }\href {https://doi.org/10.22331/q-2021-07-05-495} {\bibfield  {journal} {\bibinfo  {journal} {{Quantum}}\ }\textbf {\bibinfo {volume} {5}},\ \bibinfo {pages} {495} (\bibinfo {year} {2021}{\natexlab{a}})}\BibitemShut {NoStop}%
\bibitem [{\citenamefont {Zhang}(2012)}]{Zhang2012randomized}%
  \BibitemOpen
  \bibfield  {author} {\bibinfo {author} {\bibfnamefont {C.}~\bibnamefont {Zhang}},\ }in\ \href {https://doi.org/10.1007/978-3-642-27440-4_42} {\emph {\bibinfo {booktitle} {Monte Carlo and Quasi-Monte Carlo Methods 2010}}},\ \bibinfo {editor} {edited by\ \bibinfo {editor} {\bibfnamefont {L.}~\bibnamefont {Plaskota}}\ and\ \bibinfo {editor} {\bibfnamefont {H.}~\bibnamefont {Wo{\'{z}}niakowski}}}\ (\bibinfo  {publisher} {Springer Berlin Heidelberg},\ \bibinfo {address} {Berlin, Heidelberg},\ \bibinfo {year} {2012})\ pp.\ \bibinfo {pages} {709--719}\BibitemShut {NoStop}%
\bibitem [{\citenamefont {Campbell}(2019)}]{Campbell2019random}%
  \BibitemOpen
  \bibfield  {author} {\bibinfo {author} {\bibfnamefont {E.}~\bibnamefont {Campbell}},\ }\href {https://doi.org/10.1103/PhysRevLett.123.070503} {\bibfield  {journal} {\bibinfo  {journal} {Physical Review Letters}\ }\textbf {\bibinfo {volume} {123}},\ \bibinfo {pages} {070503} (\bibinfo {year} {2019})}\BibitemShut {NoStop}%
\bibitem [{\citenamefont {Childs}\ \emph {et~al.}(2019)\citenamefont {Childs}, \citenamefont {Ostrander},\ and\ \citenamefont {Su}}]{Childs2019fasterquantum}%
  \BibitemOpen
  \bibfield  {author} {\bibinfo {author} {\bibfnamefont {A.~M.}\ \bibnamefont {Childs}}, \bibinfo {author} {\bibfnamefont {A.}~\bibnamefont {Ostrander}},\ and\ \bibinfo {author} {\bibfnamefont {Y.}~\bibnamefont {Su}},\ }\href {https://doi.org/10.22331/q-2019-09-02-182} {\bibfield  {journal} {\bibinfo  {journal} {{Quantum}}\ }\textbf {\bibinfo {volume} {3}},\ \bibinfo {pages} {182} (\bibinfo {year} {2019})}\BibitemShut {NoStop}%
\bibitem [{\citenamefont {Calvo}\ and\ \citenamefont {Sanz-Serna}(1993)}]{Calvo93}%
  \BibitemOpen
  \bibfield  {author} {\bibinfo {author} {\bibfnamefont {M.~P.}\ \bibnamefont {Calvo}}\ and\ \bibinfo {author} {\bibfnamefont {J.~M.}\ \bibnamefont {Sanz-Serna}},\ }\href {https://doi.org/10.1137/0914073} {\bibfield  {journal} {\bibinfo  {journal} {SIAM Journal on Scientific Computing}\ }\textbf {\bibinfo {volume} {14}},\ \bibinfo {pages} {1237} (\bibinfo {year} {1993})}\BibitemShut {NoStop}%
\bibitem [{\citenamefont {Suzuki}\ and\ \citenamefont {Umeno}(1993)}]{Suzuki93}%
  \BibitemOpen
  \bibfield  {author} {\bibinfo {author} {\bibfnamefont {M.}~\bibnamefont {Suzuki}}\ and\ \bibinfo {author} {\bibfnamefont {K.}~\bibnamefont {Umeno}},\ }in\ \href@noop {} {\emph {\bibinfo {booktitle} {Computer Simulation Studies in Condensed-Matter Physics VI}}},\ \bibinfo {editor} {edited by\ \bibinfo {editor} {\bibfnamefont {D.~P.}\ \bibnamefont {Landau}}, \bibinfo {editor} {\bibfnamefont {K.~K.}\ \bibnamefont {Mon}},\ and\ \bibinfo {editor} {\bibfnamefont {H.-B.}\ \bibnamefont {Sch{\"u}ttler}}}\ (\bibinfo  {publisher} {Springer Berlin Heidelberg},\ \bibinfo {address} {Berlin, Heidelberg},\ \bibinfo {year} {1993})\ pp.\ \bibinfo {pages} {74--86}\BibitemShut {NoStop}%
\bibitem [{\citenamefont {McLachlan}(1995)}]{McLachlan95}%
  \BibitemOpen
  \bibfield  {author} {\bibinfo {author} {\bibfnamefont {R.~I.}\ \bibnamefont {McLachlan}},\ }\href {https://doi.org/10.1137/0916010} {\bibfield  {journal} {\bibinfo  {journal} {SIAM Journal on Scientific Computing}\ }\textbf {\bibinfo {volume} {16}},\ \bibinfo {pages} {151} (\bibinfo {year} {1995})}\BibitemShut {NoStop}%
\bibitem [{\citenamefont {Kahan}\ and\ \citenamefont {Li}(1997)}]{Kahan1997integrators}%
  \BibitemOpen
  \bibfield  {author} {\bibinfo {author} {\bibfnamefont {W.}~\bibnamefont {Kahan}}\ and\ \bibinfo {author} {\bibfnamefont {R.-C.}\ \bibnamefont {Li}},\ }\href {https://doi.org/10.1090/S0025-5718-97-00873-9} {\bibfield  {journal} {\bibinfo  {journal} {Mathematics of Computation}\ }\textbf {\bibinfo {volume} {66}},\ \bibinfo {pages} {1089–1099} (\bibinfo {year} {1997})}\BibitemShut {NoStop}%
\bibitem [{\citenamefont {Tsitouras}(1999)}]{Tsitouras1999}%
  \BibitemOpen
  \bibfield  {author} {\bibinfo {author} {\bibfnamefont {C.}~\bibnamefont {Tsitouras}},\ }\href {https://doi.org/10.1023/a:1008346516048} {\bibfield  {journal} {\bibinfo  {journal} {Celestial Mechanics and Dynamical Astronomy}\ }\textbf {\bibinfo {volume} {74}},\ \bibinfo {pages} {223–230} (\bibinfo {year} {1999})}\BibitemShut {NoStop}%
\bibitem [{\citenamefont {McLachlan}(2002)}]{McLachlan2002}%
  \BibitemOpen
  \bibfield  {author} {\bibinfo {author} {\bibfnamefont {R.~I.}\ \bibnamefont {McLachlan}},\ }\href {https://doi.org/10.1023/a:1021195019574} {\bibfield  {journal} {\bibinfo  {journal} {Numerical Algorithms}\ }\textbf {\bibinfo {volume} {31}},\ \bibinfo {pages} {233–246} (\bibinfo {year} {2002})}\BibitemShut {NoStop}%
\bibitem [{\citenamefont {Hairer}\ \emph {et~al.}(2006)\citenamefont {Hairer}, \citenamefont {Wanner},\ and\ \citenamefont {Lubich}}]{Hairer2006}%
  \BibitemOpen
  \bibfield  {author} {\bibinfo {author} {\bibfnamefont {E.}~\bibnamefont {Hairer}}, \bibinfo {author} {\bibfnamefont {G.}~\bibnamefont {Wanner}},\ and\ \bibinfo {author} {\bibfnamefont {C.}~\bibnamefont {Lubich}},\ }\href {https://doi.org/10.1007/3-540-30666-8} {\emph {\bibinfo {title} {Geometric Numerical Integration: Structure-Preserving Algorithms for Ordinary Differential Equations}}}\ (\bibinfo  {publisher} {Springer-Verlag},\ \bibinfo {year} {2006})\BibitemShut {NoStop}%
\bibitem [{\citenamefont {Sofroniou}\ and\ \citenamefont {Spaletta}(2005)}]{Sofroniou2005integrators}%
  \BibitemOpen
  \bibfield  {author} {\bibinfo {author} {\bibfnamefont {M.}~\bibnamefont {Sofroniou}}\ and\ \bibinfo {author} {\bibfnamefont {G.}~\bibnamefont {Spaletta}},\ }\href {https://doi.org/10.1080/10556780500140664} {\bibfield  {journal} {\bibinfo  {journal} {Optimization Methods and Software}\ }\textbf {\bibinfo {volume} {20}},\ \bibinfo {pages} {597} (\bibinfo {year} {2005})}\BibitemShut {NoStop}%
\bibitem [{\citenamefont {{Blanes}}\ \emph {et~al.}(2008)\citenamefont {{Blanes}}, \citenamefont {{Casas}},\ and\ \citenamefont {{Murua}}}]{Blanes2008summary}%
  \BibitemOpen
  \bibfield  {author} {\bibinfo {author} {\bibfnamefont {S.}~\bibnamefont {{Blanes}}}, \bibinfo {author} {\bibfnamefont {F.}~\bibnamefont {{Casas}}},\ and\ \bibinfo {author} {\bibfnamefont {A.}~\bibnamefont {{Murua}}},\ }\href {https://www.sema.org.es/Documentos/Fotos/1/0/2/5/pagina_fichero_1025-ficheros-1644772684-01752300-68691.pdf} {\bibfield  {journal} {\bibinfo  {journal} {Boletin de la Sociedad Espanola de Matematica Aplicada}\ }\textbf {\bibinfo {volume} {45}},\ \bibinfo {pages} {89} (\bibinfo {year} {2008})}\BibitemShut {NoStop}%
\bibitem [{\citenamefont {Blanes}\ \emph {et~al.}(2006)\citenamefont {Blanes}, \citenamefont {Casas},\ and\ \citenamefont {Murua}}]{Blanes2006processing}%
  \BibitemOpen
  \bibfield  {author} {\bibinfo {author} {\bibfnamefont {S.}~\bibnamefont {Blanes}}, \bibinfo {author} {\bibfnamefont {F.}~\bibnamefont {Casas}},\ and\ \bibinfo {author} {\bibfnamefont {A.}~\bibnamefont {Murua}},\ }\href {https://doi.org/10.1137/030601223} {\bibfield  {journal} {\bibinfo  {journal} {SIAM Journal on Scientific Computing}\ }\textbf {\bibinfo {volume} {27}},\ \bibinfo {pages} {1817} (\bibinfo {year} {2006})}\BibitemShut {NoStop}%
\bibitem [{\citenamefont {Blanes}\ \emph {et~al.}(2013)\citenamefont {Blanes}, \citenamefont {Casas}, \citenamefont {Chartier},\ and\ \citenamefont {Murua}}]{Blanes2013highorder}%
  \BibitemOpen
  \bibfield  {author} {\bibinfo {author} {\bibfnamefont {S.}~\bibnamefont {Blanes}}, \bibinfo {author} {\bibfnamefont {F.}~\bibnamefont {Casas}}, \bibinfo {author} {\bibfnamefont {P.}~\bibnamefont {Chartier}},\ and\ \bibinfo {author} {\bibfnamefont {A.}~\bibnamefont {Murua}},\ }\href {http://www.jstor.org/stable/42002709} {\bibfield  {journal} {\bibinfo  {journal} {Mathematics of Computation}\ }\textbf {\bibinfo {volume} {82}},\ \bibinfo {pages} {1559} (\bibinfo {year} {2013})}\BibitemShut {NoStop}%
\bibitem [{\citenamefont {Blanes}\ \emph {et~al.}(2024{\natexlab{a}})\citenamefont {Blanes}, \citenamefont {Casas},\ and\ \citenamefont {Murua}}]{Blanes_Casas_Murua_2024}%
  \BibitemOpen
  \bibfield  {author} {\bibinfo {author} {\bibfnamefont {S.}~\bibnamefont {Blanes}}, \bibinfo {author} {\bibfnamefont {F.}~\bibnamefont {Casas}},\ and\ \bibinfo {author} {\bibfnamefont {A.}~\bibnamefont {Murua}},\ }\href {https://doi.org/10.1017/S0962492923000077} {\bibfield  {journal} {\bibinfo  {journal} {Acta Numerica}\ }\textbf {\bibinfo {volume} {33}},\ \bibinfo {pages} {1–161} (\bibinfo {year} {2024}{\natexlab{a}})}\BibitemShut {NoStop}%
\bibitem [{\citenamefont {Lopez-Marcos}\ \emph {et~al.}(1996)\citenamefont {Lopez-Marcos}, \citenamefont {Skeel},\ and\ \citenamefont {Sanz-Serna}}]{Lopez96}%
  \BibitemOpen
  \bibfield  {author} {\bibinfo {author} {\bibfnamefont {M.~A.}\ \bibnamefont {Lopez-Marcos}}, \bibinfo {author} {\bibfnamefont {R.~D.}\ \bibnamefont {Skeel}},\ and\ \bibinfo {author} {\bibfnamefont {J.~M.}\ \bibnamefont {Sanz-Serna}},\ }in\ \href@noop {} {\emph {\bibinfo {booktitle} {Numerical Analysis}}}\ (\bibinfo  {publisher} {Longman Scientific and Technical},\ \bibinfo {year} {1996})\ pp.\ \bibinfo {pages} {107--122}\BibitemShut {NoStop}%
\bibitem [{\citenamefont {Butcher}\ and\ \citenamefont {Sanz-Serna}(1996)}]{butcher1996number}%
  \BibitemOpen
  \bibfield  {author} {\bibinfo {author} {\bibfnamefont {J.}~\bibnamefont {Butcher}}\ and\ \bibinfo {author} {\bibfnamefont {J.}~\bibnamefont {Sanz-Serna}},\ }\href {https://doi.org/10.1016/S0168-9274(96)00028-1} {\bibfield  {journal} {\bibinfo  {journal} {Applied Numerical Mathematics}\ }\textbf {\bibinfo {volume} {22}},\ \bibinfo {pages} {103} (\bibinfo {year} {1996})}\BibitemShut {NoStop}%
\bibitem [{\citenamefont {McLachlan}(1996)}]{mclachlan1996more}%
  \BibitemOpen
  \bibfield  {author} {\bibinfo {author} {\bibfnamefont {R.~I.}\ \bibnamefont {McLachlan}},\ }\href@noop {} {\bibfield  {journal} {\bibinfo  {journal} {Integration Algorithms and Classical Mechanics}\ }\textbf {\bibinfo {volume} {10}},\ \bibinfo {pages} {141} (\bibinfo {year} {1996})}\BibitemShut {NoStop}%
\bibitem [{\citenamefont {Wisdom}\ \emph {et~al.}(1996)\citenamefont {Wisdom}, \citenamefont {Holman},\ and\ \citenamefont {Touma}}]{wisdom1996symplectic}%
  \BibitemOpen
  \bibfield  {author} {\bibinfo {author} {\bibfnamefont {J.}~\bibnamefont {Wisdom}}, \bibinfo {author} {\bibfnamefont {M.}~\bibnamefont {Holman}},\ and\ \bibinfo {author} {\bibfnamefont {J.}~\bibnamefont {Touma}},\ }\href@noop {} {\bibfield  {journal} {\bibinfo  {journal} {Integration Algorithms and Classical Mechanics}\ }\textbf {\bibinfo {volume} {10}},\ \bibinfo {pages} {217} (\bibinfo {year} {1996})}\BibitemShut {NoStop}%
\bibitem [{\citenamefont {Blanes}\ \emph {et~al.}(1999)\citenamefont {Blanes}, \citenamefont {Casas},\ and\ \citenamefont {Ros}}]{blanes1999symplectic}%
  \BibitemOpen
  \bibfield  {author} {\bibinfo {author} {\bibfnamefont {S.}~\bibnamefont {Blanes}}, \bibinfo {author} {\bibfnamefont {F.}~\bibnamefont {Casas}},\ and\ \bibinfo {author} {\bibfnamefont {J.}~\bibnamefont {Ros}},\ }\href {https://doi.org/10.1137/S1064827598332497} {\bibfield  {journal} {\bibinfo  {journal} {SIAM Journal on Scientific Computing}\ }\textbf {\bibinfo {volume} {21}},\ \bibinfo {pages} {711} (\bibinfo {year} {1999})}\BibitemShut {NoStop}%
\bibitem [{\citenamefont {Blanes}(2001)}]{blanes2001high}%
  \BibitemOpen
  \bibfield  {author} {\bibinfo {author} {\bibfnamefont {S.}~\bibnamefont {Blanes}},\ }\href {https://doi.org/10.1016/S0168-9274(00)00044-1} {\bibfield  {journal} {\bibinfo  {journal} {Applied Numerical Nathematics}\ }\textbf {\bibinfo {volume} {37}},\ \bibinfo {pages} {289} (\bibinfo {year} {2001})}\BibitemShut {NoStop}%
\bibitem [{\citenamefont {Blanes}\ \emph {et~al.}(2024{\natexlab{b}})\citenamefont {Blanes}, \citenamefont {Casas},\ and\ \citenamefont {Escorihuela-Tomàs}}]{blanes2024families}%
  \BibitemOpen
  \bibfield  {author} {\bibinfo {author} {\bibfnamefont {S.}~\bibnamefont {Blanes}}, \bibinfo {author} {\bibfnamefont {F.}~\bibnamefont {Casas}},\ and\ \bibinfo {author} {\bibfnamefont {A.}~\bibnamefont {Escorihuela-Tomàs}},\ }\href {https://doi.org/https://doi.org/10.1016/j.apnum.2024.06.002} {\bibfield  {journal} {\bibinfo  {journal} {Applied Numerical Mathematics}\ }\textbf {\bibinfo {volume} {204}},\ \bibinfo {pages} {86} (\bibinfo {year} {2024}{\natexlab{b}})}\BibitemShut {NoStop}%
\bibitem [{\citenamefont {Blanes}\ and\ \citenamefont {Casas}(2004)}]{BLANES2004bch}%
  \BibitemOpen
  \bibfield  {author} {\bibinfo {author} {\bibfnamefont {S.}~\bibnamefont {Blanes}}\ and\ \bibinfo {author} {\bibfnamefont {F.}~\bibnamefont {Casas}},\ }\href {https://doi.org/https://doi.org/10.1016/j.laa.2003.09.010} {\bibfield  {journal} {\bibinfo  {journal} {Linear Algebra and its Applications}\ }\textbf {\bibinfo {volume} {378}},\ \bibinfo {pages} {135} (\bibinfo {year} {2004})}\BibitemShut {NoStop}%
\bibitem [{\citenamefont {Forest}\ and\ \citenamefont {Ruth}(1990)}]{FOREST1990105}%
  \BibitemOpen
  \bibfield  {author} {\bibinfo {author} {\bibfnamefont {E.}~\bibnamefont {Forest}}\ and\ \bibinfo {author} {\bibfnamefont {R.~D.}\ \bibnamefont {Ruth}},\ }\href {https://doi.org/https://doi.org/10.1016/0167-2789(90)90019-L} {\bibfield  {journal} {\bibinfo  {journal} {Physica D: Nonlinear Phenomena}\ }\textbf {\bibinfo {volume} {43}},\ \bibinfo {pages} {105} (\bibinfo {year} {1990})}\BibitemShut {NoStop}%
\bibitem [{\citenamefont {Campostrini}\ and\ \citenamefont {Rossi}(1990)}]{CAMPOSTRINI1990753}%
  \BibitemOpen
  \bibfield  {author} {\bibinfo {author} {\bibfnamefont {M.}~\bibnamefont {Campostrini}}\ and\ \bibinfo {author} {\bibfnamefont {P.}~\bibnamefont {Rossi}},\ }\href {https://doi.org/https://doi.org/10.1016/0550-3213(90)90081-N} {\bibfield  {journal} {\bibinfo  {journal} {Nuclear Physics B}\ }\textbf {\bibinfo {volume} {329}},\ \bibinfo {pages} {753} (\bibinfo {year} {1990})}\BibitemShut {NoStop}%
\bibitem [{\citenamefont {Forest}(2006)}]{Forest_2006}%
  \BibitemOpen
  \bibfield  {author} {\bibinfo {author} {\bibfnamefont {E.}~\bibnamefont {Forest}},\ }\href {https://doi.org/10.1088/0305-4470/39/19/S03} {\bibfield  {journal} {\bibinfo  {journal} {Journal of Physics A: Mathematical and General}\ }\textbf {\bibinfo {volume} {39}},\ \bibinfo {pages} {5321} (\bibinfo {year} {2006})}\BibitemShut {NoStop}%
\bibitem [{\citenamefont {Babbush}\ \emph {et~al.}(2018)\citenamefont {Babbush}, \citenamefont {Gidney}, \citenamefont {Berry}, \citenamefont {Wiebe}, \citenamefont {McClean}, \citenamefont {Paler}, \citenamefont {Fowler},\ and\ \citenamefont {Neven}}]{BabbushPRX18}%
  \BibitemOpen
  \bibfield  {author} {\bibinfo {author} {\bibfnamefont {R.}~\bibnamefont {Babbush}}, \bibinfo {author} {\bibfnamefont {C.}~\bibnamefont {Gidney}}, \bibinfo {author} {\bibfnamefont {D.~W.}\ \bibnamefont {Berry}}, \bibinfo {author} {\bibfnamefont {N.}~\bibnamefont {Wiebe}}, \bibinfo {author} {\bibfnamefont {J.}~\bibnamefont {McClean}}, \bibinfo {author} {\bibfnamefont {A.}~\bibnamefont {Paler}}, \bibinfo {author} {\bibfnamefont {A.}~\bibnamefont {Fowler}},\ and\ \bibinfo {author} {\bibfnamefont {H.}~\bibnamefont {Neven}},\ }\href {https://doi.org/10.1103/PhysRevX.8.041015} {\bibfield  {journal} {\bibinfo  {journal} {Physical Review X}\ }\textbf {\bibinfo {volume} {8}},\ \bibinfo {pages} {041015} (\bibinfo {year} {2018})}\BibitemShut {NoStop}%
\bibitem [{\citenamefont {Ostmeyer}(2023)}]{Ostmeyer_2023}%
  \BibitemOpen
  \bibfield  {author} {\bibinfo {author} {\bibfnamefont {J.}~\bibnamefont {Ostmeyer}},\ }\href {https://doi.org/10.1088/1751-8121/acde7a} {\bibfield  {journal} {\bibinfo  {journal} {Journal of Physics A: Mathematical and Theoretical}\ }\textbf {\bibinfo {volume} {56}},\ \bibinfo {pages} {285303} (\bibinfo {year} {2023})}\BibitemShut {NoStop}%
\bibitem [{\citenamefont {Blanes}\ and\ \citenamefont {Moan}(2002)}]{BLANES2002313}%
  \BibitemOpen
  \bibfield  {author} {\bibinfo {author} {\bibfnamefont {S.}~\bibnamefont {Blanes}}\ and\ \bibinfo {author} {\bibfnamefont {P.}~\bibnamefont {Moan}},\ }\href {https://doi.org/https://doi.org/10.1016/S0377-0427(01)00492-7} {\bibfield  {journal} {\bibinfo  {journal} {Journal of Computational and Applied Mathematics}\ }\textbf {\bibinfo {volume} {142}},\ \bibinfo {pages} {313} (\bibinfo {year} {2002})}\BibitemShut {NoStop}%
\bibitem [{\citenamefont {Alberdi}\ \emph {et~al.}(2019)\citenamefont {Alberdi}, \citenamefont {Antoñana}, \citenamefont {Makazaga},\ and\ \citenamefont {Murua}}]{alberdi2019algorithm}%
  \BibitemOpen
  \bibfield  {author} {\bibinfo {author} {\bibfnamefont {E.}~\bibnamefont {Alberdi}}, \bibinfo {author} {\bibfnamefont {M.}~\bibnamefont {Antoñana}}, \bibinfo {author} {\bibfnamefont {J.}~\bibnamefont {Makazaga}},\ and\ \bibinfo {author} {\bibfnamefont {A.}~\bibnamefont {Murua}},\ }\href {https://arxiv.org/abs/1909.07263} {\bibfield  {journal} {\bibinfo  {journal} {arXiv: 1909.07263}\ } (\bibinfo {year} {2019})}\BibitemShut {NoStop}%
\bibitem [{\citenamefont {Gidney}\ and\ \citenamefont {Ekerå}(2021)}]{Gidney2021}%
  \BibitemOpen
  \bibfield  {author} {\bibinfo {author} {\bibfnamefont {C.}~\bibnamefont {Gidney}}\ and\ \bibinfo {author} {\bibfnamefont {M.}~\bibnamefont {Ekerå}},\ }\href {https://doi.org/10.22331/q-2021-04-15-433} {\bibfield  {journal} {\bibinfo  {journal} {Quantum}\ }\textbf {\bibinfo {volume} {5}},\ \bibinfo {pages} {433} (\bibinfo {year} {2021})}\BibitemShut {NoStop}%
\bibitem [{\citenamefont {Low}\ \emph {et~al.}(2023)\citenamefont {Low}, \citenamefont {Su}, \citenamefont {Tong},\ and\ \citenamefont {Tran}}]{Low2022trotter}%
  \BibitemOpen
  \bibfield  {author} {\bibinfo {author} {\bibfnamefont {G.~H.}\ \bibnamefont {Low}}, \bibinfo {author} {\bibfnamefont {Y.}~\bibnamefont {Su}}, \bibinfo {author} {\bibfnamefont {Y.}~\bibnamefont {Tong}},\ and\ \bibinfo {author} {\bibfnamefont {M.~C.}\ \bibnamefont {Tran}},\ }\href {https://doi.org/10.1103/PRXQuantum.4.020323} {\bibfield  {journal} {\bibinfo  {journal} {PRX Quantum}\ }\textbf {\bibinfo {volume} {4}},\ \bibinfo {pages} {020323} (\bibinfo {year} {2023})}\BibitemShut {NoStop}%
\bibitem [{\citenamefont {Rubin}\ \emph {et~al.}(2024)\citenamefont {Rubin}, \citenamefont {Berry}, \citenamefont {Kononov}, \citenamefont {Malone}, \citenamefont {Khattar}, \citenamefont {White}, \citenamefont {Lee}, \citenamefont {Neven}, \citenamefont {Babbush},\ and\ \citenamefont {Baczewski}}]{rubin2023quantum}%
  \BibitemOpen
  \bibfield  {author} {\bibinfo {author} {\bibfnamefont {N.~C.}\ \bibnamefont {Rubin}}, \bibinfo {author} {\bibfnamefont {D.~W.}\ \bibnamefont {Berry}}, \bibinfo {author} {\bibfnamefont {A.}~\bibnamefont {Kononov}}, \bibinfo {author} {\bibfnamefont {F.~D.}\ \bibnamefont {Malone}}, \bibinfo {author} {\bibfnamefont {T.}~\bibnamefont {Khattar}}, \bibinfo {author} {\bibfnamefont {A.}~\bibnamefont {White}}, \bibinfo {author} {\bibfnamefont {J.}~\bibnamefont {Lee}}, \bibinfo {author} {\bibfnamefont {H.}~\bibnamefont {Neven}}, \bibinfo {author} {\bibfnamefont {R.}~\bibnamefont {Babbush}},\ and\ \bibinfo {author} {\bibfnamefont {A.~D.}\ \bibnamefont {Baczewski}},\ }\href {https://doi.org/10.1073/pnas.2317772121} {\bibfield  {journal} {\bibinfo  {journal} {Proceedings of the National Academy of Sciences}\ }\textbf {\bibinfo {volume} {121}},\ \bibinfo {pages} {e2317772121} (\bibinfo {year} {2024})}\BibitemShut {NoStop}%
\bibitem [{\citenamefont {Su}\ \emph {et~al.}(2021{\natexlab{b}})\citenamefont {Su}, \citenamefont {Berry}, \citenamefont {Wiebe}, \citenamefont {Rubin},\ and\ \citenamefont {Babbush}}]{SuPRXQuantum21}%
  \BibitemOpen
  \bibfield  {author} {\bibinfo {author} {\bibfnamefont {Y.}~\bibnamefont {Su}}, \bibinfo {author} {\bibfnamefont {D.~W.}\ \bibnamefont {Berry}}, \bibinfo {author} {\bibfnamefont {N.}~\bibnamefont {Wiebe}}, \bibinfo {author} {\bibfnamefont {N.}~\bibnamefont {Rubin}},\ and\ \bibinfo {author} {\bibfnamefont {R.}~\bibnamefont {Babbush}},\ }\href {https://doi.org/10.1103/PRXQuantum.2.040332} {\bibfield  {journal} {\bibinfo  {journal} {PRX Quantum}\ }\textbf {\bibinfo {volume} {2}},\ \bibinfo {pages} {040332} (\bibinfo {year} {2021}{\natexlab{b}})}\BibitemShut {NoStop}%
\bibitem [{\citenamefont {Cho}\ \emph {et~al.}(2024)\citenamefont {Cho}, \citenamefont {Berry},\ and\ \citenamefont {Hsieh}}]{Cho2024}%
  \BibitemOpen
  \bibfield  {author} {\bibinfo {author} {\bibfnamefont {C.-H.}\ \bibnamefont {Cho}}, \bibinfo {author} {\bibfnamefont {D.~W.}\ \bibnamefont {Berry}},\ and\ \bibinfo {author} {\bibfnamefont {M.-H.}\ \bibnamefont {Hsieh}},\ }\href {https://doi.org/10.1103/PhysRevA.109.062431} {\bibfield  {journal} {\bibinfo  {journal} {Physical Review A}\ }\textbf {\bibinfo {volume} {109}},\ \bibinfo {pages} {062431} (\bibinfo {year} {2024})}\BibitemShut {NoStop}%
\bibitem [{\citenamefont {{Van-Brunt}}\ and\ \citenamefont {{Visser}}(2016)}]{VanBrunt-BCH}%
  \BibitemOpen
  \bibfield  {author} {\bibinfo {author} {\bibfnamefont {A.}~\bibnamefont {{Van-Brunt}}}\ and\ \bibinfo {author} {\bibfnamefont {M.}~\bibnamefont {{Visser}}},\ }\href {https://doi.org/10.1063/1.4939929} {\bibfield  {journal} {\bibinfo  {journal} {Journal of Mathematical Physics}\ }\textbf {\bibinfo {volume} {57}},\ \bibinfo {eid} {023507} (\bibinfo {year} {2016})}\BibitemShut {NoStop}%
\bibitem [{\citenamefont {Duleba}\ and\ \citenamefont {Karcz-Duleba}(2020)}]{Duleba2019LieMonomials}%
  \BibitemOpen
  \bibfield  {author} {\bibinfo {author} {\bibfnamefont {I.}~\bibnamefont {Duleba}}\ and\ \bibinfo {author} {\bibfnamefont {I.}~\bibnamefont {Karcz-Duleba}},\ }in\ \href {https://doi.org/10.1007/978-3-030-45093-9_56} {\emph {\bibinfo {booktitle} {Computer Aided Systems Theory -- EUROCAST 2019}}},\ \bibinfo {editor} {edited by\ \bibinfo {editor} {\bibfnamefont {R.}~\bibnamefont {Moreno-D{\'i}az}}, \bibinfo {editor} {\bibfnamefont {F.}~\bibnamefont {Pichler}},\ and\ \bibinfo {editor} {\bibfnamefont {A.}~\bibnamefont {Quesada-Arencibia}}}\ (\bibinfo  {publisher} {Springer International Publishing},\ \bibinfo {address} {Cham},\ \bibinfo {year} {2020})\ pp.\ \bibinfo {pages} {465--473}\BibitemShut {NoStop}%
\end{thebibliography}

\appendix
\section{Background: Extending Yoshida's method to 10th order}\label{sec:yoshida-order10}

Here we give a pedagogical explanation of how to extend the method of Yoshida to obtain the equations for a 10th order integrator.
See also Ref.~\cite{Hairer2006} for an explanation of how to derive order conditions in terms of tree diagrams, or Ref.~\cite{Sofroniou2005integrators} for explicit 10th order conditions.
We also provide an introductory explanation of how the method is used for 6th order in \sec{yoshidas-method}.
The general principle used there was to provide an expression for $S^{(m)}(\tau)$ in Eq.~\eqref{eq:Yosh-order6} with the expression in the exponential given up to 6th order.
Then the coefficients of the multicommutators of operators needed to be made equal to zero in order to obtain a 6th order approximation.
Here we apply the same principle, except now we need to derive the terms up to 10th order.

In order to do this, we start by expressing $S_2(t)$ up to 10th order as
\begin{equation}
    S_2(t)=e^{t \alpha_1+t^3 \alpha_3+t^5 \alpha_5+t^7 \alpha_7+t^9 \alpha_9+\mathcal{O}(t^{11})}
\end{equation}
where we follow the notation from \cref{cor:symmetric_BCH} with $\alpha_j$ defined as commutators of operators.
We will then iteratively apply the symmetric BCH expansion for $Z$ in $e^Z = e^C e^D e^C$ in order to obtain an expression for $S^{(m)}(\tau)$ up to 10th order. 
It will be helpful to consider the following notation for commutators. For any operators $X_1, X_2,\cdots, X_L$ we define 
\begin{equation}
[X_1^{n_1},X_2^{n_1},\cdots,X_L^{n_L}]:=[\underbrace{X_1,\cdots,X_1}_\text{$n_1$ times },\underbrace{X_2,\cdots,X_2}_\text{$n_2$ times },\cdots,\underbrace{X_L,\cdots,X_L}_\text{$n_L$ times }],
\end{equation}
where the commutator on the right hand should be understood as in \cref{eq:nested_commutator}. To express the results
it will be useful to define the following commutators
\begin{align}
\beta_9 &= [\alpha_1,\alpha_1,\alpha_7], \\
\gamma^{(1)}_9&=[\alpha_1,\alpha_3,\alpha_5],\\
\gamma^{(2)}_9&=[\alpha_3,\alpha_1,\alpha_5],\\
\gamma^{(3)}_9&=[\alpha_5,\alpha_1,\alpha_3],\\
\delta^{(1)}_9& = [\alpha_1^4, \alpha_5],\\
\delta^{(2)}_9&=[\alpha_3,\alpha_1^3,\alpha_3], \\
\delta^{(3)}_9&=[\alpha_1,\alpha_3,\alpha_1^2,\alpha_3],\\
\epsilon_9 &= [\alpha_1^6,\alpha_3] .
\end{align}
We also use the notation of Yoshida \cite{Yoshida1990} for the commutators
\begin{align}
    \beta_5&=[\alpha_1,\alpha_1,\alpha_3] \equiv E_{5,2},\\
    \beta_7&=[\alpha_1,\alpha_1,\alpha_5] \equiv E_{7,3},\\
    \gamma_7&=[\alpha_3,\alpha_3,\alpha_1] \equiv -E_{7,2},\\
    \delta_7&=[\alpha_1,\alpha_1,\alpha_1,\alpha_1,\alpha_3] \equiv E_{7,4} ,
\end{align}
where we have indicated the equivalence to the multicommutators in the notation of Ref.~\cite{Blanes2006processing}.

Note that here we have only defined commutators of up to 7 of the $\alpha_j$ operators.
In the following we will be expanding expressions up to 9th order in $t$, because the symmetry of the formulae means that 10th order terms (and all even-order terms) must be zero.
The only way to obtain 9th order in $t$ with commutators of $t\alpha_1$, $t^3\alpha_3$, etc., is to have commutators of all $\alpha_1$ as $[\alpha_1,\alpha_1,\cdots,\alpha_1]$.
But that expression must be zero, because $[\alpha_1,\alpha_1]=0$.
Then when we express the expansion of $Z$ in $e^Z = e^C e^D e^C$, the first-order terms in $C$ and $D$ will be proportional to $\alpha_1$.
That means the only 9th order terms coming from commutators of 9 of $C,D$ would correspond to $[\alpha_1,\alpha_1,\cdots,\alpha_1]$ and therefore be zero.
This means we will only need commutators of up to 7 operators when expressing the expansion of $Z$ in $e^Z = e^C e^D e^C$ as well.

To obtain the coefficients multiplying the commutators with up to 7 operators in the symmetric BCH expansion for $e^Z = e^C e^D e^C$, we first use the algorithm defined in Section V of Ref.~\cite{VanBrunt-BCH}.
The algorithm in that work generates the scalar coefficients multiplying products of operators rather than their commutators, as we need here.
In order to derive the corresponding expressions in terms of commutators, we express the symmetric BCH expansion in the Ph.\ Hall basis, which is a basis for writing Lie monomials consisting of commutators of the generators of the Lie algebra.
For example, the elements up to 4th order in $C,D$ are the operators $C,D$ themselves, as well as
\begin{equation}
    [C,D], \quad [C,C,D], \quad [D,C,D], \quad [C,C,C,D], \quad [D,C,C,D], \quad [D,D,C,D].
\end{equation}
For a list of operators in this basis up to 7th order, see Table 1 in \cite{Duleba2019LieMonomials}.

We obtain the coefficients for the Ph.~Hall basis by solving the corresponding linear problem of changing from one basis to another.
As an example, consider the term with $3$ operators in the symmetric BCH expansion from \cref{cor:symmetric_BCH}, given as  $\alpha_3=\frac{1}{12}[Y,[Y,X]]-\frac{1}{24}[X,[X,Y]]$. We can also express the commutators as products by expanding out the commutators, which gives $\tilde{\alpha}_3=\frac{1}{24}\left( 2Y^2 X -4YXY - 2 XY^2 - X^2 Y + 2 XYX + Y X^2 \right)$. The algorithm in \cite{VanBrunt-BCH} outputs expressions with the commutators expanded out as in $\tilde{\alpha}_3$. In order to obtain the original expression $\alpha_3$, we write $\tilde{\alpha}_3=a[Y,[Y,X]]+b[X,[X,Y]]$ with $a,b\in\mathbb{R}$ and expand the commutators $[Y,[Y,X]]$ and $[X,[X,Y]]$. This gives several linear equations that can be written in terms of a matrix. By inverting this matrix, we obtain the coefficients $a$ and $b$.

By using this method, we obtain the symmetric BCH expansion for $e^Z = e^C e^D e^C$ up to 7th order as
\begin{align}\label{eq:symBCHo7}
    Z &= 2C + D + \frac{1}{6}([D,D,C]-[C,C,D]) \nn
    &\quad + \frac{7}{360}[C,C,C,C,D]-\frac{1}{360}[D,D,D,D,C]\nn
    &\quad + \frac{1}{90}[C,D,D,D,C]+\frac{1}{45}[D,C,C,C,D]\nn
    &\quad - \frac{1}{60}[C,C,D,D,C] + \frac{1}{30}[D,D,C,C,D]\nn
    &\quad -\frac{31}{15120}[C,C,C,C,C,C,D] - \frac{31}{5040} [D,C,C,C,C,C,D]\nn
    &\quad  - \frac{13}{1890} [D,D,C,C,C,C,D] - \frac{53}{15120} [D,D,D,C,C,C,D]\nn
    &\quad  -\frac{1}{1260} [D,D,D,D,C,C,D] - \frac{1}{15120} [D,D,D,D,D,C,D] + \mathcal{R}_{(9\leq)}.
\end{align}
Here $\mathcal{R}_{(9\leq)}$ is an infinite sum with commutators of an odd number of operators equal to or higher than $9$.
Now we prove the following Lemma which gives us the expansion up to 10th order of $S^{(m)}(\tau)$.
This will allow us to derive the equations for 10th order product formulae.
\begin{lemma}\label{lem:genera_Sm}
Using the definition for $S^{(m)}(\tau)$ as given
in \cref{eq:Yoshida_ansatz}, we have that for all $m\in\mathbb{N}$
\begin{align}\label{eq:genera_Sm}
    S^{(m)}(\tau)&=\exp\bigg\{\tau A_{1,m} \alpha_1 + \tau^3 A_{3,m}\alpha_3 + \tau^5( A_{5,m}\alpha_5 + B_{5,m} \beta_5 )\nonumber\\
    &\quad +\tau^7 (A_{7,m}\alpha_7 + B_{7,m} \beta_7 + C_{7,m}\gamma_7 + D_{7,m}\delta_7)\nonumber\\
    &\quad + \tau^9 (A_{9,m}\alpha_9 + B_{9,m}\beta_{9} + C^{(1)}_{9,m}\gamma^{(1)}_{9} + C^{(2)}_{9,m} \gamma^{(2)}_9 +  C^{(3)}_{9,m} \gamma^{(3)}_9\nonumber\\
    &\quad  + D^{(1)}_{9,m} \delta^{(1)}_9 + D^{(2)}_{9,m} \delta^{(2)}_9 + D^{(3)}_{9,m} \delta^{(3)}_9 + E_{9,m} \epsilon_9) + \order{\tau^{11}} \bigg\} ,
\end{align}
where the variables in upper case denote polynomials in the variables $(w_1,\cdots,w_m)$.
\end{lemma}
\begin{proof}
We proceed by induction. First, note that the statement is true for the case $m=0$,
\begin{align}
    S^{(m=0)}(\tau)&=S_2(w_0 \tau),\\
    &=\exp{t w_0 \alpha_1+t^3 w_0^{3} \alpha_3+t^5 w_0^{5} \alpha_5+t^7 w_0^{7} \alpha_7+t^9 w_0^{9}\alpha_9+\mathcal{O}(t^{11})}.
\end{align}
This clearly has the form of \cref{eq:genera_Sm} by taking $A_{j,m=0}=w_0^j$ and all other scalar variables as $0$.

Assume now that \cref{eq:genera_Sm} is correct, we want to derive an expression for $S^{(m+1)}$. We then have $S^{(m+1)}(\tau)= S_2(w_{m+1} \tau)S^{(m)}( \tau)S_2(w_{m+1} \tau)$ and thus
\begin{align}\label{eq:recursion-order10}
    S_2(w_{m+1} \tau)S^{(m)}( \tau)S_2(w_{m+1} \tau) &= \exp \bigg\{ \tau w_{m+1} \alpha_1 + \tau^3 w_{m+1}^3 \alpha_3 +\tau^5 w_{m+1}^5 \alpha_5+ \tau^7 w_{m+1}^7 \alpha_7 + \tau^9 w_{m+1}^9 \alpha_9 +\mathcal{O}(\tau^{11}) \bigg\} \nonumber\\
    &\quad  \times \exp\bigg\{\tau A_{1,m} \alpha_1 + \tau^3 A_{3,m}\alpha_3 + \tau^5( A_{5,m}\alpha_5 + B_{5,m} \beta_5 )\nonumber\\
    &\quad +\tau^7 (A_{7,m}\alpha_7 + B_{7,m} \beta_7 + C_{7,m}\gamma_7 + D_{7,m}\delta_7)\nonumber\\
    &\quad + \tau^9 (A_{9,m}\alpha_9 + B_{9,m}\beta_{9} + C^{(1)}_{9,m}\gamma^{(1)}_{9} + C^{(2)}_{9,m} \gamma^{(2)}_9 +  C^{(3)}_{9,m} \gamma^{(3)}_9\nonumber\\
    &\quad  + D^{(1)}_{9,m} \delta^{(1)}_9 + D^{(2)}_{9,m} \delta^{(2)}_9 + D^{(3)}_{9,m} \delta^{(3)}_9 + E_{9,m} \epsilon_9) + \order{\tau^{11}} \bigg\} \nonumber\\
    &\quad \times \exp \bigg\{ \tau w_{m+1} \alpha_1 + \tau^3 w_{m+1}^3 \alpha_3 +\tau^5 w_{m+1}^5 \alpha_5+ \tau^7 w_{m+1}^7 \alpha_7 + \tau^9 w_{m+1}^9 \alpha_9 +\order{\tau^{11}} \bigg\}.
\end{align}
We compute the right-hand-side (RHS) of \cref{eq:recursion-order10} applying the symmetric BCH formula from Corollary \ref{cor:symmetric_BCH}. Writing the RHS as $e^C e^D e^C$, we have that
\begin{align}
    C &=\tau w_{m+1} \alpha_1 + \tau^3 w_{m+1}^3 \alpha_3 +\tau^5 w_{m+1}^5 \alpha_5+ \tau^7 w_{m+1}^7 \alpha_7 + \tau^9 w_{m+1}^9 \alpha_9 + \mathcal{O}(\tau^{11}) \\
    D &=\tau A_{1,m} \alpha_1 + \tau^3 A_{3,m}\alpha_3 + \tau^5( A_{5,m}\alpha_5 + B_{5,m} \beta_5 )
    +\tau^7 (A_{7,m}\alpha_7 + B_{7,m} \beta_7 + C_{7,m}\gamma_7 + D_{7,m}\delta_7)\nn
    & \quad + \tau^9 (A_{9,m}\alpha_9 + B_{9,m}\beta_{9} + C^{(1)}_{9,m}\gamma^{(1)}_{9} + C^{(2)}_{9,m} \gamma^{(2)}_9 +  C^{(3)}_{9,m} \gamma^{(3)}_9
     + D^{(1)}_{9,m} \delta^{(1)}_9 + D^{(2)}_{9,m} \delta^{(2)}_9 + D^{(3)}_{9,m} \delta^{(3)}_9 + E_{9,m} \epsilon_9) \nn & \quad + \order{\tau^{11}} .
\end{align}
We then compute the commutators of $C$ and $D$ that appear in the symmetric BCH formula, here we give the resulting 9th order operators after applying the commutators.
When we write $[C,D,\cdots,C]_9$, the subscript indicates that we are only keeping the $9$th order terms when expanding the commutator. We will explain in detail how to compute the commutator $\mathcal{C}=[D,D,C]_9$, the other commutators are computed in a similar way. Since we only need to consider terms of 9th order when computing $\mathcal{C}$, each term will have contributions from each operator inside $\mathcal{C}$ (in this case two operators $D$ and one $C$) which is comprised of odd numbers that sum up to $9$ such that the commutator is non-zero.
We then have that
\begin{align}
    [D,D,C]_9 &=\tau^9 (A^2_{1,m}w_{m+1}^7 -A_{1,m}A_{7,m}w_{m+1})[\alpha_1,\alpha_1,\alpha_7]\nonumber\\
    &\quad + \tau^9 A_{1,m}B_{7,m}w_{m+1}[\alpha_1,\beta_7,\alpha_1]\nonumber\\
    &\quad + \tau^9 A_{1,m}C_{7,m}w_{m+1} [\alpha_1,\gamma_7,\alpha_1]\nonumber\\
    &\quad + \tau^9 A_{1,m}D_{7,m}w_{m+1} [\alpha_1,\delta_7,\alpha_1]\nonumber\\
    &\quad + \tau^9 (A_{1,m}A_{3,m}w_{m+1}-A_{1,m}A_{5,m}w_{m+1}^3)[\alpha_1,\alpha_3,\alpha_5]\nonumber\\
    &\quad + \tau^9 A_{1,m}B_{5,m}w_{m+1}^3[\alpha_1,\beta_5,\alpha_3]\nonumber\\
    &\quad + \tau^9 (A_{3,m}A_{1,m}w_{m+1}^5-A_{3,m}A_{5,m}w_{m+1})[\alpha_3,\alpha_1,\alpha_5]\nonumber\\
    &\quad + \tau^9 A_{1,m}B_{5,m}w_{m+1}[\alpha_3,\beta_5,\alpha_1]\nonumber\\
    &\quad + \tau^9 (A_{5,m}A_{1,m}w_{m+1}^3 - A_{5,m}A_{3,m}w_{m+1})[\alpha_5,\alpha_1,\alpha_3]\nonumber\\
    &\quad + \tau^9 B_{5,m}A_{3,m}w_{m+1}[\beta_5,\alpha_3,\alpha_1].
\end{align}
Given how we have defined the commutator, we have then
\begin{align}
    [D,D,C]_9 &= \tau^9 (A_{1,m}^2 w_{m+1}^7 - A_{1,m}A_{7,m} w_{m+1}) \beta_9-\tau^9 A_{1,m}B_{7,m}w_{m+1}\delta^{(1)}_9 \nonumber\\
    &\quad + \tau^9 A_{1,m}C_{7,m}w_{m+1}\delta^{(3)}_9 -\tau^9 A_{1,m}D_{7,m}w_{m+1} \epsilon_9 \nonumber\\
    &\quad + \tau^9 (A_{1,m}A_{3,m} w_{m+1}^5 - A_{1,m}A_{5,m} w_1^3) \gamma^{(1)}_9 - \tau^9 A_{1,m}B_{5,m}w_{m+1}^3 \delta^{(3)}_9 \nonumber\\
    &\quad + \tau^9(A_{3,m}A_{1,m}w_{m+1}^5 - A_{3,m}A_{5,m}w_{m+1})\gamma^{(2)}_9 - \tau^9 A_{3,m}B_{5,m}w_{m+1}\delta^{(2)}_9 \nonumber\\
    &\quad + \tau^9 (A_{5,m}A_{1,m}w_{m+1}^3 - A_{5,m}A_{3,m}w_{m+1}^2)\gamma^{(3)}_9 \nn
    & \quad + \tau^9 (B_{5,m}A_{1,m}w_{m+1}^3 - B_{5,m}A_{3,m}w_{m+1})(\delta^{(2)}_9 - \delta^{(3)}_9) \\
    [C,C,D]_9 &= \tau^9 (w_{m+1}^2 A_{7,m} - w_{m+1}^8 A_{1,m})\beta_9 + \tau^9 w_{m+1}^2 B_{7,m} \delta^{(1)}_9\nonumber\\
    &\quad -\tau^9 w_{m+1}^2 C_{7,m}\delta^{(3)}_9 + \tau^9 w_{m+1}^2 D_{7,m} \epsilon_9 \nonumber\\
    &\quad + \tau^9 (w_{m+1}^4 A_{5,m} - w_{m+1}^6 A_{3,m}) \gamma^{(1)}_9 + \tau^9 w_{m+1}^4 B_{5,m} \delta^{(3)}_9 \nonumber\\
    &\quad + \tau^9 (w_{m+1}^4 A_{5,m} - w_{m+1}^8 A_{1,m})\gamma^{(2)}_9 + \tau^9 w_{m+1}^4 B_{5,m} \delta^{(2)}_9 \nonumber\\
    &\quad + \tau^9 (w_{m+1}^6 A_{3,m} - w_{m+1}^8 A_{1,m})\gamma^{(3)}_9 \\
    [C,C,C,C,D]_9 &= \tau^9 (w_{m+1}^4 A_{5,m} - w_{m+1}^8 A_{1,m}) \delta^{(1)}_9 + \tau^9 A_{1,m}^3 B_{5,m}w_{m+1} \epsilon_9  \nonumber\\
    &\quad  + \tau^9 (w_{m+1}^6 A_{3,m} - w_{m+1}^8 A_{1,m}) \delta^{(2)}_9 \nonumber\\
    &\quad + \tau^9 2(w_{m+1}^6 A_{3,m} - w_{m+1}^8 A_{1,m})\delta^{(3)}_9 \\
    [D,D,D,D,C]_9 &= \tau^9 (A_{1,m}^4 w_{m+1}^5 - A_{1,m}^3 A_{5,m} w_{m+1})\delta^{(1)}_9 - \tau^9 A_{1,m}^3 B_{5,m} w_{m+1} \epsilon_9 \nonumber\\
    &\quad  + \tau^9 (A_{3,m}A_{1,m}^3 w_{m+1}^3 - A_{3,m}^2 A_{1,m}^2 w_{m+1})\delta^{(2)}_9 \nonumber \\
    &\quad + \tau^9 2(A_{1,m}^3 A_{3,m}w_{m+1}^3 - A_{1,m}^2 A_{3,m}^2 w_{m+1})\delta^{(3)}_9  \\
    [C,D,D,D,C]_9 &= \tau^9 (w_{m+1}^6 A_{1,m}^3 - w_{m+1}^2 A_{1,m}^2 A_{5,m})\delta^{(1)}_9 - \tau^9 A_{1,m}^2 B_{5,m}w_{m+1}^2 \epsilon_9\nonumber \\
    &\quad + \tau^9 (w_{m+1}^6A_{1,m}^3 - w_{m+1}^4 A_{1,m}^2 A_{3,m})\delta^{(2)}_9 \nonumber\\
    &\quad + \tau^9 2(w_{m+1}^4 A_{3,m}A_{1,m}^2 - w_{m+1}^2A_{3,m}^2 A_{1,m})\delta^{(3)}_9\\
    [D,C,C,C,D]_9 &= \tau^9 (A_{1,m}A_{5,m}w_{m+1}^3 - A_{1,m}^2 w_{m+1}^7)\delta^{(1)}_9 + \tau^9 w_{m+1}^3 A_{1,m}B_{5,m} \epsilon_9 \nonumber\\
    &\quad + \tau^9 (A_{3,m}^2 w_{m+1}^3 - A_{1,m} A_{3,m} w_{m+1}^5)\delta^{(2)} \nonumber\\
    &\quad + \tau^9 2(A_{1,m}A_{3,m}w_{m+1}^5 - A_{1,m}^2 w_{m+1}^7)\delta^{(3)}_9 \\
    [C,C,D,D,C]_9 &= \tau^9 (w_{m+1}^7 A_{1,m}^2 - w_{m+1}^3 A_{1,m}A_{5,m})\delta^{(1)}_9 - w_{m+1}^3 A_{1,m}B_{5,m} \epsilon_9 \nonumber\\
    &\quad + \tau^9 (w_{m+1}^7 A_{1,m}^2 - w_{m+1}^5 A_{1,m}A_{3,m})\delta^{(2)}_9 \nonumber\\
    &\quad +\tau^9 (w_{m+1}^5 A_{1,m} A_{3,m} - w_{m+1}^3 A_{3,m}^2)\delta^{(3)}_9 \\
    [D,D,C,C,D]_9 &= \tau^9 (A_{1,m}^2 A_{5,m}w_{m+1}^2 - A_{1,m}^3 w_{m+1}^6)\delta^{(1)}_9 + A_{1,m}^2 B_{5,m}w_{m+1}^2 \epsilon_9\nonumber\\
    &\quad + (A_{3,m}^2 A_{1,m}w_{m+1}^2 - A_{3,m}A_{1,m}^2 w_{m+1}^4)\delta^{(2)}_9\nonumber\\
    &\quad + \tau^9 (A_{3,m}^2 A_{1,m}w_{m+1}^2 - A_{3,m} A_{1,m}^2 w_{m+1}^4)\delta^{(3)}_9 + \tau^9 (A_{1,m}^2 A_{3,m}w_{m+1}^4 - A_{1,m}^3 w_{m+1}^6)\delta^{(3)}_9\\
    [C,C,C,C,C,C,D]_9 &= \tau^9 (w_{m+1}^6 A_{3,m} - w_{m+1}^8 A_{1,m})\epsilon_9 \nonumber\\
    [D,C,C,C,C,C,D]_9 &= \tau^9 (w_{m+1}^5 A_{1,m}A_{3,m} - w_{m+1}^7 A_{1,m}^2)\epsilon_9\\
    [D,D,C,C,C,C,D]_9 &= \tau^9 (A_{1,m}^2 A_{3,m}w_{m+1}^4 - A_{1,m}^3 w_{m+1}^6)\epsilon_9 \\
    [D,D,D,C,C,C,D]_9 &= \tau^9 (A_{1,m}^3A_{3,m}w_{m+1}^3 - A_{1,m}^4 w_{m+1}^5)\epsilon_9\\
    [D,D,D,D,C,C,D]_9 &= \tau^9 (A_{1,m}^4A_{3,m}w_{m+1}^2 - A_{1,m}^5 w_{m+1}^4)\epsilon_9\\
    [D,D,D,D,D,C,D]_9 &= \tau^9 (A_{1,m}^5A_{3,m}w_{m+1}-A_{1,m}^6w_{m+1}^3)\epsilon_9
\end{align}
Note that all the terms previously computed have terms that can be written as in \cref{eq:genera_Sm}, thus proving that $S^{(m+1)}$ can also be written in this way.
\end{proof}

Having proved \cref{lem:genera_Sm}, we can now compute the polynomials in \cref{eq:genera_Sm}. The polynomials are obtained from the recursion in \cref{eq:recursion-order10}, the left hand side corresponds to $S^{(m+1)}$ and can can be written as a single exponential, the same is true of the right side which is written as a single exponential. We have then the following polynomials:
\begin{align}
    A_{9,m+1}&=A_{9,m}+2w_{m+1}^9\label{eq:recursion_A9} \\
    B_{9,m+1}&=B_{9,m}+\frac{1}{6}(A_{1,m}^2w_{m+1}^7-A_{1,m}A_{7,m}w_{m+1})\nonumber\\
    &\quad -\frac{1}{6}(A_{7,m}w_{m+1}^2 - A_{1,m}w_{m+1}^8)\label{eq:recursion_B9}\\
    C_{9,m+1}^{(1)}&=C_{9,m}^{(1)} + \frac{1}{6}(A_{3,m}^2 A_{1,m} w_{m+1}^5 - A_{1,m}A_{5,m}w_{m+1}^3)\nonumber\\
    &\quad - \frac{1}{6}(A_{5,m}w_{m+1}^4-A_{3,m}w_{m+1}^6)\label{eq:recursion_C91}\\
        C_{9,m+1}^{(2)}&=C_{9,m}^{(2)} + \frac{1}{6}(A_{3,m}^2 A_{1,m} w_{m+1}^5 - A_{3,m}A_{5,m}w_{m+1})\nonumber\\
    &\quad - \frac{1}{6}(A_{5,m}w_{m+1}^4-A_{1,m}w_{m+1}^8)\label{eq:recursion_C92}\\
        C_{9,m+1}^{(3)}&=C_{9,m}^{(3)} + \frac{1}{6}(A_{5,m} A_{1,m} w_{m+1}^3 - A_{3,m}A_{5,m}w_{m+1})\nonumber\\
    &\quad - \frac{1}{6}(A_{3,m}w_{m+1}^6-A_{1,m}w_{m+1}^8)\label{eq:recursion_C93}\\
     D_{9,m+1}^{(1)}&=D_{9,m}^{(1)} - \frac{1}{6}(A_{1,m}B_{7,m}w_{m+1} + w_{m+1}^2 B_{7,m})  \nonumber\\
    &\quad +\frac{7}{360}(A_{5,m}w_{m+1}^4-w_{m+1}^8 A_{1,m})\nonumber\\
    &\quad -\frac{1}{360}(A_{1,m}^4 w_{m+1}^5 - A_{1,m}^3 A_{5,m} w_{m+1})\nonumber\\
    &\quad +\frac{1}{90}(A_{1,m}^3 w_{m+1}^6 - A_{1,m}^2A_{5,m}w_{m+1}^2)\nonumber\\
    &\quad +\frac{1}{45}(A_{1,m}A_{5,m}w_{m+1}^3-A_{1,m}^2 w_{m+1}^7)\nonumber\\
    &\quad -\frac{1}{60}(A_{1,m}^2w_{m+1}^7-A_{1,m} A_{5,m} w_{m+1}^3)\nonumber\\
    &\quad +\frac{1}{30}(A_{1,m}^2A_{5,m} w_{m+1}^2-A_{1,m}^3w_{m+1}^6)\label{eq:recursion_D91}\\
    D_{9,m+1}^{(2)}&=D_{9,m}^{(2)} - \frac{1}{6}(A_{3,m}B_{5,m}w_{m+1} + w_{m+1}^4 B_{5,m})  \nonumber\\
    &\quad +\frac{7}{360}(A_{3,m}w_{m+1}^6-w_{m+1}^8 A_{1,m})\nonumber\\
    &\quad -\frac{1}{360}(A_{1,m}^3A_{3,m} w_{m+1}^3 - A_{1,m}^2 A_{3,m}^2 w_{m+1})\nonumber\\
    &\quad +\frac{1}{90}(A_{1,m}^3 w_{m+1}^6 - A_{1,m}^2A_{3,m}w_{m+1}^4)\nonumber\\
    &\quad +\frac{1}{45}(A_{3,m}^2w_{m+1}^3-A_{1,m}A_{3,m} w_{m+1}^5)\nonumber\\
    &\quad -\frac{1}{60}(A_{1,m}^2w_{m+1}^7-A_{1,m} A_{3,m} w_{m+1}^5)\nonumber\\
    &\quad +\frac{1}{30}(A_{3,m}^2A_{1,m} w_{m+1}^2-A_{1,m}^2 A_{3,m} w_{m+1}^4)\label{eq:recursion_D92}\\
     D_{9,m+1}^{(3)}&=D_{9,m}^{(3)} - \frac{1}{6}(A_{1,m}B_{5,m}w_{m+1}^3 + w_{m+1}^4 B_{5,m})  \nonumber\\
    &\quad +\frac{14}{360}(A_{3,m}w_{m+1}^6-w_{m+1}^8 A_{1,m})\nonumber\\
    &\quad -\frac{2}{360}(A_{1,m}^3A_{3,m} w_{m+1}^3 - A_{1,m}^2 A_{3,m}^2 w_{m+1})\nonumber\\
    &\quad +\frac{2}{90}(A_{1,m}^2 A_{3,m} w_{m+1}^4 - A_{1,m}A_{3,m}^2w_{m+1}^2)\nonumber\\
    &\quad +\frac{2}{45}(A_{3,m}A_{1,m}w_{m+1}^5-A_{1,m}^2w_{m+1}^7)\nonumber\\
    &\quad -\frac{1}{60}(A_{1,m}^2w_{m+1}^7-A_{1,m} A_{3,m} w_{m+1}^5)\nonumber\\
    &\quad +\frac{1}{30}(A_{3,m}^2A_{1,m} w_{m+1}^2-A_{1,m}^2 A_{3,m} w_{m+1}^4)\nonumber\\
    &\quad + \frac{1}{6}(A_{1,m}C_{7,m}w_{m+1} + w_{m+1}^2 c_{7,m})\nonumber\\
    &\quad -\frac{1}{60}(w_{m+1}^5 A_{1,m}A_{3,m} - w_{m+1}^3 A_{3,m}^2)\nonumber\\
    &\quad +\frac{1}{30}(A_{1,m}^2 A_{3,m} w_{m+1}^4 - A_{1,m}^3 w_{m+1}^6)\nonumber\\
    &\quad -\frac{1}{6}(B_{5,m}A_{1,m}w_{m+1}^3 - B_{5,m}A_{3,m}w_{m+1})\label{eq:recursion_D93}\\
    E_{9,m+1} &= E_{9,m} - \frac{1}{6}(A_{1,m}D_{7,m} - w_{m+1}^2 D_{7,m})\nonumber\\
    &\quad +\frac{7}{360}w_{m+1}^4 B_{5,m}\nonumber\\
    &\quad +\frac{1}{360}A_{1,m}^3 B_{5,m} w_{m+1}\nonumber\\
    &\quad -\frac{1}{90}A_{1,m}^2 B_{5,m} w_{m+1}^2\nonumber\\
    &\quad +\frac{1}{45}A_{1,m} B_{5,m} w_{m+1}^3\nonumber\\
    &\quad +\frac{1}{60}A_{1,m}B_{5,m} w_{m+1}^3\nonumber\\
    &\quad +\frac{1}{30} A_{1,m}^2B_{5,m}w_{m+1}^2\nonumber\\
    &\quad -\frac{31}{15120}(w_{m+1}^6 A_{3,m}-w_{m+1}^8A_{1,m})\nonumber\\
    &\quad -\frac{31}{5040}(w_{m+1}^5A_{1,m}A_{3,m}-w_{m+1}^7A_{1,m}^2)\nonumber\\
    &\quad -\frac{13}{1890}(A_{1,m}^2A_{3,m}w_{m+1}^4 - A_{1,m}^3w_{m+1}^6)\nonumber\\
    &\quad -\frac{53}{15120}(A_{1,m}^3A_{3,m}w_{m+1}^3-A_{1,m}^4w_{m+1}^5)\nonumber\\
    &\quad -\frac{1}{1260}(A_{1,m}^4A_{3,m}w_{m+1}^2 - A_{1,m}^5w_{m+1}^4)\nonumber\\
    &\quad -\frac{1}{15120}(A_{1,m}^5A_{3,m}w_{m+1}-A_{1,m}^6w_{m+1}^3).
\end{align}
We obtain the polynomial equations for the tenth order product formula by imposing that $A_{1,m}=1$ and all other terms are equal to zero. 
Because $C^{(2)}_{9,m} = C^{(1)}_{9,m}+C^{(3)}_{9,m}$, one equation is eliminated, there are 15 equations to solve.

\section{Method for determining processors}\label{app:processors}
A processed formula is composed of two elements: a kernel, \(\kernel(t)\), and a processor, \(\processor(t)\). The effective order \(k\) captures up to which order in \(t\) the full product formula, including the processor, reproduces the the target dynamics, \(\processor(t) \kernel(t) \processor(t)^{-1} = e^{(X + Y)t} + \order{t^{k+1}}\).
In this work, we use processors that are constructed with the same procedure as Ref.~\cite{Blanes2006processing}. This type of processors are products of \(S_2(\tau w_i)\) arranged as \(Q^{(m)}(\tau) Q^{(m)}(-\tau)\), where \(Q^{(m)}(\tau) = S_2(\tau w_m) S_2(\tau w_{m-1}) \cdots S_2(\tau w_{0})\).
In the language of Ref.~\cite{Blanes2006processing}, \( Q^{(m)}(\pm \tau) \) is an element of the group of integrators \(\mathcal G_3\). The same reference gives a basis for the generating algebra, which we used to perform the following calculations. The basis can be found in Table 2 of Ref.~\cite{Blanes2006processing}, and we use the same naming scheme for the algebra elements.

We can obtain recursive formulas for \(Q^{(m)}(\tau) = S_2(\tau w_{m}) Q^{(m-1)}(\tau)\) via successive applications of the BCH formula (as opposed to the symmetric BCH formula in the case of the kernel). After taking the product \(Q^{(m)}(\tau) Q^{(m)}(-\tau)\), the processor simplifies  compared to \(Q^{(m)}\). In particular, the logarithm of \(\processor(\tau)\) is zero up to terms of order \(t^3\). 
We illustrate the procedure that gives an iterative expression for the processor up to order 8, omitting the calculations. First, we write the product \(S_2(\tau w_1) S_2(\tau w_0)\) in the \(E_{i,j}\)-basis. We have, for the logarithm of \(Q^{(1)}(\tau) = S_2(w_1 \tau) S_2(w_0 \tau)\),
\begin{align}
\log(Q^{(1)}(\tau)) &= 
  p_{1,1}^{(1)} Y_1 t 
 + p_{3,1}^{(1)} Y_3 t^3 
 + p_{4,1}^{(1)} E_{4,1} t^4 
 + \left( p_{5,1}^{(1)} E_{5,1} t^5 
 + p_{5,2}^{(1)} E_{5,2} \right)
 + \left( p_{6,1}^{(1)} E_{6,1}
 + p_{6,2}^{(1)} E_{6,2} \right) t^6
+ \order{t^7} \,,
\end{align}
where each coefficient is a function of \(w_0,w_1\). Using an inductive argument analogous to the derivation in the case of the symmetric BCH formula, this expression will give a recursive expression for \(Q^{(m)}(\tau) = S_2(\tau w_{m})Q^{(m-1)}(\tau)\), whose logarithm is an expansion in terms of \(p^{(m)}_{i,j}(w_0, \dots, w_m)\). The resulting expression is identical to the one above, with the superscripts \((1)\) replaced by \((m)\), and each coefficient is a polynomial of \(w_0,\dots, w_m\). In fact, it is not necessary to compute all the \(p\)-coefficients to find \(\processor^{(m)} (\tau) \). For example, only three of them (out of seven) are necessary to find the processor up to order \(6\), since 
\begin{align}
    \log(\processor^{(m)}(\tau)) =
    2 p^{(m)}_{4,1} E_{4,1} t^4 + p^{(m)}_{1,1} p^{(m)}_{4,1} E_{5,3} t^5
    + \frac{1}{3} \left( \left(p^{(m)}_{1,1}\right)^2 p^{(m)}_{4,1} E_{6,1}
    + 6 p^{(m)}_{6,2} E_{6,2} \right) t^6 + \order{t^7}
    \,.
\end{align}
Similar cancellations occur at any order. In our simulations, terms up to eighth order are required. The seventh-order term in the expansion is (omitting the dependency on \(m\))
\begin{align}
    \frac{1}{2}p_{3,1} p_{4,1} E_{7,2} +
    p_{1,1} p_{6,1} E_{7,3} +
    \left( \frac{1}{12}p_{1,1}^3 p_{4,1} + p_{1,1} p_{6,2}\right)E_{7,4}
    \,,
\end{align}
while the eighth-order term is
\begin{align}
\frac{1}{3} E_{8,4} p_{1,1}^2 p_{6,1} 
+ 2 E_{8,1} p_{8,1} 
+ 2 E_{8,2} p_{8,2} 
+ \frac{1}{2} E_{8,3} \left(p_{1,1} p_{3,1} p_{4,1} - p_{3,1} p_{5,2} + 4 p_{8,3}\right) &
\nonumber \\
+ 2 E_{8,4} p_{8,4} 
+ E_{8,5} \left(\frac{1}{60} P_{1,1}^4 p_{4,1} 
+ \frac{1}{3} p_{1,1}^2 p_{6,2} + 2 p_{8,5}\right)
\,. &
\end{align}
Note that not all the coefficients that appear in the expansion of \(Q^{(m)}(\tau)\) are necessary.

\section{Error constants for fermionic Hamiltonians with 4 orbitals}\label{sec:fermHam_d4}

In this Appendix we provide the average error constant $\omega$ for the eigenvalue error for the evolution of fermionic Hamiltonians in the case that the number of orbitals is $d=4$. We find that the $\omega$ obtained are close to those we obtain when $d=6$.
Note that our product formulae for 8th order still provide the best performance (albeit with similar performance between the worst of our three product formulae and the best from prior work).

\begin{table}[H]
\centering
\begin{tabular}{|c|c|c|c|c|c|} 
\hline
label & $M$ & processing & reference  & $\omega$   & $M\omega^{1/k}$  \\ 
\hline
PPBCM6m9 & $9$ & Y & $P_9 6$ Table 5 of \cite{Blanes2006processing} & $2.9\times 10^{-9}$ & $0.34$ \\
PPBCM6m6 & $13$ & Y & $P_{13}6$ in Table 6 of \cite{Blanes2006processing} & $1.4\times 10^{-9}$ & $0.43$ \\
BCE6m10 & $20$ & Y & $\psi_{10}^{[6]}$ Table 8 of \cite{blanes2024families} & $3.3\times 10^{-11}$ & $0.36$ \\
\hline
\end{tabular}
\caption{Comparison of constant factors $\omega$ for a selection of the lowest-error product formulae for 6th order. We generate $1000$ random Hamiltonians with $d=4$ orbitals as in \cref{eq:fermion-ham} and compute the average $\omega$.}
\label{tab:factor6-chemistry-eig4}
\end{table}

\begin{table}[H]
\centering
\begin{tabular}{|c|c|c|c|c|c|} 
\hline
label & $M$ & processing & reference  & $\omega$   & $M\omega^{1/k}$  \\ 
\hline
SS8s19 & 19 & N & Section 4.3 of \cite{Sofroniou2005integrators} &   $3.4\times 10^{-11}$                          &      $0.93 $                                  \\
PP8s13 & 13 & Y & $P_{13}8$ in Table 6 of \cite{Blanes2006processing} &   $4.2\times 10^{-10}$                           &      $0.87$                                  \\
Y8m10 & 21 & N & \cref{tab:order8-solm10} (our new result) &  $8.7\times 10^{-12}$                            &     $0.87$                                   \\
Y8m10b & 21 & N & \cref{tab:order8-solm10} (our new result) &    $1.5\times 10^{-12}$                          &    $0.68$                                \\
YP8m8 & 17 & Y & \cref{tab:process} (our new result) &   $2.3 \times 10^{-12}$                           &      $0.59$                                  \\
\hline
\end{tabular}
\caption{Comparison of constant factors $\omega$ for a selection of the lowest-error product formulae for 8th order. We generate $1000$ random Hamiltonians with $d=4$ orbitals as in \cref{eq:fermion-ham} and compute the average $\omega$.}
\label{tab:factor8-chemistry-eig4}
\end{table}

\begin{table}[H]
\centering
\begin{tabular}{|c|c|c|c|c|c|} 
\hline
label & $M$ & processing & reference  & $\omega$   & $M\omega^{1/k}$  \\ 
\hline
SS10s35 & $35$ & N & Section 4.4 of \cite{Sofroniou2005integrators} &   $3.0\times 10^{-15}$                          &      $1.24$                                  \\
PP10s23 & $23$ & Y & $P_{23}10$ in Table 6 of \cite{Blanes2006processing} &     $1.5\times 10^{-12}$                        &          1.51                              \\
Y10m17 & $35$ & N & \cref{tab:order10-solm17} &     $2.5\times 10^{-14}$                        &          1.53                              \\
\hline
\end{tabular}
\caption{Comparison of constant factors $\omega$ for a selection of the lowest-error product formulae for 10th order. We generate $1000$ random Hamiltonians with $d=4$ orbitals as in \cref{eq:fermion-ham} and compute the average $\omega$.}
\label{tab:factor10-chemistry-eig4}
\end{table}

\end{document}